\documentclass[twoside]{article}

\usepackage[subtle]{savetrees}
\usepackage{tabularx}

\usepackage[accepted]{aistats2025}
\usepackage{graphicx} 
\usepackage{graphicx, centernot, amsmath, bbm, booktabs, multirow, array, capt-of, varwidth, natbib, amstext, tabularx}

\usepackage[utf8]{inputenc}
\usepackage{dirtytalk}
\usepackage{csquotes,enumitem}
\usepackage{graphicx}
\usepackage{amsmath,amsthm,amsfonts}
\usepackage{lipsum}
\usepackage[utf8]{inputenc}
\usepackage{hyperref}
\usepackage{float}
\hypersetup{
    colorlinks=true,
    linkcolor=blue,
    filecolor=magenta,      
    urlcolor=cyan,
    pdftitle={Overleaf Example},
    pdfpagemode=FullScreen,
    }
\usepackage{tikz}
\usetikzlibrary{bayesnet}
\usepackage{caption}
\usepackage{subcaption}

\usepackage[utf8]{inputenc} 
\usepackage[T1]{fontenc}    
\usepackage{hyperref}       
\usepackage{url}            
\usepackage{booktabs}       
\usepackage{amsfonts}       
\usepackage{nicefrac}       
\usepackage{microtype}      
\usepackage{xcolor}         

\begin{document}

%

\newcommand{\lnote}{\textcolor[rgb]{1,0,0}{Lydia: }\textcolor[rgb]{0,0,1}}
\newcommand{\todo}{\textcolor[rgb]{1,0,0.5}{To do: }\textcolor[rgb]{0.5,0,1}}

\newcommand{\state}{S}
\newcommand{\meas}{M}
\newcommand{\out}{\mathrm{out}}
\newcommand{\piv}{\mathrm{piv}}
\newcommand{\pivotal}{\mathrm{pivotal}}
\newcommand{\isnot}{\mathrm{not}}
\newcommand{\pred}{^\mathrm{predict}}
\newcommand{\act}{^\mathrm{act}}
\newcommand{\pre}{^\mathrm{pre}}
\newcommand{\post}{^\mathrm{post}}
\newcommand{\calM}{\mathcal{M}}

\newcommand{\game}{\mathbf{V}}
\newcommand{\strategyspace}{S}
\newcommand{\payoff}[1]{V^{#1}}
\newcommand{\eff}[1]{E^{#1}}
\newcommand{\p}{\vect{p}}
\newcommand{\simplex}[1]{\Delta^{#1}}

\newcommand{\recdec}[1]{\bar{D}(\hat{Y}_{#1})}

\newcommand{\sphereone}{\calS^1}
\newcommand{\samplen}{S^n}
\newcommand{\wA}{w}
\newcommand{\Awa}{A_{\wA}}
\newcommand{\Ytil}{\widetilde{Y}}
\newcommand{\Xtil}{\widetilde{X}}
\newcommand{\wst}{w_*}
\newcommand{\wls}{\widehat{w}_{\mathrm{LS}}}
\newcommand{\dec}{^\mathrm{dec}}
\newcommand{\sub}{^\mathrm{sub}}

\newcommand{\calP}{\mathcal{P}}
\newcommand{\totspace}{\calZ}
\newcommand{\clspace}{\calX}
\newcommand{\attspace}{\calA}

\newcommand{\Ftil}{\widetilde{\calF}}

\newcommand{\totx}{Z}
\newcommand{\classx}{X}
\newcommand{\attx}{A}
\newcommand{\calL}{\mathcal{L}}

\newcommand{\defeq}{\mathrel{\mathop:}=}
\newcommand{\vect}[1]{\ensuremath{\mathbf{#1}}}
\newcommand{\mat}[1]{\ensuremath{\mathbf{#1}}}
\newcommand{\dd}{\mathrm{d}}
\newcommand{\grad}{\nabla}
\newcommand{\hess}{\nabla^2}
\newcommand{\argmin}{\mathop{\rm argmin}}
\newcommand{\argmax}{\mathop{\rm argmax}}
\newcommand{\Ind}[1]{\mathbf{1}\{#1\}}

\newcommand{\norm}[1]{\left\|{#1}\right\|}
\newcommand{\fnorm}[1]{\|{#1}\|_{\text{F}}}
\newcommand{\spnorm}[2]{\left\| {#1} \right\|_{\text{S}({#2})}}
\newcommand{\sigmin}{\sigma_{\min}}
\newcommand{\tr}{\text{tr}}
\renewcommand{\det}{\text{det}}
\newcommand{\rank}{\text{rank}}
\newcommand{\logdet}{\text{logdet}}
\newcommand{\trans}{^{\top}}
\newcommand{\poly}{\text{poly}}
\newcommand{\polylog}{\text{polylog}}
\newcommand{\st}{\text{s.t.~}}
\newcommand{\proj}{\mathcal{P}}
\newcommand{\projII}{\mathcal{P}_{\parallel}}
\newcommand{\projT}{\mathcal{P}_{\perp}}
\newcommand{\projX}{\mathcal{P}_{\mathcal{X}^\star}}
\newcommand{\inner}[1]{\langle #1 \rangle}

\renewcommand{\Pr}{\mathbb{P}}
\newcommand{\Z}{\mathbb{Z}}
\newcommand{\N}{\mathbb{N}}
\newcommand{\R}{\mathbb{R}}
\newcommand{\E}{\mathbb{E}}
\newcommand{\F}{\mathcal{F}}
\newcommand{\var}{\mathrm{var}}
\newcommand{\cov}{\mathrm{cov}}

\newcommand{\calN}{\mathcal{N}}

\newcommand{\jccomment}{\textcolor[rgb]{1,0,0}{C: }\textcolor[rgb]{1,0,1}}
\newcommand{\fracpar}[2]{\frac{\partial #1}{\partial  #2}}

\newcommand{\A}{\mathcal{A}}
\newcommand{\B}{\mat{B}}

\newcommand{\I}{\mat{I}}
\newcommand{\M}{\mat{M}}
\newcommand{\D}{\mat{D}}
\newcommand{\V}{\mat{V}}
\newcommand{\W}{\mat{W}}
\newcommand{\X}{\mat{X}}
\newcommand{\Y}{\mat{Y}}
\newcommand{\mSigma}{\mat{\Sigma}}
\newcommand{\mLambda}{\mat{\Lambda}}
\newcommand{\e}{\vect{e}}
\newcommand{\g}{\vect{g}}
\renewcommand{\u}{\vect{u}}
\newcommand{\w}{\vect{w}}
\newcommand{\x}{\vect{x}}
\newcommand{\y}{\vect{y}}
\newcommand{\z}{\vect{z}}
\newcommand{\fI}{\mathfrak{I}}
\newcommand{\fS}{\mathfrak{S}}
\newcommand{\fE}{\mathfrak{E}}
\newcommand{\fF}{\mathfrak{F}}

\newcommand{\Risk}{\mathcal{R}}

\renewcommand{\L}{\mathcal{L}}
\renewcommand{\H}{\mathcal{H}}

\newcommand{\cn}{\kappa}
\newcommand{\nn}{\nonumber}

\newcommand{\Hess}{\nabla^2}
\newcommand{\tlO}{\tilde{O}}
\newcommand{\tlOmega}{\tilde{\Omega}}

\newcommand{\calF}{\mathcal{F}}
\newcommand{\fhat}{\widehat{f}}
\newcommand{\calS}{\mathcal{S}}

\newcommand{\calX}{\mathcal{X}}
\newcommand{\calY}{\mathcal{Y}}
\newcommand{\calD}{\mathcal{D}}
\newcommand{\calZ}{\mathcal{Z}}
\newcommand{\calA}{\mathcal{A}}
\newcommand{\fbayes}{f^B}
\newcommand{\func}{f^U}

\newcommand{\bayscore}{\text{calibrated Bayes score}}
\newcommand{\bayrisk}{\text{calibrated Bayes risk}}

\newtheorem{example}{Example}[section]
\newtheorem{exc}{Exercise}[section]

\newtheorem{theorem}{Theorem}[section]
\newtheorem{definition}{Definition}
\newtheorem{proposition}[theorem]{Proposition}
\newtheorem{corollary}[theorem]{Corollary}

\newtheorem{remark}{Remark}[section]
\newtheorem{lemma}[theorem]{Lemma}
\newtheorem{claim}[theorem]{Claim}
\newtheorem{fact}[theorem]{Fact}
\newtheorem{assumption}{Assumption}

\newcommand{\iidsim}{\overset{\mathrm{i.i.d.}}{\sim}}
\newcommand{\unifsim}{\overset{\mathrm{unif}}{\sim}}
\newcommand{\sign}{\mathrm{sign}}
\newcommand{\wbar}{\overline{w}}
\newcommand{\what}{\widehat{w}}
\newcommand{\KL}{\mathrm{KL}}
\newcommand{\Bern}{\mathrm{Bernoulli}}
\newcommand{\ihat}{\widehat{i}}
\newcommand{\Dwst}{\calD^{w_*}}
\newcommand{\fls}{\widehat{f}_{n}}

\newcommand{\brpi}{\pi^{br}}
\newcommand{\brtheta}{\theta^{br}}

\twocolumn[

\runningtitle{Evaluating Prediction-based Interventions with Human Decision Makers In Mind}

\aistatstitle{\parbox{\linewidth}{\centering Evaluating Prediction-based Interventions \\ with Human Decision Makers In Mind}}

\aistatsauthor{Inioluwa Deborah Raji \And Lydia Liu}

\aistatsaddress{University of California, Berkeley \And  Princeton University } 
]

\begin{abstract}
Automated decision systems (ADS) are broadly deployed to inform and support human decision-making across a wide range of consequential settings. However, various context-specific details complicate the goal of establishing meaningful experimental evaluations for prediction-based interventions. Notably, current experiment designs rely on simplifying assumptions about human decision making in order to derive causal estimates. In reality, \emph{specific experimental design decisions} may induce cognitive biases in human decision makers, which could then significantly alter the observed effect sizes of the prediction intervention. In this paper, we formalize and investigate various models of human decision-making in the presence of a predictive model aid. We show that each of these behavioral models produces dependencies across decision subjects and results in the violation of existing assumptions, with consequences for treatment effect estimation. This work aims to further advance the scientific validity of intervention-based evaluation schemes for the assessment of ADS deployments.
\end{abstract}

\section{INTRODUCTION}
When making use of a predictive model in a decision-making process, human discretion often influences the final decisions -- hence, there is "not necessarily a one-to-one mapping from the predictive tool to the final outcome'' \citep{albright2019if}. 
A growing body of research hints at why -- 
while the field of machine learning is focused on the evaluation of predictive model outputs using accuracy metrics often tested on benchmark datasets \citep{raji2021ai}, these assessments do not reveal the full pattern of how humans \emph{interact} with the algorithm~\citep{green2019disparate}, and thus provides little assurance that the predictions will improve downstream outcomes on the impacted population through better decision-making 
\citep{liu2023reimagining, liu2024actionability}. ~\cite{kleinberg2018human} goes further to illustrate how, unlike prediction evaluation which focuses on minimizing model generalization error, decisions involve a more complex calculation of cost-benefit trade-offs, requiring considerations beyond just model performance. 

This difference between \emph{predictions} and \emph{decisions} - also sometimes referred to as \emph{judgements} ~\citep{agrawal2018prediction} - necessitates a critical expansion of scope in terms of what is being evaluated.
Once deployed, predictive models operate much more like typical policy interventions, not isolated prediction problems. These \emph{prediction-based interventions} are thus ideally evaluated around the causal question of how much the introduction of the predictive system impacts some important downstream outcome. This paradigm is about hypothesis testing a treatment - effectively thinking through the counterfactual measurement of what happens in the presence and in absence of the prediction-based intervention for a given decision-maker on a specific case. What we actually need to evaluate is not just the accuracy of the models' output predictions relative to outcomes, but the impact of such predictions on the actions of decision-makers and the appropriateness of those decisions.

One obvious approach to the estimation of these causal effects is through the implementation of \emph{experimental evaluations} for prediction-based interventions. Several efforts across criminal justice~\cite{B17imai2023experimental}, healthcare~\citep{B4plana2022randomized}, and education~\citep{B9dawson2017prediction} have begun conducting such evaluations to estimate the impact of introducing a prediction-based intervention to modify a given default policy. However, it is clear that those conducting these experiments not necessarily factoring in the particularities of the prediction-based intervention context into their experiment design. We argue that increased attentiveness to experiment design choices in this context is important -- as such choices may impact the \emph{responsiveness} (i.e., tendency to be swayed by an algorithmic decision aid) of the decision-makers relying on algorithmic recommendations, which in turn distorts our understanding of the average treatment effect of prediction-based interventions. 





Our contributions are thus as follows:

\begin{enumerate}[leftmargin=*]
    \item We examine the role of specific \emph{experiment design choices} on judge responsiveness -- specifically, choices that experiment designers make on (1) treatment assignment ($Z$), (2) the model positive prediction rate  ($P (\hat{Y} = 1)$) (by setting the threshold of the model corresponding  to a recommended action) and (3) the model correctness ($P (\hat{Y} = Y)$) (through model selection). We conclude that beyond intrinsic judge biases, \emph{experiment design choices can also impact judge responsiveness}. 
    \item We mathematically formalize how such experimental design choices impact judge responsiveness via novel models of judge decision making and prove that this leads to violations of SUTVA. We furthermore discuss how these design choices and their impact on judge responsiveness can lead to a mis-estimation of the average treatment effect in experimental evaluation settings for prediction-based interventions. 
    \item Using existing experiment data of a prediction-based intervention \citep{imai2020experimental}, we simulate a scenario of modifying experimental design choices and observing differences in judge decisions and estimated average treatment effect. 
\end{enumerate}

\section{RELATED WORK}

Here we provide a high level overview of background literature -- further discussion on related work is available in Appendix~\ref{app:further_rel}.

\paragraph{Modeling Discretion in Quasi-experimental investigations.} It is well known that decision-makers are subject to intuitive biases that inform how they respond to cases more generally ~\citep{guthrie2007blinking,guthrie2000inside}, and specifically how they factor in information from algorithmic decision aids~\citep{demichele2021intuitive}. However, most experimental and quasi-experimental work that have considered judge responsiveness assume that the tendency to follow or not follow the algorithmic recommendation is inherent to each judge -- a \emph{static} attribute of the judges' internal state, or their access to privileged information, i.e. factors outside of the experimenter's control \citep{mullainathan2022diagnosing,albright2019if}. For example, \citet{hoffman2018discretion} considered possible access to privileged information that impacted the hiring manager's judgement (each hiring manager having a particular \emph{exception rate} for making a different decision from the algorithm). 
In contrast, we consider a \emph{dynamic} model of judge responsiveness that is sensitive to factors in experimental design, acknowledging the connection between these experiment design factors and the induced behavioral biases that may impact judge decision-making.

\paragraph{The Rise of Experimental Evaluations.}
 In this work, we focus on \citet{imai2020experimental}, who presented a case study involving a Randomized control trial (RCT) examining the use of Public Safety Assessment (PSA) risk scores by judges in Dane County, Wisconsin. However, various other examples of RCTs for risk assessments have been explored in other domains. For instance, \citet{wilder2021clinical} ran a ``clinical trial'' on an algorithm used to select peer leaders in a public health program from homelessness youth. 
Similarly, several studies in education execute RCTs to assess models used to predict which students (in real-world and MOOC settings) are most likely to drop out~\citep{dawson2017from, borrella2019predict}.  
In healthcare, ~\citet{plana2022randomized} provides a recent survey of the 41 RCTs of AI/ML-based healthcare interventions done to date, critiquing a systematic lack of adherence to documentation standards such as CONSORT-AI and SPIRIT-AI ~\citep{liu2020reporting}. 
That being said, many of the RCTs cited suffer from the same design flaws and challenges we observe in ~\cite{imai2020experimental}, which we argue ignores the impact of experiment design choices on judge responsiveness.

\paragraph{Spillover effects.}
It is well-known, in the experimental economics literature \citep[][]{banerjee2009experimental,deaton2018understanding} for example, that insights from RCTs often fail to scale up when the intervention is applied broadly across the entire target population. This failure is often attributed to \emph{interaction} or \emph{interference} between decision subjects, and has also been referred to, in various contexts, as \emph{general equilibrium effects} or \emph{spillover effects}. In the potential outcomes model for causal inference, this is referred to as a violation of the Stable Unit Treatment Value Assumption (SUTVA) \citep{rubin1980discussion,rubin1990formal}. Recognizing the presence of interference, experimenters typically employ a two-level randomized experiment design \citep{duflo2003role,crepon2013labor} in order to estimate average treatment effects under different treatment proportions (e.g. proportion of job seekers receiving an intervention); treatment proportions are randomly assigned to the treatment clusters, e.g., the labor market. Although there are some exceptions in healthcare~\citep{gohil202113}, such experiment designs have not yet widely been applied to study human-algorithm decision making, where each treatment cluster is associated with a particular human decision-maker. 
%



\section{EXISTING EXPERIMENTAL PARADIGM AND ASSUMPTIONS}
\label{sec:existing}

We are interested in the problem of experimentally evaluating how algorithmic interventions in the form of predictive decision aids influence human decision-making. In this section, we describe the existing paradigm for experimental evaluations---we call this the \emph{case-independent} model---and the common assumptions required for causal identification. 

 In many existing experimental evaluations of algorithmic decision aids, the \emph{treatment unit} is presumed to be the decision subject and the \emph{treatment}, $Z_i$ is conceptualized as the provision of an predictive risk score $\hat{Y}_i$ to a decision maker who is responsible for making the final decision $D_i$. In the case of pre-trial detention decisions, the decision subject is the defendant in a particular court case, the decision maker is the judge presiding on the case, and the treatment $Z_i$ is a binary variable indicating if the PSA is shown to the judge for case $i$. $D_i$ represents the judge's binary detention decision, where $D_i = 1$ means detaining and $D_i = 0$ means releasing the arrestee before trial. $Y_i$ is the binary outcome variable, with $Y_i = 1$ indicating the arrestee committed an NCA, and $Y_i = 0$ indicating they did not. $X_i$ is a vector of observed pre-treatment covariates for case $i$, including age, gender, race, and prior criminal history. Independence is assumed between each case. We illustrate these variables and their dependencies (or lack thereof) in Figure~\ref{fig:indep_causal_model}.

 \begin{figure}[htbp]
  \centering
 \includegraphics[trim=1.5cm 2.5cm 3cm 10.5cm, clip, width=\columnwidth]{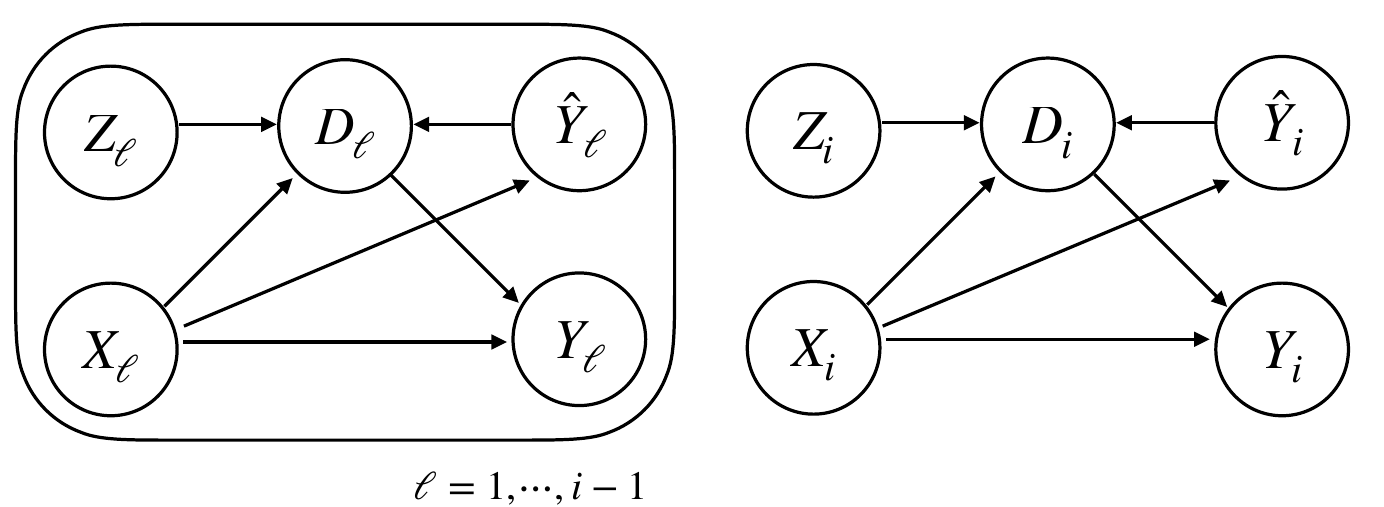}
  \caption{The case-independent model of human decision making with a predictive decision aid. This is the causal model assumed in prior work \citep[e.g.,][]{imai2020experimental}}
  \label{fig:indep_causal_model}
\end{figure}

The crucial assumption here is that of non-interference between decision subjects 
(see Assumption~\ref{assmp:sutva}, SUTVA). In words, the treatment status of one unit should not affect the decision for another unit. As a result, $Z_i$ is typically randomized at a single-level, i.e., on individual cases; and there is no randomization at the level of decision makers. 
In \citet{imai2020experimental} for example, where only one judge participated in the study, the algorithmic score is shown only for cases with an even case number (i.e., $Z_i =1$); for odd-numbered cases, no algorithmic recommendation is shown (i.e., $Z_i =0$). 

The assumption of non-interference across subjects is questionable, for multiple reasons. For one, in the psychology and behavioural economics literature, it is well-known that decision-makers have a bias towards consistency to their internal policies for decision-making and generally do not change on a case by case basis~\cite{tversky1985framing}. This means, in ~\cite{imai2020experimental} for example, where the judge only sees the risk score for a fraction of the cases, they might choose to ignore the PSA completely on cases where they do see it, or to infer patterns for cases where they don’t see the score, in order to be as consistent as possible in their decision-making (in this work, we mainly model effects such as the former). 
As such, experiment designs that assume the case-independent model does not allow us to draw clear conclusions on how a judge would modify their default decision-making in response to the algorithmic intervention implemented as a full-scale \emph{policy}, since we only observe the judge's decision-making under a \emph{partial} implementation of the would-be algorithmic intervention.

\section{MODEL}\label{sec:model}
In this section, we propose a novel causal model of algorithm-aided human decision making that is distinct from case-independent model, as depicted in Figure~\ref{fig:indep_causal_model}. Our model, depicted as a causal directed acyclic graph (Figure~\ref{fig:expanded_causal_model}), hypothesizes and accounts for dependence in the decisions across cases induced by the human decision maker's cognitive bias. We consider three types of cognitive bias that are particularly relevant to the current setting of human decision making with a predictive decision aid---all of which are directly influenced by experimental design choices (Table~\ref{table:cognitive_bias_models}). We do so by introducing two latent variables $J_{i,k}$ and $\epsilon_{i,k}$ that track the internal state of the decision maker $k$ and how it affects the decision for subject $i$.

\begin{figure}[htbp]
  \centering
  \includegraphics[trim=1.5cm 2cm 3cm 7cm, clip, width=\columnwidth]{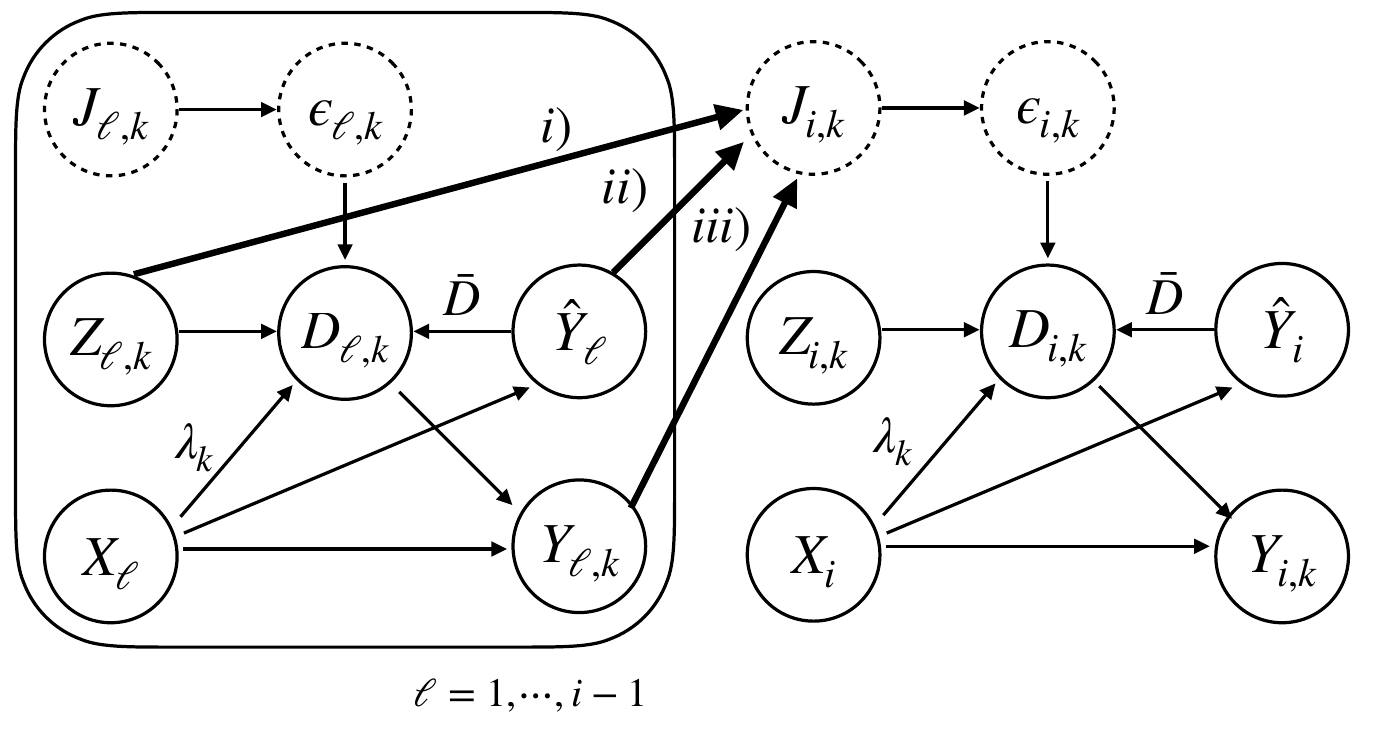}
  \caption{Proposed causal model that accounts for human decision maker bias. We describe three versions of this model (see Table~\ref{table:cognitive_bias_models}). Under the \emph{treatment exposure} model, arrow (i) is activated, but not (ii) and (iii). Under the \emph{capacity constraint} model, arrows (i) and (ii) are activated, but not (iii). Under the \emph{low trust} model, all three arrows, (i-iii), are activated.}
  \label{fig:expanded_causal_model}
\end{figure}

As in the case-independent model, we assume that $X_i$'s are independent and identically distributed.
In contrast to the case-independent model that does not include the role of the human decision maker (i.e., the judge), we explicitly model the judge---we index each case decision $D_{i,k}$ by both the decision maker (index $k$ in Figure~\ref{fig:expanded_causal_model}) and the decision subject (index $i$ or $\ell$ in Figure~\ref{fig:expanded_causal_model}), as opposed to only indexing it by the decision subject. Similarly, the treatment assignment variable $Z_{i,k}$ denotes the treatment status of both the decision subject $i$ and the decision maker $k$. This provides a fuller account of the space of possible treatment assignment counterfactuals, e.g., a case could have been assigned to a different judge leading to a different decision and outcome, all else held constant. 
We further assume that the judge is sequentially exposed to cases, where the judges consider case $i$ after the case $i-1$, and one case at a time.

To be precise, we define
\begin{equation*}
    Z_{i,k}
 = 
\begin{cases}
1 \text{ if unit } i \text{ is assigned to Judge }k \\ \text{ and unit } i \text{ received algorithmic treatment}\\
0 \text{ if unit } i \text{ is assigned to Judge }k \\ \text{ and unit } i \text{ received no treatment} \\
-1 \text{ if unit } i \text{ is not assigned to Judge }k 
  \end{cases}
\end{equation*}  

As the decision outcome $Y_{i,k}$ is downstream of $D_{i,k}$, we also index it by both the judge and the decision subject.

\subsection{Judge's decision}\label{sec:judge_dec}

We model the judge's decision for each treatment unit in the experiment as a random event, whose probability is determined by the individual judge's decision parameters as well as aspects of the experimental design (e.g. the treatment assignments).

For each judge $k$, we define a default decision function $\lambda_k$ that takes in observable characteristics $X_i \in \calX$ and operates independently of any algorithmic decision aid.\footnote{To illustrate the model in Figure~\ref{fig:expanded_causal_model}, we suggest that $\lambda_k$ only takes $X_i$ as input. For full generality, the Judge's default decision can be expressed as $\lambda(X_i, U_i)$, where $U_i$ is an additional \emph{exogenous} variable. Here $U_i$ could represent information available to the judge that is not available to the decision aid algorithm, or other exogenous noise. All our subsequent results also apply to this more general formulation, but to streamline our notation, we will omit $U_i$ and subsume it within $\lambda_k$. In other words, $\lambda_k(X_i)$ need not be deterministic given $X_i$.} When no predictive decision aid is provided to the unit~$i$ assigned to judge $k$ (i.e., $Z_{i,k} = 0$), the judge's decision $D_{i,k}$ takes value $\lambda_k(X_i)$.

Suppose instead that the unit $i$ is treated (i.e., $Z_{i,k} = 1$). In this case, the judge receives the algorithmic score $\hat{Y}_i$. We denote the recommended decision function as $\bar{D}(y)$ which maps a prediction $y$ to a decision. This function can represent existing decision guidelines used by judges for PSA scores~\citep{stith1998fear}. Note that prediction (lower-case) $y$ in $\bar{D}(y)$ is a variable that can take on any value $\hat{Y}_i$. This is to distinguish $\bar{D}(y)$ (the recommended decision, which is a function of the prediction), and $D(z)$ (the actual decision, which is a function of the judge being treated or not). 


In the case that $\recdec{i}$ differs from the judge's default decision $\lambda_k(X_i)$, the judge may choose to `comply' with the $\recdec{i}$, or to disagree with the ADS recommendation and choose $\lambda_k(X_i)$. We model this as a random event as follows. Let $J_{i,k}$ denote the \emph{responsiveness} of Judge $k$ for treatment unit~$i$, i.e. the probability that the Judge follows the PSA recommendation for case $i$. In other words, $J_{i,k}$ models the automation bias of the judge at the point they are deciding on unit $i$. Let $\epsilon_{i,k} \sim Ber(J_{i,k})$ be the random variable denoting judge response, drawn independently for each $i$; when $\epsilon_{i,k} = 1$, the judge  `complies' with the ADS recommendation.  


To summarize, we have the following decision for judge $k$ on unit $i$:
\begin{equation}\label{eq:model_of_judge}
    D_{i,k} = \begin{cases}
        \lambda_k(X_i) \text{ if } Z_{i,k} = 0 \text{ (no PSA) } \\
        \begin{cases}
     \recdec{i} &\text{if } \recdec{i} = \lambda_k(X_i)  \text{ or } \epsilon_{i,k}=1 \\
        \lambda_k(X_i) &\text{if }\recdec{i} \ne \lambda_k(X_i) \text{ and } \epsilon_{i,k}=0 
        \end{cases} \text{o.w.}
    \end{cases}.
\end{equation}


Note that in this case $J_{i,k}$ denotes the probability that a judge decides to follow $\recdec{i}$, in the case that $\lambda_k(X_i) \neq \recdec{i}$. Hence the judge responsiveness $J_{i,k}$ is a lower bound for actual frequency that the judge is makes the same decision as the algorithmic recommendation.


Drawing upon cognitive psychology, behaviorial economics, and human computer interaction literatures\footnote{It has long been understood in psychology and behavioral economics that the ``framing’’ or context in which nudges or advice is given, meaningful impacts how much this additional information is heeded during decision-making (see \citep{1kahneman2011thinking, 4klein2017sources, 8tversky1989rational}). This discussion has been further elaborated on in computer science contexts, through the formalization of broader ``algorithm-in-the-loop’’ considerations ~\citep{10green2019principles, 9guo2024decision}, and notions of ``automation bias’’ (ie. the degree of reliance of a human decision-maker on an algorithmic recommendation) ~\citep{14goddard2012automation, 15albright2019if}.}, we propose that the following three \textbf{responsiveness factors} are likely to impact judge response: 1) Treatment exposure, 2) capacity constraint, and 3) low trust. We model each of these factors independently for ease of exposition but they may impact the judge's response at the same time. These factors are summarized in Table~\ref{table:cognitive_bias_models}.

\begin{table*}[h!]
\centering
\begin{tabular}{| m{2cm} | m{6cm} | m{4cm} |}
  \hline
  \textbf{Model} & \textbf{Effect on Total Responsiveness $J_{i,k}$} & \textbf{Experiment design choices} \\
  \hline
  Treatment exposure & $J_{i,k}$ increases as average exposure to predictive decision aid increases & $Z_{i,k}$ \\
  \hline
  Capacity constraint & $J_{i,k}$ decreases as average exposure to `high risk' prediction rate increases. & $Z_{i,k}, \Pr(\hat{Y}_{i} > 0) $ \\
  \hline
  Low trust & $J_{i,k}$ decreases as average exposure to prediction error increases. & $Z_{i,k},  \Pr(\hat{Y_i} = Y_i)$ \\
  \hline
\end{tabular}
\caption{Description of Three Models of Cognitive Bias in Judge Responsiveness}
\label{table:cognitive_bias_models}
\end{table*}

\paragraph{Treatment exposure model.} The judge becomes responsive to the algorithmic recommendation if she encounters the algorithmic recommendation for more subjects (i.e., greater fraction of subjects are treated).
\begin{equation}\label{eq:treatment_exposure}
    J_{i+1,k} = 
b_k + f\left(\frac{1}{i} \sum_{m=1}^{i} Z_{m,k}\right)
\end{equation}
$b_k$ is the baseline responsiveness and $f\left(\frac{1}{n} \sum_{m=1}^{i} Z_{m,k}\right)$ is the adjustment to the responsiveness based on number of treated decision subjects seen so far. 

Humans demonstrate a self-consistency principle ~\citep{luu2018post, 3beresford2008understanding}, particularly in high-pressure decision-making contexts ~\citep{4klein2017sources}. This bias may lead to reliance on consistently available information rather than adjusting decision-making protocols for each case. Consequently, judges may respond more strongly to algorithmic recommendations when exposed to them frequently, while showing less responsiveness when exposure is limited ~\citep{2tversky1981framing, 8tversky1989rational}.

To illustrate, we present a simple thresholding model for $J_{i,k}$ based on the average number of treated decision subjects:
\begin{equation}
J_{i+1,k}=
    \begin{cases}
        1 & \text{if } \sum_{m=1}^{i} Z_{m,k} > i\tau\\
        0 & \text{o.w.}
    \end{cases},
\end{equation}
where $\tau \in (0,1)$. This instance of the judge's decision making model suggests that the judge becomes always responsive to the algorithmic recommendation as long as the cumulative average treatment frequency is above $\tau$.

\paragraph{Capacity constraint model} The judge has a limited capacity to respond to positive (or `high risk') predictions. They reduce their responsiveness as the rate of positive predictions they see from the algorithm increases. 
\begin{equation}\label{eq:capacity}
    J_{i+1,k} = b_k - f\left(\frac{1}{i}\sum_{m=1}^iZ_{m,k} \cdot \Ind{\hat{Y}_m > 0}\right)
\end{equation}
In high-stakes environments, decision-makers may become less responsive after reaching a specific capacity threshold for responding to certain types of algorithmic recommendations.\footnote{We describe simple thresholding based instantiations of this model as well as the low trust model in Appendix~\ref{app:instance}.} This phenomenon has been noted by decision theory scholars ~\citep{4klein2017sources} and economists ~\citep{5boutilier2024randomized} alike.\footnote{The empirical evidence regarding its implications in resource-constrained settings is mixed. Some computer scientists have found that, in such cases, decision-makers might over-rely on algorithmic recommendations instead of engaging in more complex, independent decision-making ~\citep{13buccinca2021trust, 14goddard2012automation}. This topic remains an area of ongoing research and discussion.}

This model of responsiveness doesn't make any assumption about whether the human decision maker is observing realized outcomes~$Y_{i,k}$, only the predicted outcomes. Hence this is a different model of cognitive bias than the following model, which addresses low accuracy directly.

\paragraph{Low trust model} The judge observes realized outcomes and knows when the algorithm has made a mistake (i.e., predicted $\hat{Y} \ne$ true outcome $Y$). They reduce their responsiveness as the error rate of the predictions they see from the algorithm increases. 
\begin{equation}\label{eq:low_trust}
    J_{i+1,k} = b_k - f\left(\frac{1}{i}\sum_{m=1}^iZ_{m,k} \cdot \Ind{\hat{Y}_m \neq Y_{m,k}}\right)
\end{equation}
This phenomenon exemplifies ``alert fatigue'' in algorithmic decision systems (ADS) with high false positive rates ~(\citet{wong2021external} found ``even a small number of bad [sepsis] predictions can cause alert fatigue'' in clinicians). Research in human-computer interaction (HCI) has shown that trust in AI recommendations is influenced by both stated and observed model accuracy \citep{6yin2019understanding}, with lower perceived accuracy reducing responsiveness to algorithmic recommendations. This finding has been empirically supported in subsequent research ~\citep{7lai2023towards, 12zhang2020effect}. Note that as written \eqref{eq:low_trust} cannot be applied to settings where the judge does not observe the realized outcome before the next decision, or if decisions affect the observability of baseline outcomes, e.g., $Y_{i,k}(D_{i,k} = 0)$ is not observed when $D_{i,k} = 1$ in the pre-trial risk assessment setting. 

\section{RESULTS}
\label{sec:results}

In this section, we describe the implications that our model of human decision making has for causal experiments and the estimation of treatment effects. First, we give simple conditions under which the common assumption for treatment effect estimation, Stable Unit Treatment Value Assumption (SUTVA), is violated. Second, we show that the underestimation of treatment effect that occur due to the chosen randomization in the assignment of cases to judges, even when the treated cases are the same, focusing on the treatment exposure model.

\subsection{Violation of SUTVA}

Our first observation is that the judge's changing responsiveness induces interference between decision subjects. We will formally show that this results in the violation of SUTVA, stated as follows. 
\begin{assumption}[SUTVA \citep{rubin1990formal,angrist1996identification}]\label{assmp:sutva}
A set of treatment assignments, decisions and outcomes $(\vect{Z}, \vect{Y}, \vect{D})$ is said to satisfy SUTVA if both the following conditions hold.
\begin{enumerate}
    \item[a.] If $Z_{i,k} = Z_{i,k}'$, then $D_{i,k}(\vect{Z}) = D_{i,k}(\vect{Z}')$.
    \item[b.] If $Z_{i,k} = Z_{i,k}'$ and $D_{i,k} = D_{i,k}'$, then $Y_{i,k}(\vect{Z}, \vect{D}) = Y_{i,k}(\vect{Z}',\vect{D}')$.
\end{enumerate}
\end{assumption}

In the following result, we show that under simple conditions on $J_{i,k}$, Assumption~\ref{assmp:sutva} (part a) is violated. The proof hinges on the lack of conditional independence between $D_{i,k}$ and prior treatment assignments $(Z_{1,k}, \cdots, Z_{i-1,k})$, given current treatment $Z_{i,k}$. This may be intuitively apparent from the causal DAG in Figure~\ref{fig:expanded_causal_model}, where one can see that $D_{i,k}$ is not d-separated from $(Z_{1,k}, \cdots, Z_{i-1,k})$ by $Z_{i,k}$. We include the proof in Appendix~\ref{app:proofs}.

\begin{theorem}[Violation of SUTVA]\label{thm:sutva_violation}
    Fix $k>0$ and consider some $i > 1$. Assume the judge's decision model is as described in equation~\eqref{eq:model_of_judge}, where $J_{i,k}$ is a monotonically non-decreasing (or non-increasing) function of $Z_{1,k}, \cdots, Z_{i-1,k}$, and strictly increasing (resp. decreasing) in at least one of its arguments. Assume that the judge's default decision function $\lambda_k$ is such that $\Pr(\lambda_k(X_i) \neq \recdec{i}) > 0$.
    
    Then $D_{i,k}$ is not conditionally independent of $(Z_{1,k}, \cdots, Z_{i-1,k})$, given $Z_{i,k}$. In particular, there exists treatment assignments $\vect{Z}, \vect{Z}'$ such that $Z_{i,k} = Z_{i,k}'$ and 
    \begin{equation*}
        \Pr\left(D_{i,k}(\vect{Z}) = D_{i,k}(\vect{Z}')\right) < 1.
    \end{equation*}
\end{theorem}

Theorem~\ref{thm:sutva_violation} suggests that for the specific models of $J_{i,k}$ that we introduced in Section~\ref{sec:judge_dec}, the SUTVA is indeed violated. We state this formally in the following corollary.

\begin{corollary}
Suppose $J_{i,k}$ is as defined in \eqref{eq:treatment_exposure}, \eqref{eq:capacity}, or \eqref{eq:low_trust}.
For any $\lambda_k$, $X_{i}$, $Y_{i}$, there exists $\hat{Y}(\cdot)$, $f$ and $b_k$ such that Assumption~\ref{assmp:sutva} (part a) does not hold with some positive probability. 

\end{corollary}
\begin{proof}
   Consider linear $f$ (i.e. $f(x) = ax$), $b_k = 0$ and $\hat{Y}$ such that $\hat{Y}_\ell > 0$ and $\hat{Y}_\ell \ne Y_{\ell,k}$ for some $\ell < i$. Then $J_{i,k}$ as defined in \eqref{eq:treatment_exposure}, \eqref{eq:capacity}, or \eqref{eq:low_trust} satisfies the assumptions of Theorem~\ref{thm:sutva_violation}.
\end{proof}

We note that under our causal model, the second part of the SUTVA (Assumption~\ref{assmp:sutva}, part b) is actually not violated. This is because $Y_{i,k}$ is conditionally independent of $(Z_{1,k}, \cdots, Z_{i-1,k})$ and  $(D_{1,k}, \cdots, D_{i-1,k})$ given $D_{i,k}$, as seen from Figure~\ref{fig:expanded_causal_model}.

\citet{imai2020experimental} argued that there was no statistically significant spillover effect in their experiment, hence supporting the assumption of SUTVA. However, the question of whether a conditional randmization test such as a permutation test can detect the violation SUTVA is intricate. Permutation tests typically reorder data to assess the null hypothesis, but in this case, the order does not affect the judge's average exposure to treatment, $\E[Z_{1,k}]$. Indeed, $J_{i,k}$ is a function of an average over i.i.d. variables as defined in \eqref{eq:treatment_exposure}, \eqref{eq:capacity}, or \eqref{eq:low_trust}, and it becomes effectively constant for all cases after a certain case count $i > T$, given the law of large numbers. This convergence in $J_{i,k}$ implies that reordering would not alter the joint distribution of $D_{i,k}$'s for $i$ large enough. The conditional randomness test used by \citet{imai2020experimental}, for example, maintains the average treatment proportion, randomizing only the spillover effects that are contingent on case order, and hence cannot effectively detect interference between decision subjects that is induced by $J_{i,k}$.



\subsection{Estimation of causal effects}
\label{sec:causal_effects}

Having shown that SUTVA is violated under a model of judge bias, we now investigate the implications on the estimation of causal effects. Our goal in this section is to illustrate the underestimation of the treatment effect of predictive decision aids on the human decision that may occur, if interference due to judge bias is not taken into account in the experiment's randomization scheme. We discuss a worked example focusing on the \emph{treatment exposure} and \emph{capacity constraint} model, leaving the investigation of the \emph{low trust} model to the appendix.

Consider two different treatment assignments (to $n$ total decision subjects) that have the same treated and untreated decision subjects, but assign these treated decision subjects to different decision makers (and we assume each case is only seen by one judge in each treatment assignment). For simplicity, we consider two judges, $k=1,2$, who are assumed to be identical, apart from their assigned cases. Specifically, both of them have decision making model as described in equation~\eqref{eq:treatment_exposure}, where $f(z) = az$ is a linear function and $b_1 = b_2 = b$.

\begin{enumerate}
    \item (Uniform randomization) $\vect{Z} =(\vect{Z}_{\cdot, 1},\vect{Z}_{\cdot, 2})$ such that Judge 1 receives 50\% of total cases, 50\% of them are treated and 50\% are untreated. Judge 2 receives other 50\% of total cases, 50\% of them are treated and 50\% are untreated.
    \item (Two-level randomization) $\vect{Z}' = (\vect{Z'}_{\cdot, 1},\vect{Z'}_{\cdot, 2})$ such that Judge 1 receives 50\% of total cases, 100\% of them are treated. Judge 2 receives other 50\% of total cases, 100\% are untreated.
\end{enumerate}
The above two treatment assignments each represent a different randomization scheme: 1) single-level randomization by decision subjects (cases) only, or 2) two-level randomization by decision makers and decision subjects. This is illustrated in Figure~\ref{fig:comparison_randomization_level}.

\begin{figure*}[ht]
	\centering	
 \begin{subfigure}{0.5\textwidth}
  \centering
  \captionsetup{justification=centering}
\includegraphics[width=\linewidth]{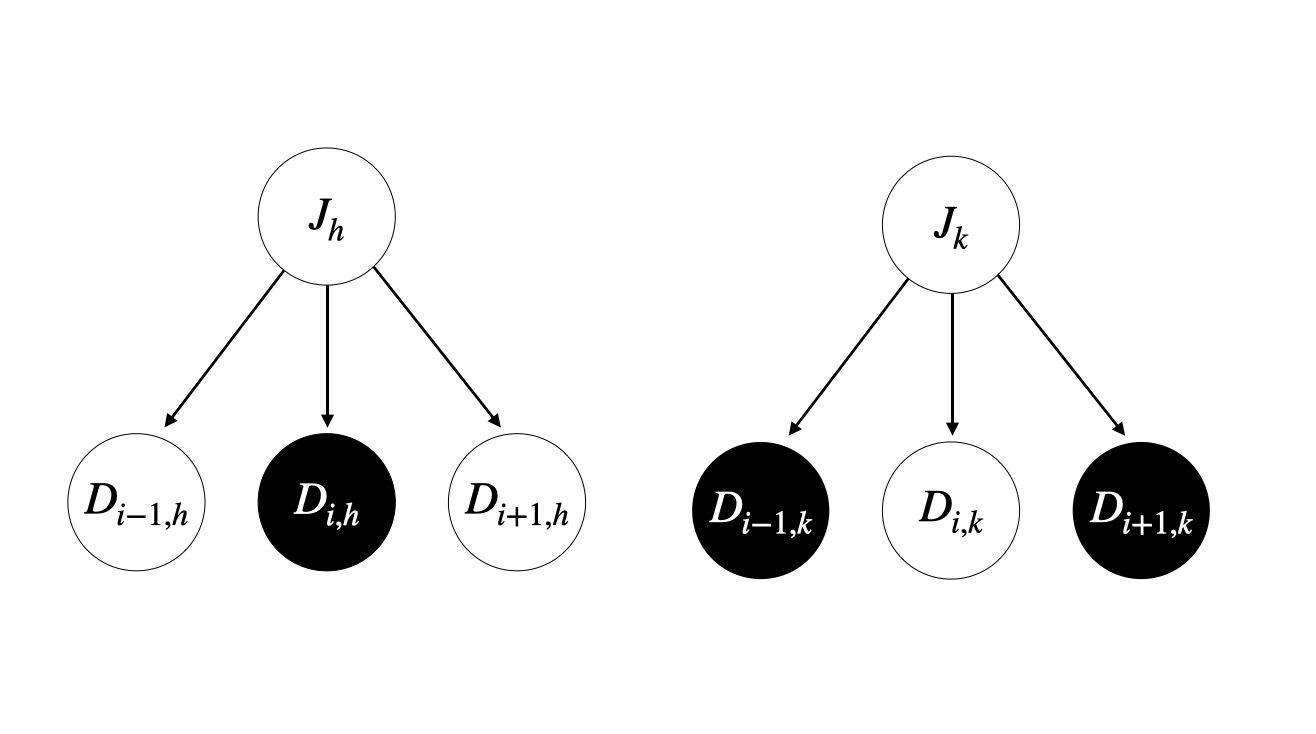}
  \caption{Uniform (case level) \\randomization.}
  \label{fig:sub1}
\end{subfigure}%
\begin{subfigure}{0.5\textwidth}
  \centering
  \captionsetup{justification=centering}
  \includegraphics[width=\linewidth]{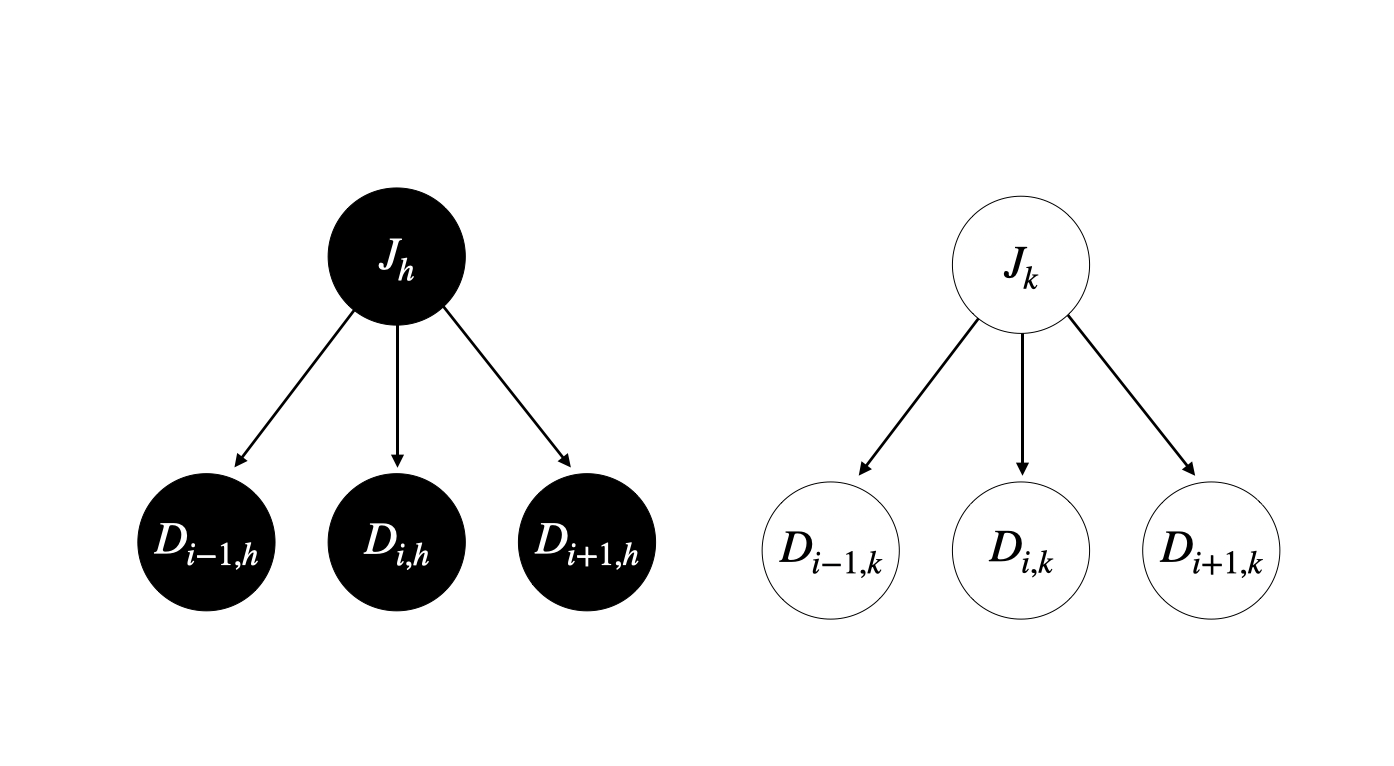}
  \caption{Two level (Decision-maker level) \\randomization.}
  \label{fig:sub2}
\end{subfigure}
\caption{Prior experimental designs randomize the treatment for the algorithmic intervention at (a) the case level, and not (b) the decision-maker level.}\label{fig:comparison_randomization_level}
\end{figure*}


Recall the definition of the average treatment effect (ATE) of treatment $Z_{i,k}$ on decisions $D_{i,k}$:
\begin{equation}
    ATE := \E[D(Z = 1) - D(Z = 0)]
\end{equation}
In our two hypothetical experiments, suppose we use the following estimator of ATE:
\begin{equation}
    \widehat{ATE} := \frac{2}{n}\sum_{i=1}^n\sum_{k=1}^2D_{i,k}\Ind{ Z_{i,k} = 1} - D_{i,k} \Ind{Z_{i,k} = 0}
\end{equation}
Given that decision subjects are randomly assigned to treatment (with half of the decision subjects being treated), $\widehat{ATE}$ appears to be an unbiased estimate of $ATE$. Yet, under our model of $D_{i,k}$ (e.g., \eqref{eq:treatment_exposure}), the $ATE$ is clearly under-specified, as $D_{i,k}$ in general depends not only on $Z_{i,k}$ but also $\{Z_{1, k}, \cdots, Z_{i-1, k}\}$. If the goal is to estimate the treatment effect of providing a predictive decision aid to a decision maker \emph{consistently}, we would really like, perhaps, to estimate the average treatment effect under total treatment.
\begin{equation}
    \overline{ATE} := \E[D(\vect{Z} = \vect{1}) - D(\vect{Z} = \vect{0})].
\end{equation} In that case, the standard estimator $\widehat{ATE}$ would suffer from large bias if used to estimate $\overline{ATE}$, resulting in underestimation (or overestimation) of the treatment effect. Proposition~\ref{prop:treatment_effect} and Proposition~\ref{prop:capacity_constraint} below show that there can be a significant gap between the expectation of the $\widehat{ATE}$ estimator under different treatment assignments, under the  treatment exposure model and the capacity constraint model of judge decision making respectively. Specifically, we compare $\widehat{ATE}_{\text{uniform}}:=\widehat{ATE}(\vect{Z})$ and that of $\widehat{ATE}_{\text{two-level}}:=\widehat{ATE}(\vect{Z}')$. We include the proofs in Appendix~\ref{app:proofs}.


\begin{proposition}[Estimated treatment effects under treatment exposure model]\label{prop:treatment_effect}
Assume both judge 1 and 2's decision model is as described in equation~\eqref{eq:model_of_judge}, and $J_{i,k}$ follows \eqref{eq:treatment_exposure}, with $b_{k} = b \in(0,1) \forall k$ and $f(x) = ax, a \in (0,1-b)$. Suppose we have $\E[\bar{D}(\hat{Y}_1) - \lambda_k(X_1)] = \rho > 0$. Then
\begin{equation}
  \E[\widehat{ATE}_{\text{two-level}}]- \E[\widehat{ATE}_{\text{uniform}}] = \frac{a \cdot \rho }{2}
\end{equation}
\end{proposition}
Intuitively, the gap between the two estimates increases as 1) the judge responsiveness is more sensitive to past exposure ($a$ is larger) and 2) the expected difference between the judge's default decision and the algorithmically recommended decision increases ($\rho$ is larger).

As seen above, under the treatment exposure model, uniform randomization results in the underestimation of the $\overline{ATE}$. Under the capacity constraint model (Proposition~\ref{prop:capacity_constraint}), however, uniform randomization actually results in the overestimation of the $\overline{ATE}$.


\begin{proposition}[Estimated treatment effects under capacity constraint model]\label{prop:capacity_constraint}
Assume both judge 1 and 2's decision model is as described in equation~\eqref{eq:model_of_judge}, and $J_{i,k}$ follows \eqref{eq:capacity}, with $b_{k} = b \in(0,1) \forall k$ and $f(x) = ax, a \in (0,1-b)$. Suppose we have $\E[\bar{D}(\hat{Y}_1) - \lambda_k(X_1)] = \rho > 0$ and $\Pr(\hat{Y}_1 > 0) = \gamma > 0$. Then
\begin{equation}
  \E[\widehat{ATE}_{\text{two-level}}]- \E[\widehat{ATE}_{\text{uniform}}] = -\frac{a\cdot\rho\cdot\gamma}{2}
\end{equation}
\end{proposition}
Here the gap between the two estimates increases as 1) $a$ is larger and 2) $\rho$ is larger, as before, but also when 3) the positive prediction rate increases ($\gamma$ is larger).

In this section we illustrated the implications that different randomization schemes have on treatment effect estimation under our model of judge decision making bias. Although we have focused our investigation on estimation of the population treatment effect, our observation that treatment randomization schemes can introduce bias into the estimated treatment effects likely also extends to other types of treatment effects such as those based on principal stratification \citep[see e.g., Average Principal Causal Effects from][]{imai2020experimental}. 

\section{SEMI-SYNTHETIC EXPERIMENTS}
\label{sec:expt}
We test our findings empirically by setting up semi-synthetic simulations, using the data from the experiment conducted in ~\cite{imai2020experimental}, as described in Section~\ref{sec:existing}. Further details of our experimental setting and results can be found in Appendix ~\ref{app:exps}. 


\paragraph{Simulation details.} 
Using our model of $J_{i,k}$, we can derive an alternate set of judge decisions to that in the original experiment. In the original study, all of the treated and untreated cases were assigned to a single judge. We simulate multiple judges, and a corresponding set of alternative decisions to observe the impact of these biases (Section~\ref{sec:judge_dec}) on estimations of the overall average treatment effect estimate on judge decisions. To do this, we change the simulated judge $J_k$ assigned to a particular case, and their subsequent decisions $D_{i,k}$, using the models of responsiveness in Equations~\ref{eq:treatment_exposure}, \ref{eq:capacity}, and \ref{eq:low_trust}. Each model is explored in Experiments 1 (Figure~\ref{fig:exp1_results}), 2 (Figure~\ref{fig:exp2_results}), and 3 (Figure~\ref{fig:exp3_results}) respectively. We do not manipulate other case details directly for our simulation, and observe the impact of these changes to the average treatment effect (ATE) of being exposed to a PSA score on decisions. We use a binary normalized form of judge decisions (i.e. $D_{i,k} = 0$ for a signature bond release, $D_{i,k} = 1$ for any kind of cash bond). 
We set $\lambda_k(X_i)$ to be the set of provided original judge decisions, and $\recdec{i}$ to be the recommended decision under official interpretation guidelines for the PSA. 
Descriptive statistics, implementation details and results for decision correctness can be found in Appendix~\ref{app:exps}. 
The data from  ~\cite{imai2020experimental} is publicly available \href{https://dataverse.harvard.edu/dataverse/harvard?q=%20Replication%20data%20for:%20Experimental%20evaluation%20of%20algorithm-assisted%20human%20decision-making:%20Application%20to%20pretrial%20public%20safety%20assessment}{here}, and the replication code for experiments is available  \href{https://anonymous.4open.science/r/exp-design-human-decisions-A3CD/README.md}{here}.
Each result is over a simulation with 1000 trials, and plotted with standard error bars.

Note that 
the average treatment effects reported for the empirical results are not absolute, but are changes in the measured treatment effect relative to  the baseline reported treatment effect in the original single-judge study. As such, we already have a baseline as defined by the original experiment and specifically report deviations from this baseline as our main results.

\subsection{Experiment 1: Treatment Exposure Effect}

As discussed in Appendix ~\ref{app:exps}, we investigate two cases -- one with two judges (treated at 50\% each (Scenario 1A) or ~100\% treated and ~0\% treated (Scenario 1B)), and one with three judges (treated at 33\% each (Scenario 2A) or at 60\%, 30\% and 10\% respectively (Scenario 2B)). We observe that the largest treatment effect measurements on decisions and decision correctness (ie. the correlation of the decision and the recorded downstream outcome) apply in Scenario 1B, though differences are smaller in the latter case. We find that modifying treatment exposure regimes can adjust the causal effect estimate on decisions by up to $~10x$ (i.e. for non-white female (NWF) subjects), though this impact varies (e.g. for white men (WM), there is little observed impact). 

\begin{figure}
  \centering
 \includegraphics[clip, width=\columnwidth]{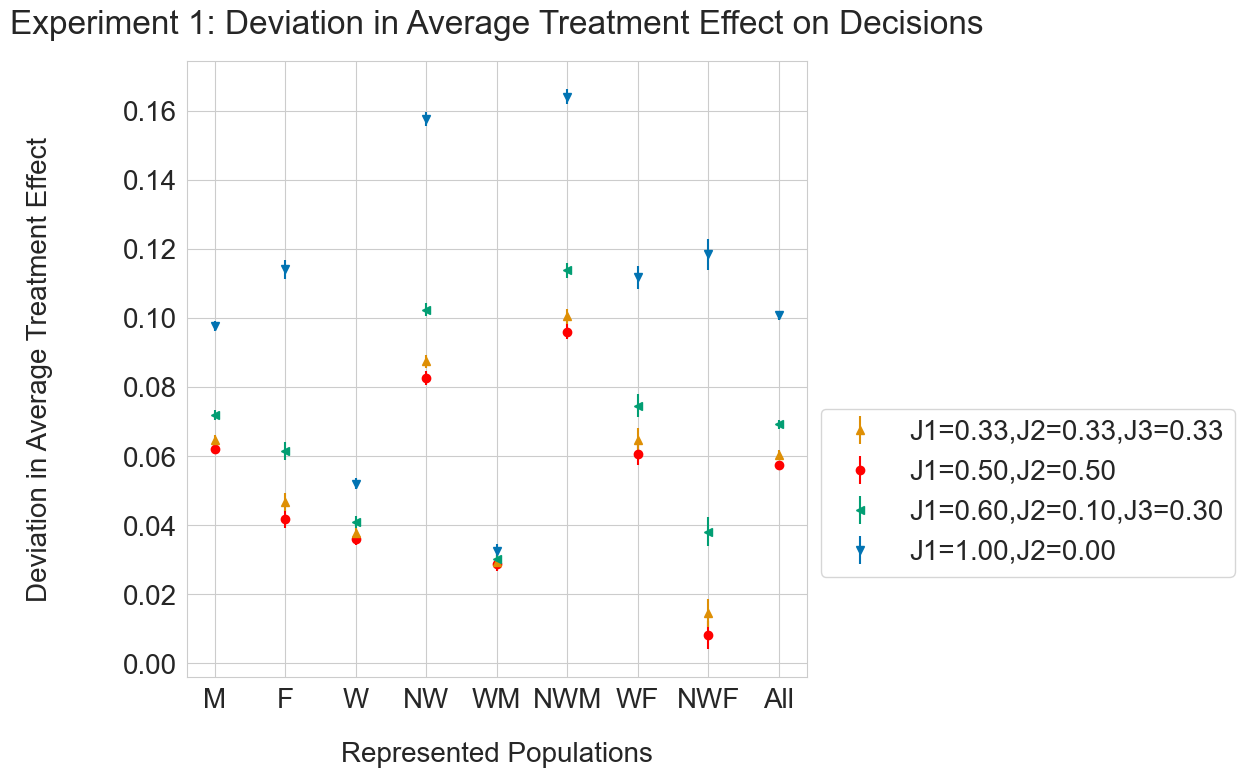}
  \caption{\textbf{Experiment 1: Treatment Exposure Effect.} We empirically observe changes to the ATE under different treatment assignments for judges $J_1$, $J_2$, $J_3$. 
  Results for 1000 trials, with M = Male, F = Female, NW = Non-White, W = White.}
  \label{fig:exp1_results}
\end{figure}

\subsection{Experiment 2: Capacity Constraint effect}

The available data allows decision thresholds to be set at a level of 3,4,5 or 6. Given the same treatment exposure, we find that there are average treatment effect measurement differences under different algorithm decision thresholds (and thus changes to the positive prediction rate of the algorithm). For example, at a lower algorithmic decision threshold (eg. threshold = 3), the algorithm predicts $\hat{Y} = 1$ (i.e. detainment) at a higher frequency and, under our judge responsiveness model, the measured treatment effect on decisions is lower. Interestingly, this effect does not meaningful impact the measure of changes to simulated decision correctness (See details in Appendix ~\ref{app:exps}).

\begin{figure}
  \centering
  \includegraphics[clip, width=0.8\columnwidth]{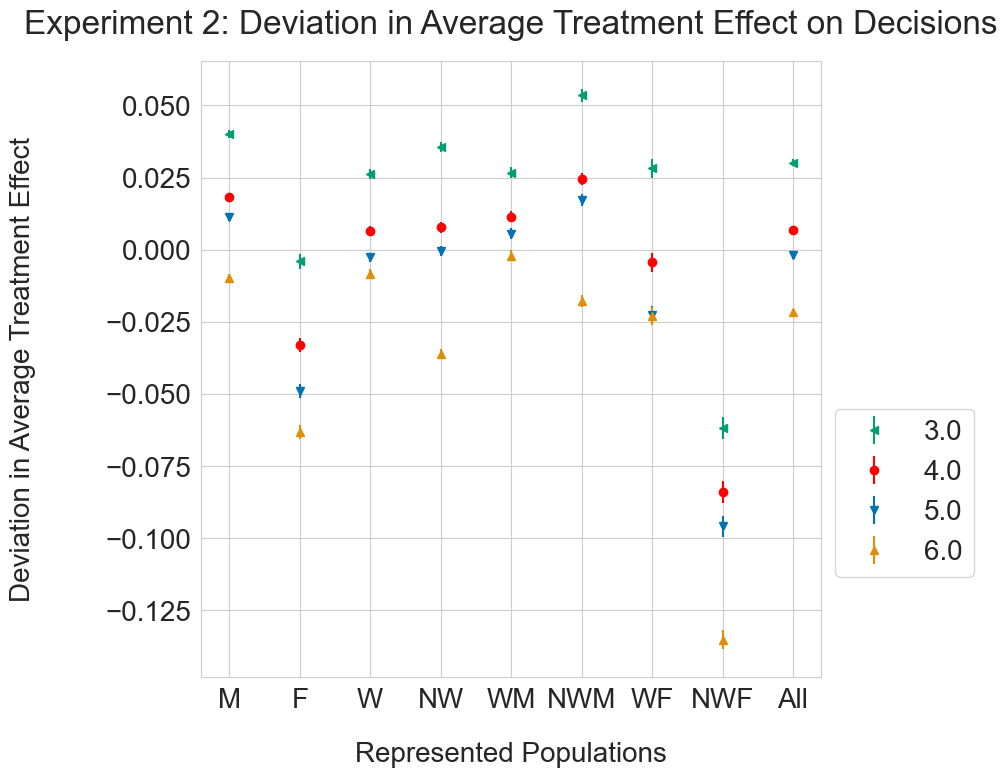}
 \caption{\textbf{Experiment 2: Capacity Constraint effect.} By modifying the threshold on the PSA score (which ranges from 3 to 6), we can simulate a modification of the decision aid's positive prediction rate. 
 We find that a lower prediction threshold means that the algorithm predicts detainment at a higher frequency and the measured treatment effect on decisions is lower. Results are for 1000 trials, with M = Male, F = Female, NW = Non-White, W = White. In this illustrative example, $J_1$, $J_2$ and $J_3$ are assigned $Z=1$ in about 33\% of their cases respectfully.} 
  \label{fig:exp2_results}
\end{figure}

\subsection{Experiment 3: Low Trust Effect}
In the low trust setting, we derive hypothetical decisions made under different algorithmic recommendation accuracy “boosts”, where an accuracy boost of 0 results in the default algorithmic recommendation accuracy of the real world setting, with respect to actual real world outcomes ($\approx 54\%$). If the boost is at 0.1, then 10\% of the algorithmic recommendations are flipped to the actual correct decision with respect to the recorded real world outcome - this simulates a higher accuracy model (eg. the 10\% boosted algorithmic recommendation accuracy is $\approx 67\%$). Under our judge responsiveness model in this setting, low accuracy predictions lead to an underestimation of the measured treatment effect on decisions and this impact is even greater for decision correctness (See details in Appendix ~\ref{app:exps}).

\begin{figure}
  \centering
 \includegraphics[clip, width=0.8\columnwidth]{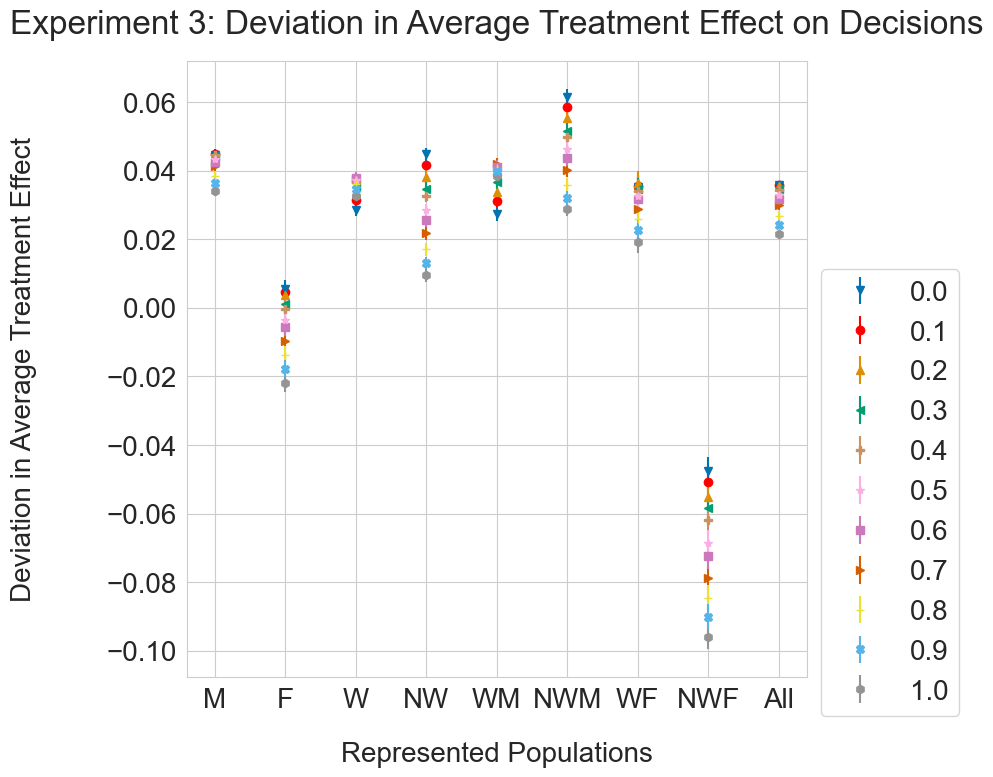}
 \caption{\textbf{Experiment 3: Low Trust Effect.} Under our model, lower accuracy algorithmic recommendations lead to less judge responsiveness and an underestimation of the measured treatment effect on decisions. We model this with hypothetical decisions made under incremental accuracy “boosts” to the algorithmic recommendation. Results are for 1000 trials, with M = Male, F = Female, NW = Non-White, W = White. In this illustrative example, $J_1$, $J_2$ and $J_3$ are assigned $Z=1$ in about 33\% of their cases respectfully.} 
  \label{fig:exp3_results}
\end{figure}

\section{CONCLUSION}

Past work looking at judge responsiveness to automated decision systems explicitly models the phenomenon as a fixed, static property of the judge interacting with the prediction used for automated decision support. This applies to past work in HCI (eg, ~\cite{green2019disparate} assume a fixed judge bias, similar to $b_k$) and past work on observational studies done in economics (eg. ~\cite{albright2019if, hoffman2018discretion} and others also assume a fixed judge bias, similar to $b_k$). In this work, we propose that the variability in judge responsiveness is not just a fixed attribute inherent to each judge, but that this judge variability likely also arises as the result of the impact of specific experiment design decisions on overall judge responsiveness (i.e. experiment design choices that impact treatment exposure, model trust, and capacity constraints); judge responsiveness is thus not just a static property inherent to each judge, but also impacted by conditions that the experimenter can control. We thus integrate our understanding of interference in causal inference with the existing discussion on judge responsiveness, demonstrating that if we model the prediction as part of an intervention, estimates of causal effects will be influenced by experiment design decisions through the mechanism of the judge’s decision-making across cases. 



This work has implications on how we design future experiments in the context of prediction-based interventions, as well as how we study and measure treatment effects in the context of experimental evaluations for prediction-based interventions.
Further study of multi-judge experimental settings, two-level randomization designs, and multi-factor experiment designs (i.e., allowing for comparisons of multiple prediction system options, multiple thresholds, and thus positive prediction rates, etc.) for prediction-based interventions is needed. 
Further theoretical directions to explore include further discussions on the interactions between multiple responsiveness factors on overall treatment effects, long-term learning effects, group-level impacts, as well as the exploration of measurement corrections to the reported deviations to causal effect estimates. We truly hope that many of these directions can be explored further in future work. 


Although our empirical results are limited to data from the judicial pre-trial risk assessment context, our justification for the judge responsiveness factors derives from empirical and theoretical evidence observed across a variety of domains, including healthcare and education (see Appendix ~\ref{app:further_rel-domains}, ~\ref{app:further_rel-evidence}). The judge responsiveness factors thus apply more broadly than the criminal justice domain and have been previously observed in other settings. We hope to see further exploration of these other domains in future work.


 \paragraph{Acknowledgments}
I.D.R. was funded by the Mozilla Foundation and the MacArthur Foundation. L.T.L. was supported by an Amazon Research Award Fall 2023. Any opinions, findings, and conclusions or recommendations expressed in this material are those of the authors and do not reflect the views of Amazon. We appreciate the comments and feedback from colleagues Avi Feller, Benjamin Recht and Imai Kosuke.  

\bibliographystyle{plainnat} 
\bibliography{ref} 

\section*{Checklist}

The checklist follows the references. For each question, choose your answer from the three possible options: Yes, No, Not Applicable.  You are encouraged to include a justification to your answer, either by referencing the appropriate section of your paper or providing a brief inline description (1-2 sentences). 
Please do not modify the questions.  Note that the Checklist section does not count towards the page limit. Not including the checklist in the first submission won't result in desk rejection, although in such case we will ask you to upload it during the author response period and include it in camera ready (if accepted).

\textbf{In your paper, please delete this instructions block and only keep the Checklist section heading above along with the questions/answers below.}

 \begin{enumerate}

 \item For all models and algorithms presented, check if you include:
 \begin{enumerate}
   \item A clear description of the mathematical setting, assumptions, algorithm, and/or model. [Yes, see Section~\ref{sec:model}.]
   \item An analysis of the properties and complexity (time, space, sample size) of any algorithm. [Not Applicable]
   \item (Optional) Anonymized source code, with specification of all dependencies, including external libraries. [Yes]
 \end{enumerate}

 \item For any theoretical claim, check if you include:
 \begin{enumerate}
   \item Statements of the full set of assumptions of all theoretical results. [Yes, see section~\ref{sec:results}]
   \item Complete proofs of all theoretical results. [Yes we sketch proof ideas in \ref{sec:results} and include complete proofs in Appendix.]
   \item Clear explanations of any assumptions. [Yes]     
 \end{enumerate}

 \item For all figures and tables that present empirical results, check if you include:
 \begin{enumerate}
   \item The code, data, and instructions needed to reproduce the main experimental results (either in the supplemental material or as a URL). [Yes in Section~\ref{sec:expt} and Appendix.]
   \item All the training details (e.g., data splits, hyperparameters, how they were chosen). [Not Applicable]
         \item A clear definition of the specific measure or statistics and error bars (e.g., with respect to the random seed after running experiments multiple times). [Yes]
         \item A description of the computing infrastructure used. (e.g., type of GPUs, internal cluster, or cloud provider). [Yes, in appendix]
 \end{enumerate}

 \item If you are using existing assets (e.g., code, data, models) or curating/releasing new assets, check if you include:
 \begin{enumerate}
   \item Citations of the creator If your work uses existing assets. [Yes, see citation in text.]
   \item The license information of the assets, if applicable. [Not Applicable]
   \item New assets either in the supplemental material or as a URL, if applicable. [Yes, code base linked anonymously]
   \item Information about consent from data providers/curators. [Not Applicable]
   \item Discussion of sensible content if applicable, e.g., personally identifiable information or offensive content. [Not Applicable]
 \end{enumerate}

 \item If you used crowdsourcing or conducted research with human subjects, check if you include:
 \begin{enumerate}
   \item The full text of instructions given to participants and screenshots. [Not Applicable]
   \item Descriptions of potential participant risks, with links to Institutional Review Board (IRB) approvals if applicable. [Not Applicable]
   \item The estimated hourly wage paid to participants and the total amount spent on participant compensation. [Not Applicable]
 \end{enumerate}

 \end{enumerate}

\appendix

\onecolumn

\section*{Appendix / supplemental material}

\section{Examples of capacity constraint and low trust model}\label{app:instance}

A particular instantiation of the \emph{capacity constraint} model is based on a simple thresholding of the average number of `high risk' predictions seen:
\begin{equation}
J_{i,k}=
    \begin{cases}
        1 & \text{if } \sum_{m=1}^{i} Z_{m,k}\cdot \Ind{\hat{Y}_m > 0} < i\tau\\
        0 & \text{o.w.}
    \end{cases},
\end{equation}
where $\tau \in (0,1)$. This instance of the judge's decision making model suggests that the judge becomes always responsive to the algorithmic recommendation as long as the cumulative average `high risk' prediction is below~$\tau$. Otherwise, they always trust themselves over the model in moments of disagreement between their default decision and the ADS decision.

A particular instantiation of the \emph{low trust} model is based on a simple thresholding of the average number of predictive `errors' observed:
\begin{equation}
J_{i,k}=
    \begin{cases}
        1 & \text{if } \sum_{m=1}^{i}Z_{m,k} \cdot \Ind{\hat{Y}_m \neq Y_{m,k}} < i\tau\\
        0 & \text{o.w.}
    \end{cases},
\end{equation}
where $\tau \in (0,1)$. This instance of the judge's decision making model suggests that the judge becomes always responsive to the algorithmic recommendation as long as the cumulative average error rate of prediction $\hat{Y}$ is below $\tau$. Otherwise, they always trust themselves over the model in moments of disagreement between their default decision and the ADS decision. 

\section{Proofs}\label{app:proofs}

\begin{proof}[Proof of Theorem~\ref{thm:sutva_violation}]  Consider $\vect{Z}, \vect{Z}'$ such that $Z_{1,k}= Z_{2,k} = \cdots = Z_{i-1,k} = 0$ and  $Z'_{1,k}= Z'_{2,k} = \cdots = Z'_{i-1,k} = 1$. Also let $Z_{i,k} = Z_{i,k}' = 1$. First consider the case that $J_{i,k}$ is monotonically non-decreasing. Then we have
\begin{align*}
    J_{i,k}(\vect{Z}) < J_{i,k}(\vect{Z}'),
\end{align*}
by our assumption that $J_{i,k}(z_{1,k}, \cdots, z_{i-1,k})$ is strictly increasing in at least one of $z_{1,k}, \cdots, z_{i-1,k}$. In other words, we have $\Pr\left(\epsilon_{i,k}(\vect{Z}) = 1) < \Pr(\epsilon_{i,k}(\vect{Z}') = 1\right)$.

Applying the definition of $D_{i,k}$, we then lower bound the probability that $D_{i,k}(\vect{Z})$ differs from $D_{i,k}(\vect{Z}')$ as follows.
\begin{align*}
    \Pr\left(D_{i,k}(\vect{Z}) \ne D_{i,k}(\vect{Z}')\right) &\ge \Pr\{\lambda_k(X_i) \neq \recdec{i}\}\cdot \Pr\left(\epsilon_{i,k}(\vect{Z}) \ne \epsilon_{i,k}(\vect{Z}')\right) \\
    &> 0.
\end{align*}
  The proof proceeds analogously for the case where $J_{i,k}$ is monotonically non-increasing.
\end{proof}

\begin{proof}[Proof of Proposition~\ref{prop:treatment_effect}]
Recall \begin{equation}
    D_{i,k} = \left( \bar{D}(\hat{Y}_i)\cdot \epsilon_{i,k} + \lambda_k(X_i) \cdot (1-\epsilon_{i,k})\right) \cdot Z_{i,k} + \lambda_k(X_i) \cdot (1- Z_{i,k}).
\end{equation}

Under uniform randomization, we have 
\begin{align*}
    \E[\widehat{ATE}_{\text{uniform}}] &= \frac{2}{n} \sum_{i=1}^n\sum_{k=1}^2 \E[D_{i,k}\Ind{ Z_{i,k} = 1} - D_{i,k} \Ind{Z_{i,k} = 0}]
    \\
    &= \frac{1}{2n}\sum_{i=1}^n\E[D_{i,1}\mid Z_{i,1} = 1] +\E[D_{i,2}\mid Z_{i,2} = 1]\\
    &\quad\quad\quad -\E[D_{i,1}\mid Z_{i,1} = 0]-\E[D_{i,2}\mid Z_{i,2} = 0] \\
    &= \frac{1}{2n}\left(\sum_{i=1}^n\E[\bar{D}(\hat{Y}_i) - \lambda_k(X_i)]\cdot \E\left[b+ \frac{a}{i-1} \sum_{m=1}^{i-1} Z_{m,1} \mid Z_{i,1} = 1\right]  \right)\\
    &\quad +\frac{1}{2n}\left(\sum_{i=1}^n\E[\bar{D}(\hat{Y}_i) - \lambda_k(X_i)]\cdot \E\left[b+ \frac{a}{i-1} \sum_{m=1}^{i-1} Z_{m,2} \mid Z_{i,2} = 1\right]  \right)\\
    &= \frac{1}{2}\left(\E[\bar{D}(\hat{Y}_i) - \lambda_k(X_i)]\cdot (b+ \frac{a}{2}) \right)+\frac{n}{4}\left(\E[\bar{D}(\hat{Y}_i) - \lambda_k(X_i)]\cdot (b+ \frac{a}{2})  \right)\\
    &= \left(\E[\bar{D}(\hat{Y}_1) - \lambda_k(X_1)]\cdot (b+ \frac{a}{2}) \right)
\end{align*}
Note that in the third equality we applied the definition of $D_{i,k}$ and $J_{i,k}$, as well as the independence between $(\hat{Y}_i, X_i)$ and $Z_{i,k}$. In the fourth equality, we use the fact that $\E[Z_{m,2}\mid Z_{i,2} = 1] = 1/2$ for this treatment assignment. In the last equality we used the fact that $\hat{Y}_i$'s abnd $X_i's$ are i.i.d.

Under two-level randomization, we have 

\begin{align*}
    \E[\widehat{ATE}_{\text{two-level}}] &=\frac{2}{n} \sum_{i=1}^n\sum_{k=1}^2 \E[D_{i,k}\Ind{ Z_{i,k} = 1} - D_{i,k} \Ind{Z_{i,k} = 0}] \\
    &= \frac{1}{n}\sum_{i=1}^n\E[D_{i,1}\mid Z_{i,1} = 1] -\E[D_{i,2}\mid Z_{i,2} = 0] \\
    &=\frac{1}{n}\sum_{i=1}^n\E[\bar{D}(\hat{Y}_i) - \lambda_k(X_i)]\cdot \E\left[b+ \frac{a}{i-1} \sum_{m=1}^{i-1} Z_{m,1} \mid Z_{i,1} = 1\right]  \\
    &= \left(\E[\bar{D}(\hat{Y}_1) - \lambda_k(X_1)]\cdot (b+ a) \right)
\end{align*}

\end{proof}

\begin{proof}[Proof of Proposition~\ref{prop:capacity_constraint}]

Under uniform randomization, we have 
\begin{align*}
    \E[\widehat{ATE}_{\text{uniform}}] &= \frac{2}{n} \sum_{i=1}^n\sum_{k=1}^2 \E[D_{i,k}\Ind{ Z_{i,k} = 1} - D_{i,k} \Ind{Z_{i,k} = 0}]
    \\
    &= \frac{1}{2n}\sum_{i=1}^n\E[D_{i,1}\mid Z_{i,1} = 1] +\E[D_{i,2}\mid Z_{i,2} = 1]\\
    &\quad\quad\quad -\E[D_{i,1}\mid Z_{i,1} = 0]-\E[D_{i,2}\mid Z_{i,2} = 0] \\
    &= \frac{1}{2n}\left(\sum_{i=1}^n\E[\bar{D}(\hat{Y}_i) - \lambda_k(X_i)]\cdot \E\left[b- \frac{a}{i-1} \sum_{m=1}^{i-1} Z_{m,1} \cdot \Ind{\hat{Y}_m >0}\mid Z_{i,1} = 1\right]  \right)\\
    &\quad +\frac{1}{2n}\left(\sum_{i=1}^n\E[\bar{D}(\hat{Y}_i) - \lambda_k(X_i)]\cdot \E\left[b- \frac{a}{i-1} \sum_{m=1}^{i-1} Z_{m,2} \cdot \Ind{\hat{Y}_m >0}\mid Z_{i,2} = 1\right]  \right)\\
    &= \frac{1}{2}\left(\E[\bar{D}(\hat{Y}_i) - \lambda_k(X_i)]\cdot (b- \frac{a\cdot \gamma}{2}) \right)+\frac{n}{4}\left(\E[\bar{D}(\hat{Y}_i) - \lambda_k(X_i)]\cdot (b- \frac{a\cdot \gamma}{2})  \right)\\
    &= \left(\E[\bar{D}(\hat{Y}_1) - \lambda_k(X_1)]\cdot (b- \frac{a\cdot \gamma}{2}) \right)
\end{align*}
Note that in the third equality we applied the definition of $D_{i,k}$ and $J_{i,k}$, as well as the independence between $(\hat{Y}_i, X_i)$ and ($Z_{m,k}$, $\hat{Y}_m$) for $m<i$. In the fourth equality, we use the facts that $\E[Z_{m,2}\mid Z_{i,2} = 1] = 1/2$ for this treatment assignment, $Z_{m,2}$ and $\hat{Y}_m$ are independent, and $\E[\Ind{\hat{Y}_m >0}\mid Z_{i,2} = 1] = \gamma$. In the last equality we used the fact that $\hat{Y}_i$'s and $X_i$'s are i.i.d.

Under two-level randomization, we have 

\begin{align*}
    \E[\widehat{ATE}_{\text{two-level}}] &=\frac{2}{n} \sum_{i=1}^n\sum_{k=1}^2 \E[D_{i,k}\Ind{ Z_{i,k} = 1} - D_{i,k} \Ind{Z_{i,k} = 0}] \\
    &= \frac{1}{n}\sum_{i=1}^n\E[D_{i,1}\mid Z_{i,1} = 1] -\E[D_{i,2}\mid Z_{i,2} = 0] \\
    &=\frac{1}{n}\sum_{i=1}^n\E[\bar{D}(\hat{Y}_i) - \lambda_k(X_i)]\cdot \E\left[b- \frac{a}{i-1} \sum_{m=1}^{i-1} Z_{m,1} \cdot\Ind{\hat{Y}_m >0}\mid Z_{i,1} = 1\right]  \\
    &= \left(\E[\bar{D}(\hat{Y}_1) - \lambda_k(X_1)]\cdot (b- a\cdot \gamma) \right)
\end{align*}

\end{proof}

\section{Treatment effect estimation under Low trust model}\label{app:low_trust}

As discussed in Section~\ref{sec:judge_dec}, equation~\eqref{eq:low_trust} cannot directly model  settings where the judge does not observe some outcomes, or where the judge's decisions affect the observability of baseline outcomes. In our illustrative example of treatment effect estimation in the low trust model, we take for granted that the judge is always able to observe the outcome of the previous case before the next case, and is using $\Ind{\hat{Y}_i \ne Y_{i,k}}$ to determine whether the ADS made a mistake, 
as well as that the judge's decision $D_{i,k}$ is independent of outcome of interest $Y_{i,k}$ (i.e., is not a parent of $Y_{i,k}$). We note that this does not hold in our motivating example of pre-trial risk assessment, but may hold approximately in other settings such as medical diagnosis \citep{wong2021external}. In this setting, our assumption is saying that $Y_{i,k}$ (a health status) will always be revealed to the judge (clinician)---regardless of whether they `respond' to the prediction $\hat{Y}_i$ (e.g. by administering antibiotics)--such that they can compare $Y_{i,k}$ and $\hat{Y}_i$ to adjust how much to trust the ADS in the future.

\begin{proposition}[Estimated treatment effects under low trust model]\label{prop:low_trust}
Assume both judge 1 and 2's decision model is as described in equation~\eqref{eq:model_of_judge}, and $J_{i,k}$ follows \eqref{eq:low_trust}, with $b_{k} = b \in(0,1) \forall k$ and $f(x) = ax, a \in (0,1-b)$. Further assume that $D_{i,k}$ is not a parent of $Y_{i,k}$. Suppose we have $\E[\bar{D}(\hat{Y}_1) - \lambda_k(X_1)] = \rho > 0$ and $\Pr(\hat{Y}_1 \neq Y_{i,k}) = \theta > 0$. Then
\begin{equation}
  \E[\widehat{ATE}_{\text{two-level}}]- \E[\widehat{ATE}_{\text{uniform}}] = -\frac{a\cdot\rho\cdot\theta}{2}
\end{equation}
\end{proposition}
Here the gap between the two estimates increases as 1) $a$ is larger,  2) $\rho$ is larger, and when 3) the error rate increases ($\theta$ is larger). The proof of the result is similar to Proposition~\ref{prop:capacity_constraint}, as follows.

\begin{proof}
    
Under uniform randomization, we have 
\begin{align*}
    \E[\widehat{ATE}_{\text{uniform}}] &= \frac{2}{n} \sum_{i=1}^n\sum_{k=1}^2 \E[D_{i,k}\Ind{ Z_{i,k} = 1} - D_{i,k} \Ind{Z_{i,k} = 0}]
    \\
    &= \frac{1}{2n}\sum_{i=1}^n\E[D_{i,1}\mid Z_{i,1} = 1] +\E[D_{i,2}\mid Z_{i,2} = 1]\\
    &\quad\quad\quad -\E[D_{i,1}\mid Z_{i,1} = 0]-\E[D_{i,2}\mid Z_{i,2} = 0] \\
    &= \frac{1}{2n}\left(\sum_{i=1}^n\E[\bar{D}(\hat{Y}_i) - \lambda_k(X_i)]\cdot \E\left[b- \frac{a}{i-1} \sum_{m=1}^{i-1} Z_{m,1} \cdot \Ind{\hat{Y}_m \ne Y_{m,k}}\mid Z_{i,1} = 1\right]  \right)\\
    &\quad +\frac{1}{2n}\left(\sum_{i=1}^n\E[\bar{D}(\hat{Y}_i) - \lambda_k(X_i)]\cdot \E\left[b- \frac{a}{i-1} \sum_{m=1}^{i-1} Z_{m,2} \cdot \Ind{\hat{Y}_m \ne Y_{m,k}}\mid Z_{i,2} = 1\right]  \right)\\
    &= \frac{1}{2}\left(\E[\bar{D}(\hat{Y}_i) - \lambda_k(X_i)]\cdot (b- \frac{a\cdot \theta}{2}) \right)+\frac{n}{4}\left(\E[\bar{D}(\hat{Y}_i) - \lambda_k(X_i)]\cdot (b- \frac{a\cdot \theta}{2})  \right)\\
    &= \left(\E[\bar{D}(\hat{Y}_1) - \lambda_k(X_1)]\cdot (b- \frac{a\cdot \theta}{2}) \right)
\end{align*}
Note that in the third equality we applied the definition of $D_{i,k}$ and $J_{i,k}$, as well as the independence between $(\hat{Y}_i, X_i)$ and $(Z_{m,k}, \hat{Y}_m, Y_{m,k})$ for $m<i$. In the fourth equality, we use the facts that (1) $\E[Z_{m,2}\mid Z_{i,2} = 1] = 1/2$ for this treatment assignment, (2) $Z_{m,2}$ is independent of  $\hat{Y}_m$ and $Y_{m,k}$ (since $D_{m,k}$ is not a parent of $Y_{m,k}$), and (3) $\E[\Ind{\hat{Y}_m \ne Y_{m,k}}\mid Z_{i,2} = 1] = \theta$. In the last equality we used the fact that $\hat{Y}_i$'s and $X_i$'s are i.i.d.

Under two-level randomization, we have 

\begin{align*}
    \E[\widehat{ATE}_{\text{two-level}}] &=\frac{2}{n} \sum_{i=1}^n\sum_{k=1}^2 \E[D_{i,k}\Ind{ Z_{i,k} = 1} - D_{i,k} \Ind{Z_{i,k} = 0}] \\
    &= \frac{1}{n}\sum_{i=1}^n\E[D_{i,1}\mid Z_{i,1} = 1] -\E[D_{i,2}\mid Z_{i,2} = 0] \\
    &=\frac{1}{n}\sum_{i=1}^n\E[\bar{D}(\hat{Y}_i) - \lambda_k(X_i)]\cdot \E\left[b- \frac{a}{i-1} \sum_{m=1}^{i-1} Z_{m,1} \cdot\Ind{\hat{Y}_m \ne Y_{m,k}}\mid Z_{i,1} = 1\right]  \\
    &= \left(\E[\bar{D}(\hat{Y}_1) - \lambda_k(X_1)]\cdot (b- a\cdot \theta) \right)
\end{align*}
\end{proof}

\section{Experiment Details}\label{app:exps}

In this section, we go over key details of the empirical investigation we conduct to validate our work. 

Details on the implementation of data pre-processing semi-synthetic experiments, including supplementary materials and reproduction details are available at: \href{https://anonymous.4open.science/r/exp-design-human-decisions-A3CD/README.md}{this anonymous repository}. The data from the ~\cite{imai2020experimental} study is publicly available \href{https://dataverse.harvard.edu/dataverse/harvard?q=%20Replication%20data%20for:%20Experimental%20evaluation%20of%20algorithm-assisted%20human%20decision-making:%20Application%20to%20pretrial%20public%20safety%20assessment}{here}.

Using these simulations, we investigate whether experimental design choices (eg. treatment assignment, the positive predictive rate of the model, and prediction accuracy) can determine how a judge responds to the prediction-based intervention, and thus directly informs a possible under/over-estimation of the measured treatment effect. Our goal is to empirically measure how such experimental design choices may impact the reported average treatment effect under our given judge responsiveness factors. 

In order to investigate this, we conduct an empirical study similar to that described in Section~\ref{sec:causal_effects}. Namely, for each study, we consider two different treatment assignments (to $n$ total decision subjects) that have the same treated and untreated decision subjects, but assign these treated decision subjects to different decision makers (and we assume each case is only seen by one judge in each treatment assignment). We consider $k$ judges, who are assumed to be identical, apart from their assigned cases. Specifically, all judges have the decision making model as described in equation~\eqref{eq:treatment_exposure}, though we also explore additional experiments where $f(z) = az$ is a simple linear function and $b_1 = b_2 = b$.

\subsection{Simulation Details}

During a 30-month assignment period, spanning 2017 until 2019, judges in Dane County, Wisconsin were either given or not given a pretrial Public Safety Assessment (PSA) score for a given case. The randomization was done by an independent administrative court member - the PSA score was calculated for each case $i$, and, for every even-numbered case, the PSA was shown to the judge $k$ as part of the case files. Otherwise, the PSA was hidden from the judge. Given these scores, the judge needs to make the decision, $D_{i,k}$ to enforce a signature bond, a small case bail (defined as less than \$1000) or a large case bail (defined as more than \$1000). Information about judge decisions, and defendant outcomes, $Y_{i,k}$ (ie. failure to appear (FTA), new criminal activity (NCA) and new violent criminal activity (NVCA)) are tracked for a period of two years after randomization. The case study~\cite{imai2020experimental} only looks at one judge, allowing them to simplify the situation to a single decision-maker scenario.

Using this experiment data, we simulate alternative scenarios of one of the discussed experiment design choices -- treatment assignment. Using our model of $J_{i,k}$, we can derive an alternate set of judge decisions to that in the original experiment. In the original study, all of the treated and untreated cases were assigned to a single judge. We simulate multiple judges, and a corresponding set of alternative decisions to observe the impact of these biases on estimations of the overall average treatment effect estimate on judge decisions. To do this, we change the simulated judge $J_k$ assigned to a particular case, and their subsequent decisions $D_{i,k}$, using our model in Equation~\ref{eq:model_of_judge}. We do not manipulate other case details directly for our simulation. We use a binary normalized form of judge decisions (i.e. $D_{i,k} = 0$ for a signature bond release, $D_{i,k} = 1$ for any kind of cash bond). \textbf{We set $\lambda_k(X_i)$ to be the set of provided original judge decisions}, and $\recdec{i}$ to be the decision recommended under official interpretation guidelines for the PSA. It is important to note that we generally do not know a judge's true default decision for treated cases from this dataset, since this could be a counterfactual, unobserved  quantity whenever we see the judge agreeing with the PSA-recommended decision. However, by heuristically setting $\lambda_k(X_i)$ to be the provided original judge decisions, we are able to simulate counterfactual decisions (to estimate treatment effect on decisions) as well as counterfactual outcomes for cases that were originally released (to estimate effect on decision correctness). Thus we consider our experiments semi-synthetic, meant to illustrate our model numerically rather than draw any statistical conclusions.

Note that the measured outcome of our simulations is not the ~\emph{absolute} average treatment effect across all judges, but the ~\emph{change} in average treatment effect from the baseline effect measured in the decisions from the original experiment. What we are testing is the deviation from the original decisions that we expect under our different judge models, given changes to the responsiveness factors described in Section~\ref{sec:judge_dec}. 

To assign each judge, we define a set of judges $J_k$ who are each assigned a distribution of the treated cases So, for example, the judge distribution ($J_{dist}$) of $J_1 =0.6$, $J_2 =0.3$ and $J_1 =0.1$ means that of all the treated cases (where $Z=1$), 60\% of such cases were assigned to Judge 1, 30\% of such cases were assigned to Judge 2 and 10\% of such cases are assigned to Judge 3. Note that this is a ~\emph{re-allocation} of existing treated cases to a simulated set of additional judges. We only examine how, conditional on this treatment, decision-making can be changed under our model. Also note that we simply re-distribute a proportion of treated cases to each of the simulated judges -- this does not a priori determine the treatment exposure for that judge. For example, if we allocate  60\% of treated cases to $J_1$, and $J_1$ being assigned to $568$ treated cases out of a total of $631$ cases, or around 90\% of their assigned cases being treated. 

In order to define the decision that any judge $J_k$ takes, the decision is an output of the model as described in Sec ~\ref{sec:judge_dec}, and in particular Eq ~\ref{eq:model_of_judge}, where we consider the default $\lambda_k(X_i)$ to be the judge decision from the original experiment -- otherwise the decision is the same as the model recommendation.

Note that at no point in this simulation do we modify ~\emph{outcomes}, since we do not have actual control over the judge decision-making in real life and so cannot impact outcomes in any way. Hence, we look only at the impact of experimental design choices on the average treatment effect on ~\emph{decisions} and ~\emph{decision correctness}. 

Each simulation is completed over 100 trials unless indicated otherwise. 

Note that the measured outcome of our simulations is not the absolute average treatment effect across all judges, but the \emph{change} in average treatment effect from the baseline effect measured in the decisions from the original experiment -- all indications in figures or text indicated "Estimated Average Effects" are thus indicative of "Estimated Deviation of Average Effects" from the reported baseline. 

\paragraph{Binary v. Linear Judge Model}

In our paper, we theoretically discuss a binary instantiation of the judge model (see Eq~\ref{eq:treatment_exposure} in the main text, Eq~\ref{eq:capacity} and Eq~\ref{eq:low_trust} in Appendix ~\ref{app:low_trust}). However, there is also a simple linear instantiation of the judge model that we explored in simulations, and we also have simulation results for a linear model of judge responsiveness, in addition to the threshold-based binary model discussed theoretically in the paper. The results for the linear judge model were also calculated and included in reference materials.

The linear implementation of the \emph{treatment exposure} model is as follows, based on the average number of treated cases seen:
\begin{equation}\label{eq:treatment_exposure_linear}
    J_{i+1,k} = 
b_k + a\left(\frac{1}{i} \sum_{m=1}^{i} Z_{m,k}\right)
\end{equation}
where $a = 1 - b_k$, so at maximum exposure (if 100\% assigned cases are treated),  we get 100\% compliance, or $J_{i+1,k} = 1$. 

The linear implementation of the \emph{capacity constraint} model is as follows, based on the average number of `high risk' predictions seen:
\begin{equation}
J_{i,k}= b_k - a\left(\frac{1}{i}\sum_{m=1}^iZ_{m,k} \cdot \Ind{\hat{Y}_m > 0}\right) 
\end{equation}
where $a = b_k$, so at maximum exposure (if 100\% assigned cases are treated and positive predictions),  we get 0\% compliance, or $J_{i+1,k} = 0$. 

The linear implementation of the \emph{low trust} model is as follows, based on the average number of `correct' predictions seen::
\begin{equation}
J_{i,k}= b_k - a\left(\frac{1}{i}\sum_{m=1}^iZ_{m,k} \cdot \Ind{\hat{Y}_m \neq Y_{m,k}}\right)
\end{equation}
where $a = b_k$, so at maximum exposure (if 100\% assigned cases are treated and correct predictions),  we get 0\% compliance, or $J_{i+1,k} = 0$. 


\paragraph{Dis-aggregated analysis}

We analyze the impact of the judge responsiveness factors and the corresponding experiment design choices on the average treatment effect on decisions and decision correctness. In order to fully understand the scope of the involved impacts, we further dis-aggregate our findings into demographic subgroup populations -- M (male), F (female), W (white), NW (non-white), WM (white male), NWM (non-white male), WF (white female), NWF (non-white female). Although out of scope for this study, we leverage this meta-data from the experiment in order to demonstrate how some of these experiment choices can disproportionately impact one subgroup over another. For instance, across many contexts, the WM subgroup appears to have their ATE measurement the least perturbed by alternative experiment design choices while the ATE estimate for NWF seems especially volatile to mis-measurement under a change of experimental conditions. 

\paragraph{Clarifications on Scope} Our empirical results only focus on data from the pre-trial risk assessment context for various reasons: as the authors state in their paper, ~\citet{imai2020experimental} is likely “to the best of our knowledge, this is the first real-world randomized controlled trial (RCT) that evaluates the impacts of algorithmic risk assessment scores on judges’ decisions in the criminal justice system”. The field of experimental evaluation for algorithmic deployments is increasing but still nascent -- many studies are observational or rely on a benchmark framework. As such, this is one of very few academic experimental evaluations of an ADS system and the only one we could identify for which the full experiment results are publicly available. As we discuss in Appendix ~\ref{app:further_rel-domains} on Additional References, there are experimental evaluations done in other fields (notably education, and healthcare), but none of these existing evaluations have a tractable release of data. For healthcare, for example, clinical trial data needs to be registered and requested through official government channels, and approval is needed for additional synthetic manipulation. That being said, in future work, through the continued explorations of further partnerships and collaborations, we hope to be able to obtain data from such settings, where there are multiple judges (ie. doctors, teachers, hiring managers) available in the data. This would allow us to further empirically validate some of our assumptions on judge responsiveness. Also, as we mention in Appendix ~\ref{app:further_rel-evidence}, our justification for the judge responsiveness factors derives from empirical and theoretical evidence referenced to be observed across a variety of domains -- thus, we know from these historical and academic observations that the judge responsiveness factors apply more broadly than the criminal justice domain, and has been previously observed in other settings.

\subsection{Experiment 1: Treatment Exposure Effect}

We first analyze the impact of the experiment design choice of treatment assignment on judge responsiveness. These results are over a simulation with 100 trials. As discussed in Section ~\ref{sec:expt}, We investigate two cases. One case involves  two judges (assigned 50\% of treated cases each (Scenario 1A) or one assigned ~100\% of treated cases and the other assigned ~0\% of treated cases(Scenario 1B)). The other case involves three judges (each assigned 33\% of treated cases (Scenario 2A) or one assigned 60\%, 30\% and 10\% of treated cases respectively (Scenario 2B)). We observe that the largest treatment effect measurements on decisions and decision correctness (ie. the correlation of the decision and the recorded downstream outcome) apply in Scenario 1B, though differences are smaller in the latter case. We find that modifying treatment exposure regimes can adjust the causal effect estimate on decisions by up to ~10x (for non-white female (NWF) subjects), though this impact varies (for white men (WM), there is little observed impact).

We conduct an empirical study similar to that described in Section ~\ref{sec:results}. Namely, for each study, we consider two different treatment assignments (to $n$ total decision subjects) that have the same treated and untreated decision subjects, but assign these treated decision subjects to different decision makers (and we assume each case is only seen by one judge in each treatment assignment). We consider $k$ judges, who are assumed to be identical, apart from their assigned cases (i.e. each judge possesses a uniform baseline bias of $b_k$). Specifically, both of them have decision making model as described in equation~\eqref{eq:treatment_exposure}, where $f(z) = az$ is a linear function and $b_1 = b_2 = b$.

Under these assumptions, we conduct the following set up for Experiment 1:
\begin{enumerate}
    \item (Uniform randomization) $\vect{Z} =(\vect{Z}_{\cdot, 1},\vect{Z}_{\cdot, 2})$ such that Judge 1 receives 50\% of total cases, 50\% of them are treated and 50\% are untreated. Judge 2 receives other 50\% of total cases, 50\% of them are treated and 50\% are untreated.
    \item (Two-level randomization) $\vect{Z}' = (\vect{Z'}_{\cdot, 1},\vect{Z'}_{\cdot, 2})$ such that Judge 1 receives 50\% of total cases, 100\% of them are treated. Judge 2 receives other 50\% of total cases, 100\% are untreated.
\end{enumerate}

Under these assumptions, we conduct the following set up for Experiment 2:
\begin{enumerate}
    \item (Uniform randomization) $\vect{Z} =(\vect{Z}_{\cdot, 1},\vect{Z}_{\cdot, 2},\vect{Z}_{\cdot, 3})$ such that each judge  receives 33.3\% of total cases, 33.3\% of them are treated. Note that this gives each judge a treatment exposure of around 50\% (since cases are split evenly and treated cases are split evenly amongst judges). 
    \item (Two-level randomization) $\vect{Z}' = (\vect{Z'}_{\cdot, 1},\vect{Z'}_{\cdot, 2},\vect{Z}_{\cdot, 3})$ such that Judge 1 receives 33.3\% of total cases, 60\% of total treated cases. Judge 2 receives other 33.3\% of total cases, but only 30\% of total treated cases. Judge 3 receives other 33.3\% of total cases, but only 10\% of total treated cases.
\end{enumerate}

Note that overall, we only ~/emph{re-assign} treated cases to hypothetical judges for the sake of simulation. We do not ever change the status of a case (from treated to untreated, or vice versa), or modify outcomes. We only update decisions made by such hypothetical judge assignments, based on the heuristics discussed in our paper in section ~\ref{sec:judge_dec}. 

Main results for this experiment are summarized in Figures ~\ref{fig:myfig_ex1_1}, ~\ref{fig:myfig_ex1_2}, ~\ref{fig:myfig_ex1_3}, ~\ref{fig:myfig_ex1_4}.

\begin{figure}[htbp]
\centering
\subfloat[Judge bias $b_k =0.6$, Exposure threshold $t=0.4$]{\label{fig:a}\includegraphics[width=0.45\linewidth]{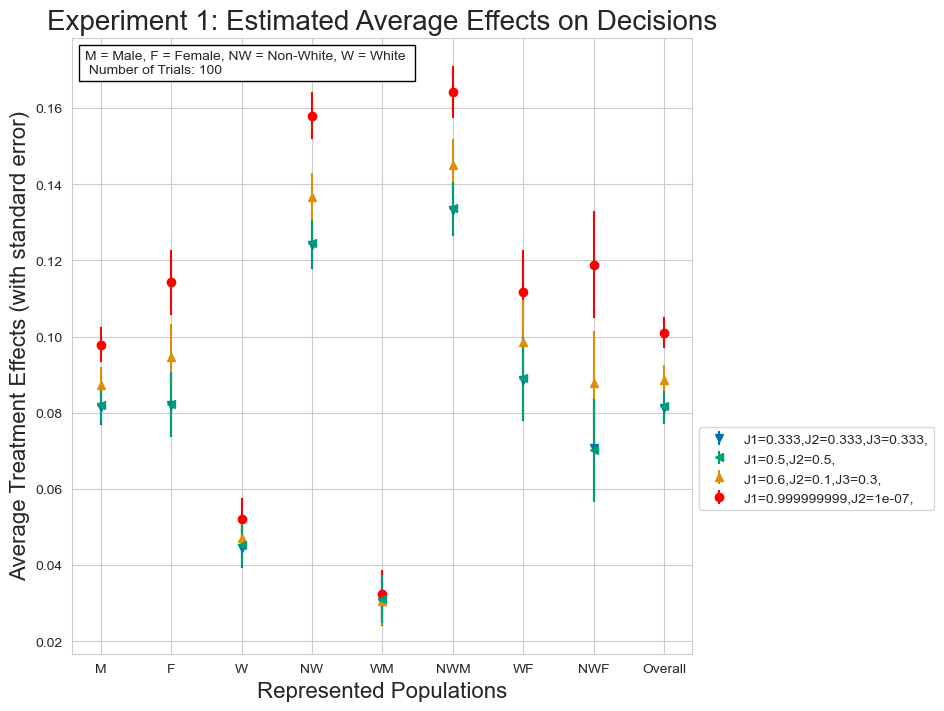}}\qquad
\subfloat[Judge bias $b_k =0.4$, Exposure threshold $ t=0.6$]{\label{fig:b}\includegraphics[width=0.45\linewidth]{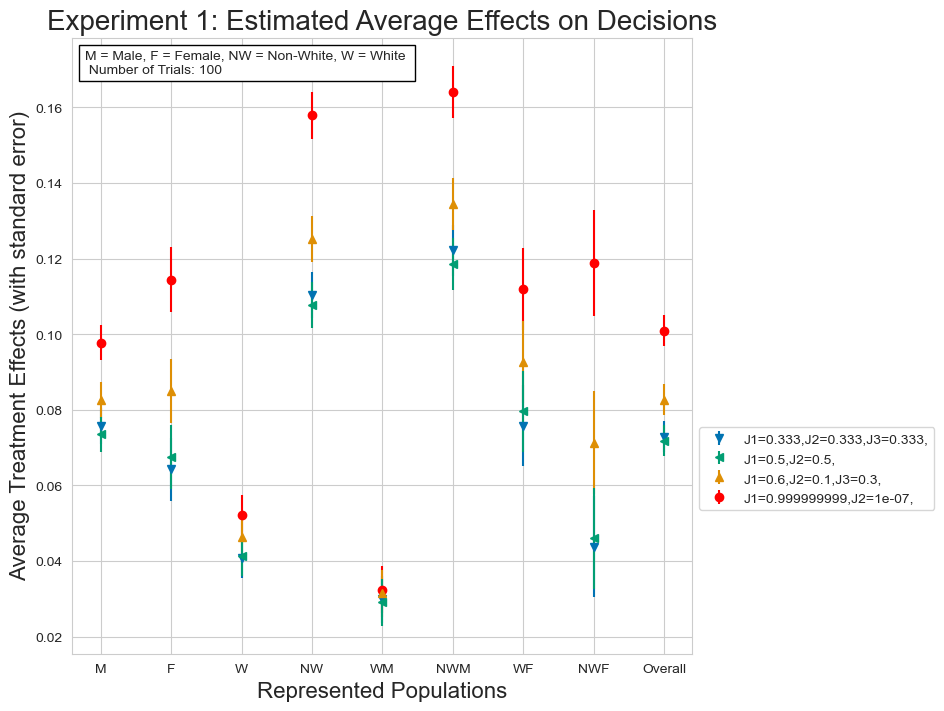}}\\
\subfloat[Judge bias $b_k =0.1$, treatment exposure threshold $ t=0.6$]{\label{fig:c}\includegraphics[width=0.45\textwidth]{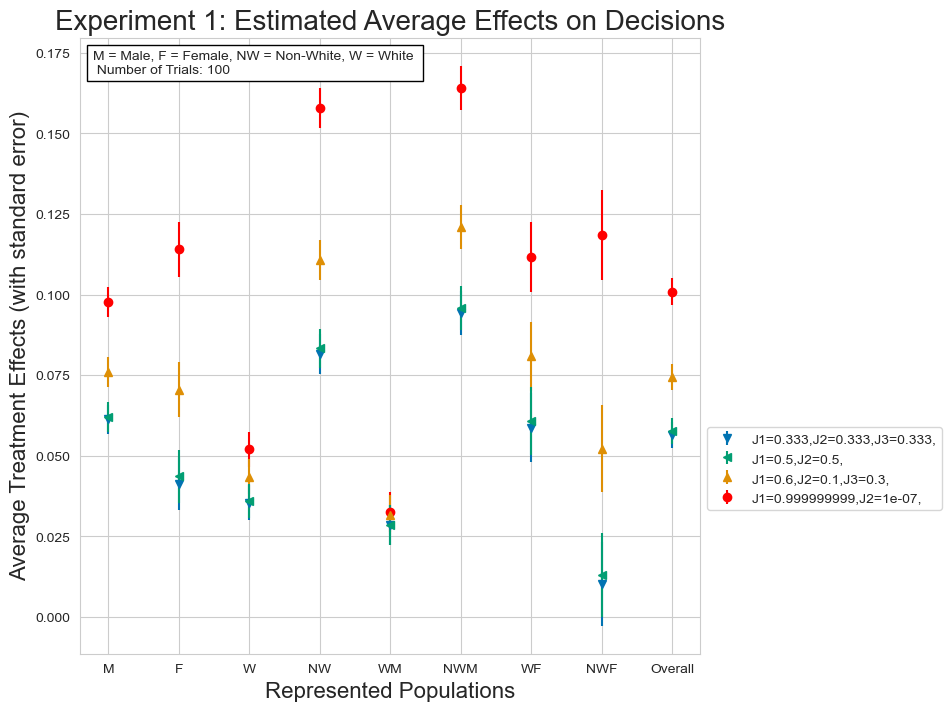}}\qquad%
\subfloat[Judge bias $b_k =0.1$, treatment exposure threshold $ t=0.4$]{\label{fig:d}\includegraphics[width=0.45\textwidth]{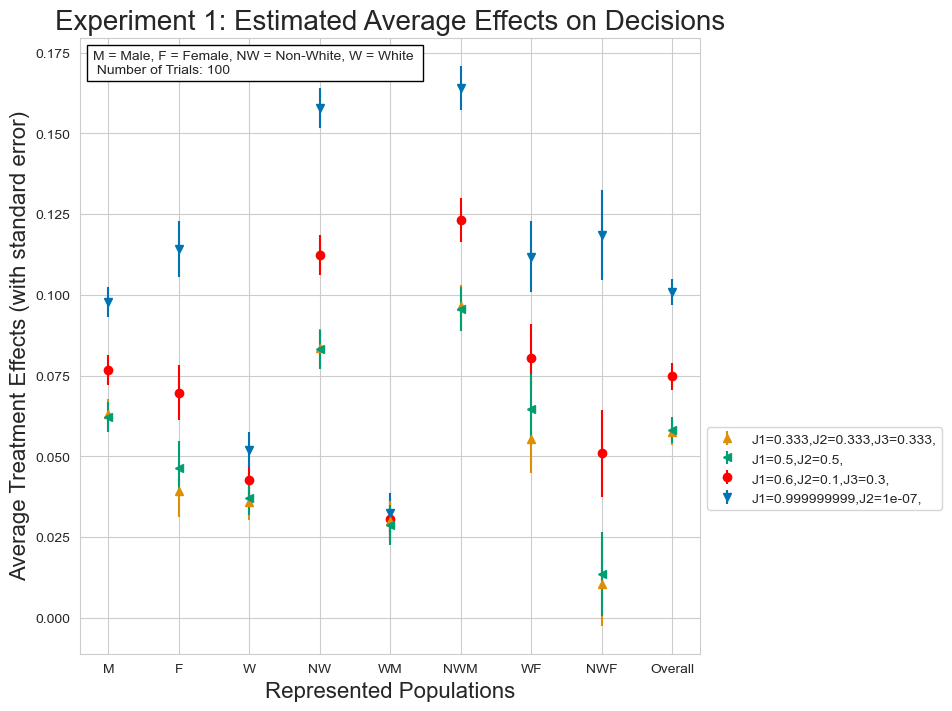}}%
\caption{\textbf{Experiment 1; Average Treatment Effect Changes on Decisions due to treatment exposure.} Judge assignments represent allocations of treated cases across all available simulated judges. So given a judge assignment of $J_1 = 0.6$, 60\% of available treated cases are assigned to $J_1$, giving them an overall treatment exposure of $\approx 90\%$ over all their cases  -- Linear judge model.}
\label{fig:myfig_ex1_1}
\end{figure}

\begin{figure}[htbp]
\centering
\subfloat[Judge bias $b_k =0.4$, Exposure threshold $ t=0.6$]{\label{fig:a}\includegraphics[width=0.45\linewidth]{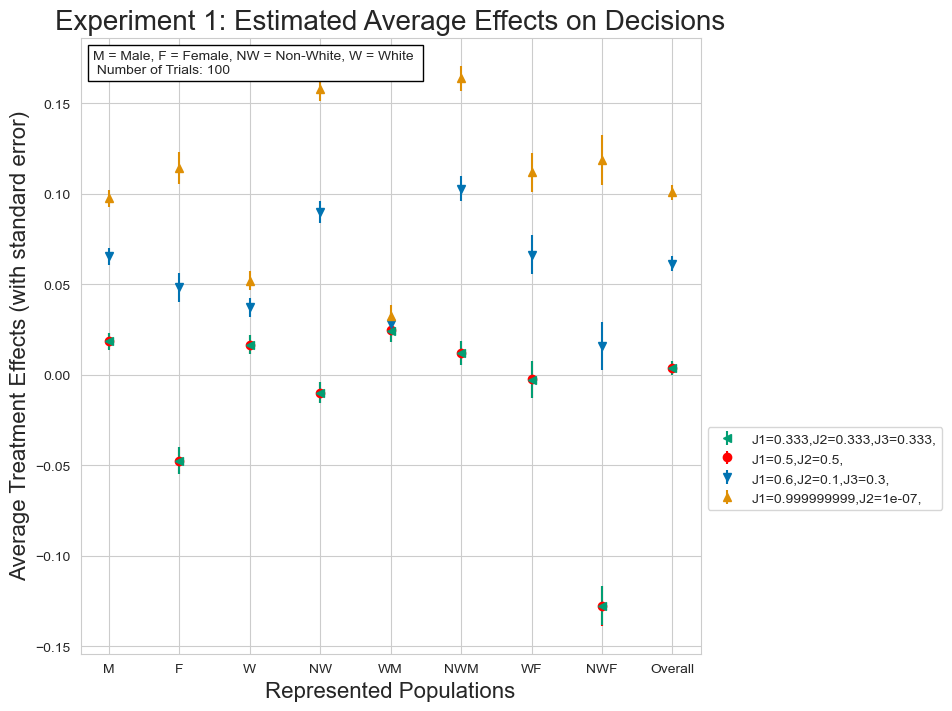}}\qquad
\subfloat[Judge bias $b_k =0.6$, Exposure threshold $ t=0.4$]{\label{fig:b}\includegraphics[width=0.45\linewidth]{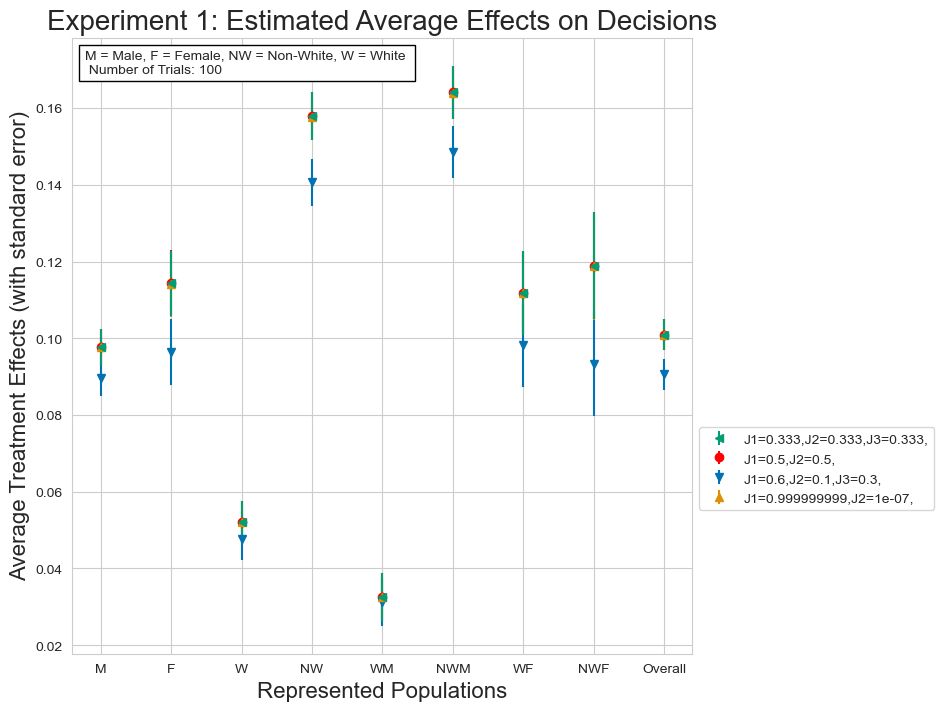}}\\
\subfloat[Judge bias $b_k =0.1$, Exposure threshold $ t=0.6$]{\label{fig:c}\includegraphics[width=0.45\textwidth]{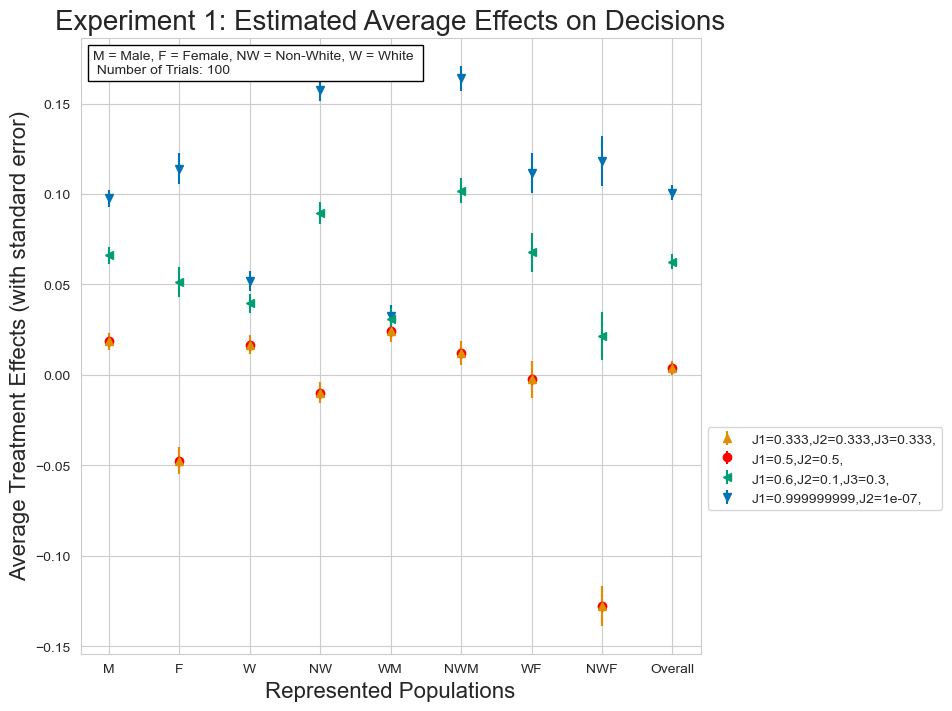}}\qquad%
\subfloat[Judge bias $b_k =0.1$, Exposure threshold $t=0.4$]{\label{fig:d}\includegraphics[width=0.45\textwidth]{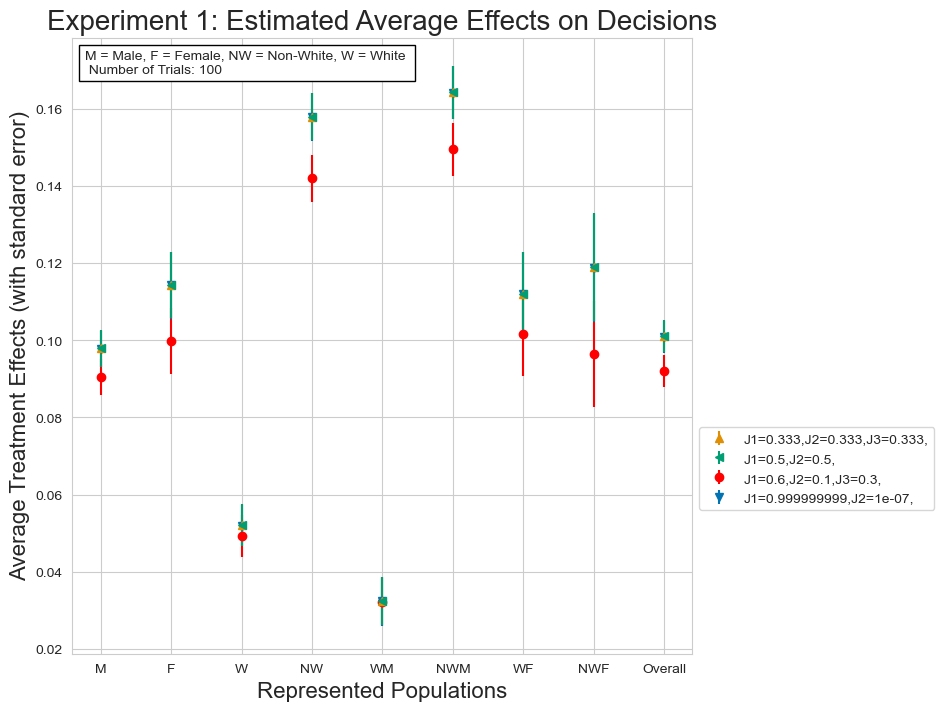}}%
\caption{\textbf{Experiment 1; Average Treatment Effect Changes on Decisions due to treatment exposure.} Judge assignments represent allocations of treated cases across all available simulated judges. So given a judge assignment of $J_1 = 0.6$, 60\% of available treated cases are assigned to $J_1$, giving them an overall treatment exposure of $\approx 90\%$ over all their cases  -- Non-Linear judge model.}
\label{fig:myfig_ex1_2}
\end{figure}

\begin{figure}[htbp]
\centering
\subfloat[Judge bias $b_k =0.6$, Exposure threshold $ t=0.4$]{\label{fig:a}\includegraphics[width=0.45\linewidth]{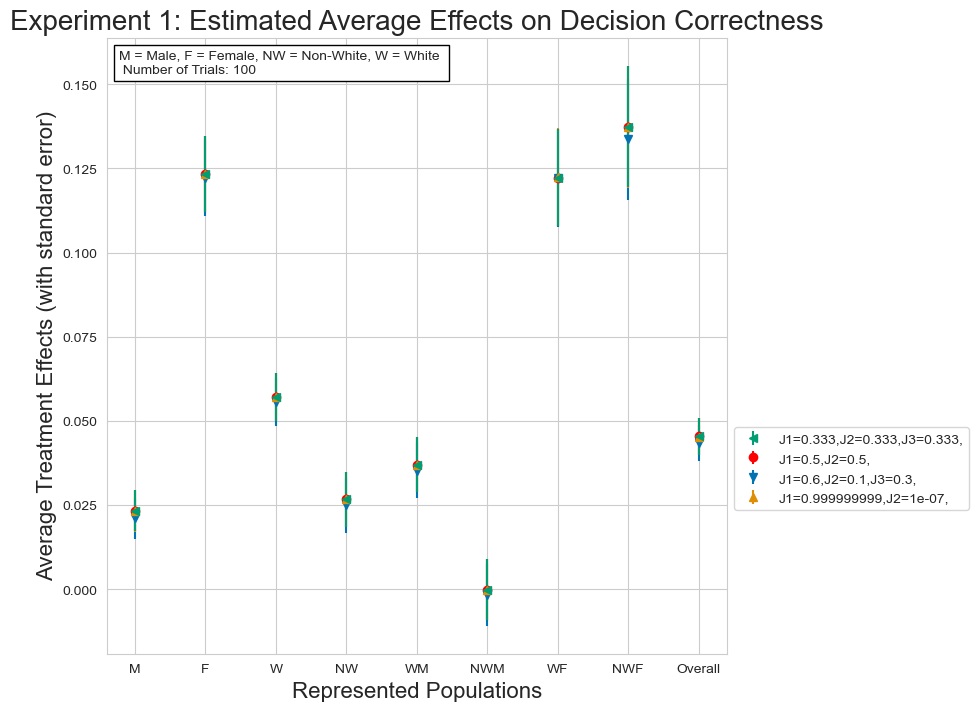}}\qquad
\subfloat[Judge bias $b_k =0.4$, Exposure threshold $ t=0.6$]{\label{fig:b}\includegraphics[width=0.45\linewidth]{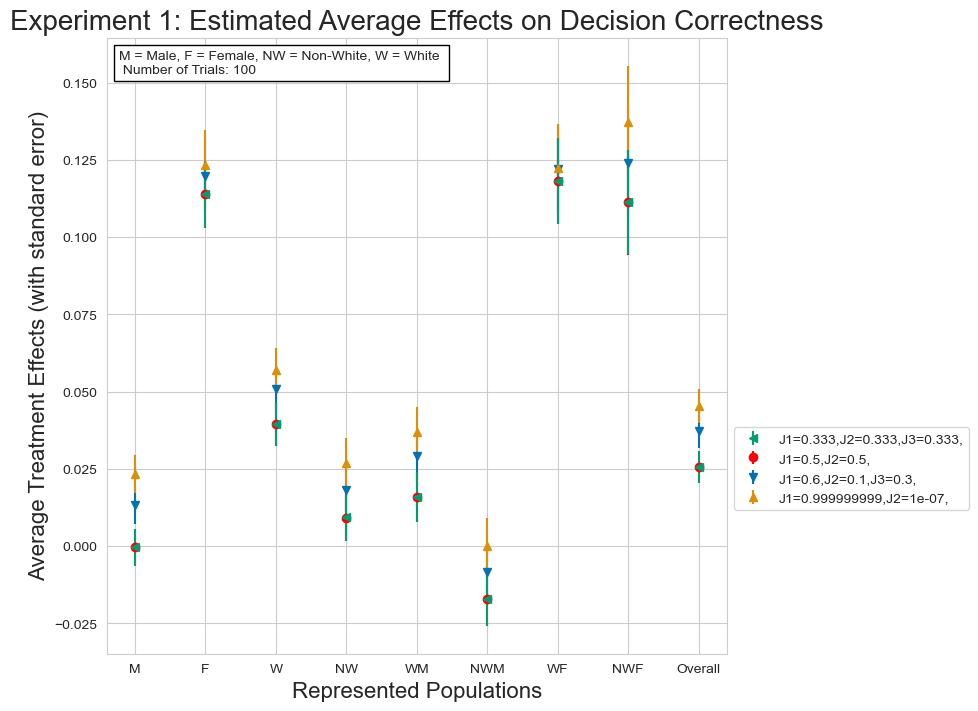}}\\
\subfloat[Judge bias $b_k =0.1$, Exposure threshold $t=0.6$]{\label{fig:c}\includegraphics[width=0.45\textwidth]{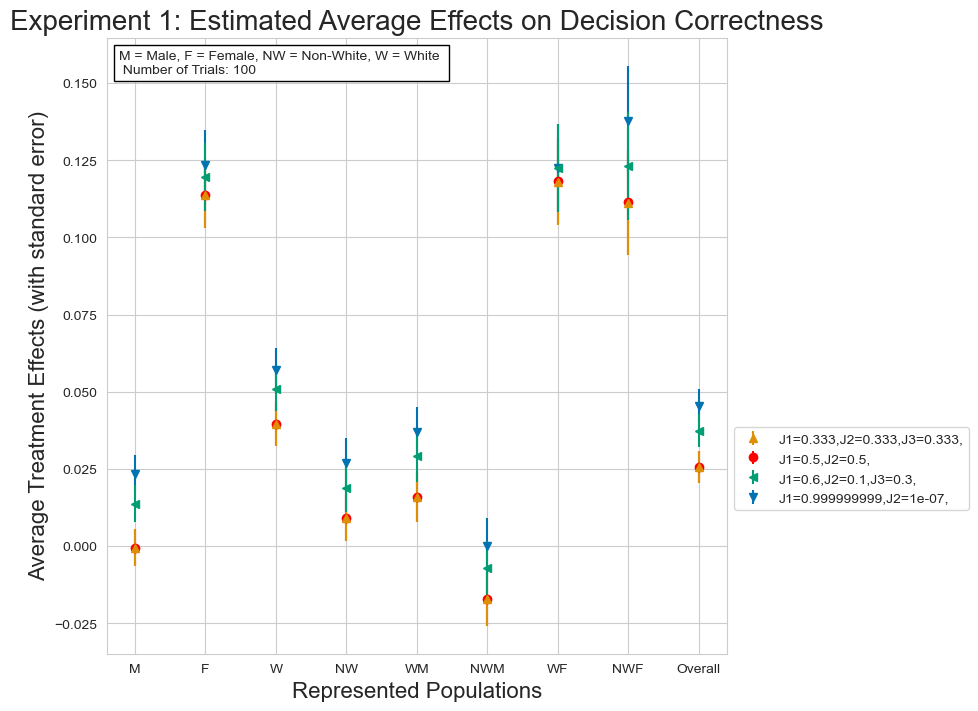}}\qquad%
\subfloat[Judge bias $b_k =0.1$, Exposure threshold $t=0.4$]{\label{fig:d}\includegraphics[width=0.45\textwidth]{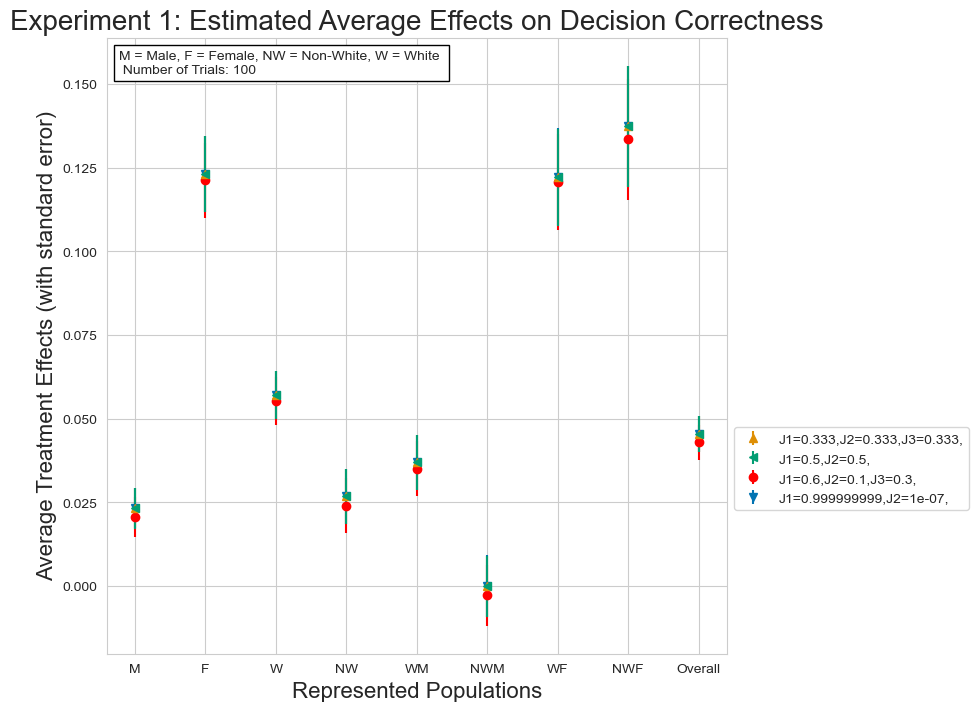}}%
\caption{\textbf{Experiment 1; Average Treatment Effect Changes on Decision Correctness due to average treatment exposure.} Judge assignments represent allocations of treated cases across all available simulated judges. So given a judge assignment of $J_1 = 0.6$, 60\% of available treated cases are assigned to $J_1$, giving them an overall treatment exposure of $\approx 90\%$ over all their cases  -- Non-Linear judge model.}
\label{fig:myfig_ex1_3}
\end{figure}

\begin{figure}[htbp]
\centering
\subfloat[Judge bias $b_k =0.6$, Exposure threshold $t=0.4$]{\label{fig:a}\includegraphics[width=0.45\linewidth]{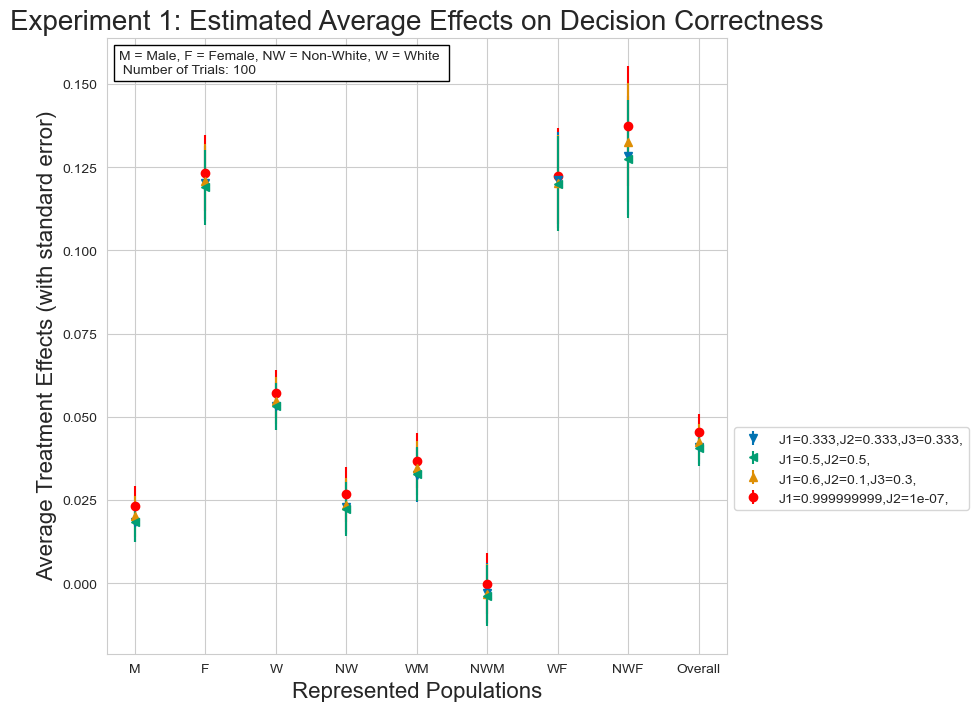}}\qquad
\subfloat[Judge bias $b_k =0.4$, Exposure threshold $t=0.6$]{\label{fig:b}\includegraphics[width=0.45\linewidth]{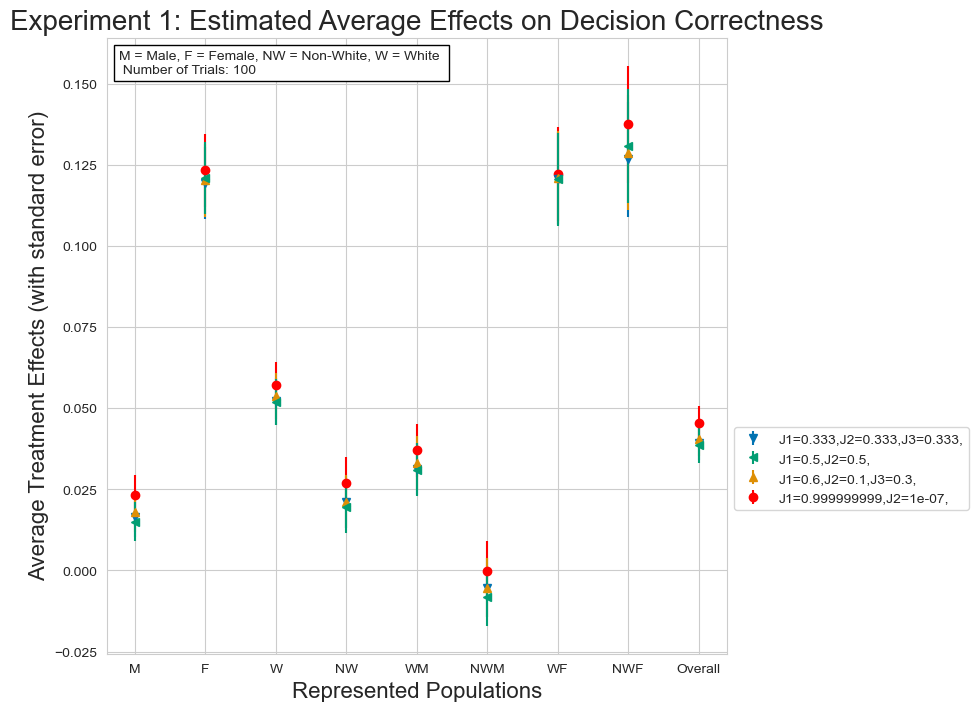}}\\
\subfloat[Judge bias $b_k =0.1$, Exposure threshold $t=0.6$]{\label{fig:c}\includegraphics[width=0.45\textwidth]{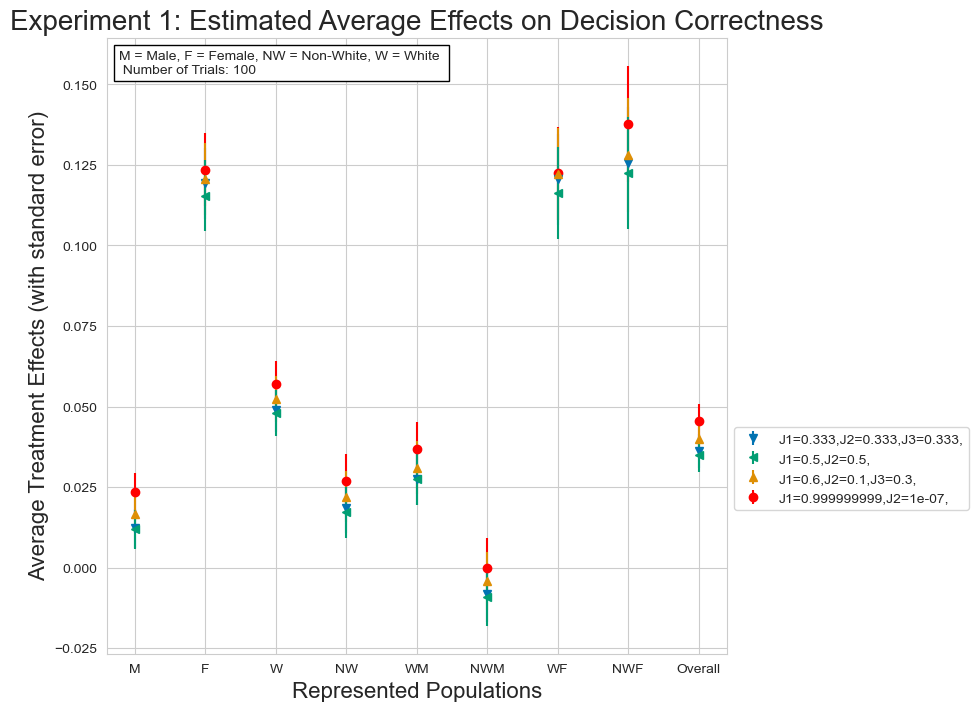}}\qquad%
\subfloat[ Judge bias $b_k =0.1$, Exposure threshold $t=0.4$]{\label{fig:d}\includegraphics[width=0.45\textwidth]{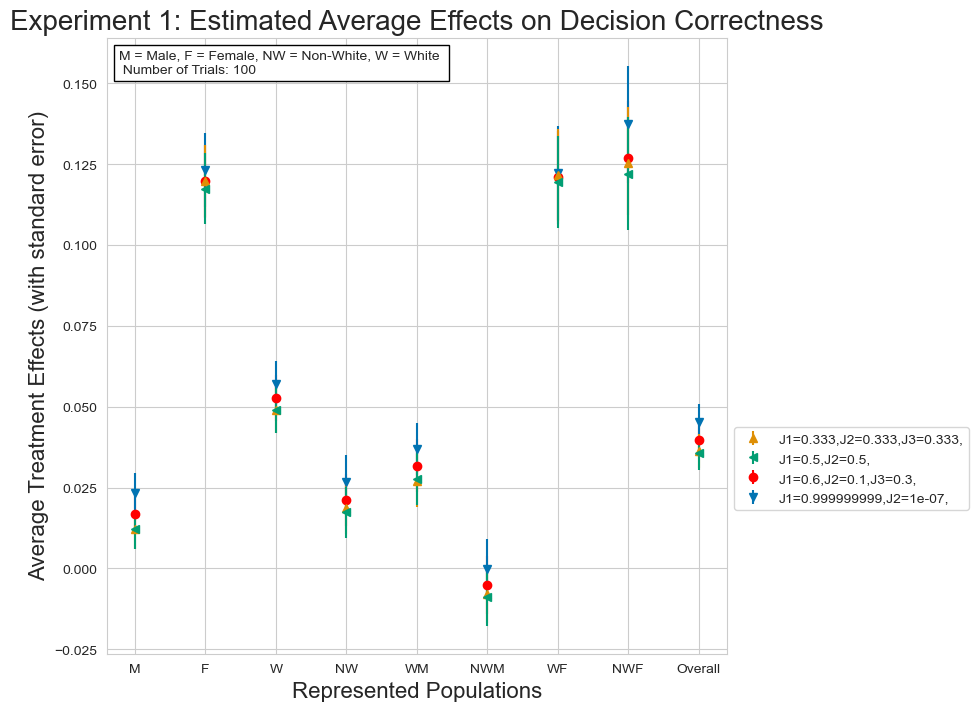}}%
\caption{\textbf{Experiment 1; Average Treatment Effect Changes on Decision Correctness due to average treatment exposure.} Judge assignments represent allocations of treated cases across all available simulated judges. So given a judge assignment of $J_1 = 0.6$, 60\% of available treated cases are assigned to $J_1$, giving them an overall treatment exposure of $\approx 90\%$ over all their cases -- Linear judge model. }
\label{fig:myfig_ex1_4}
\end{figure}


\subsection{Experiment 2: Capacity Constraint effect}

 When the judge is over-exposed to positive recommendations, our capacity constraint model concludes that the judge then becomes less responsive. Thus, to model the capacity constraint setting, we set arbitrary thresholds on the model recommendation, in order to then modify the positive predictive rate of the algorithm. These results are over a simulation with 100 trials, and the data we have available allows decision thresholds to be set at a level of 3,4,5 or 6.

Given the same treatment exposure, we find that there are average treatment effect measurement differences under different algorithm decision thresholds. Under our judge model, this is induced by changes to the positive predictive rate of the algorithm, which in turn impacts judge compliance through the capacity constraint. For example, a lower algorithmic decision threshold (eg. threshold = 3) means that the algorithm predicts detainment at a higher frequency and the measured treatment effect on such decisions seem to be higher than the effect measured for higher thresholds. Interestingly, this effect does not meaningful impact the measure of changes to decision correctness.

Main results for this experiment can be found in in Figures ~\ref{fig:myfig_exp2_1}, ~\ref{fig:myfig_exp2_2}, ~\ref{fig:myfig_exp2_3}, ~\ref{fig:myfig_exp2_4}, ~\ref{fig:myfig_exp2_5}, ~\ref{fig:myfig_exp2_6}, ~\ref{fig:myfig_exp2_7}, ~\ref{fig:myfig_exp2_8}. 

\begin{figure}[htbp]
\centering
\subfloat[Judge Assignment: $J_1=0.6$, $J_2=0.3$, $J_3=0.1$]{\label{fig:a}\includegraphics[width=0.45\linewidth]{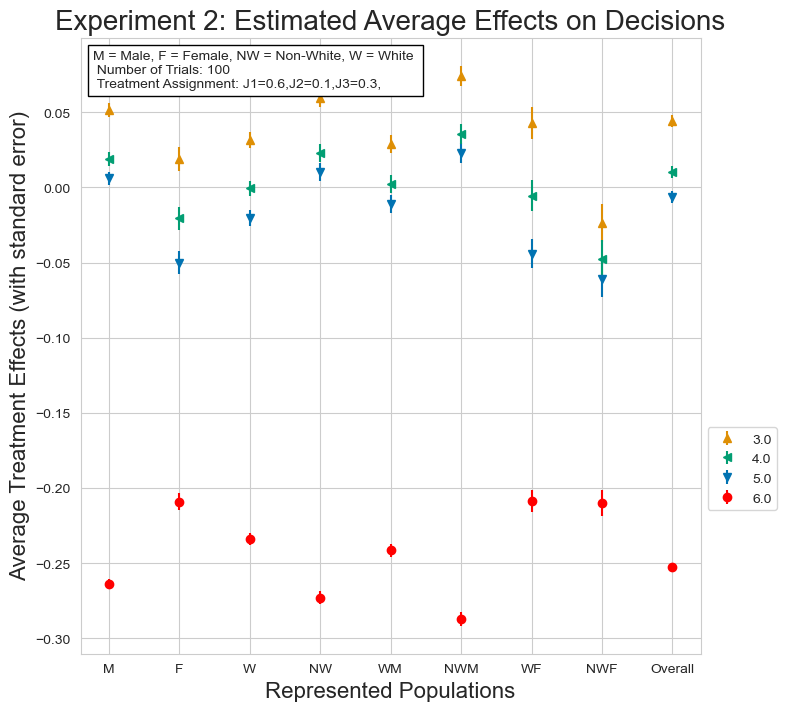}}\qquad
\subfloat[Judge Assignment: $J_1=0.5$, $J_2=0.5$]{\label{fig:b}\includegraphics[width=0.45\linewidth]{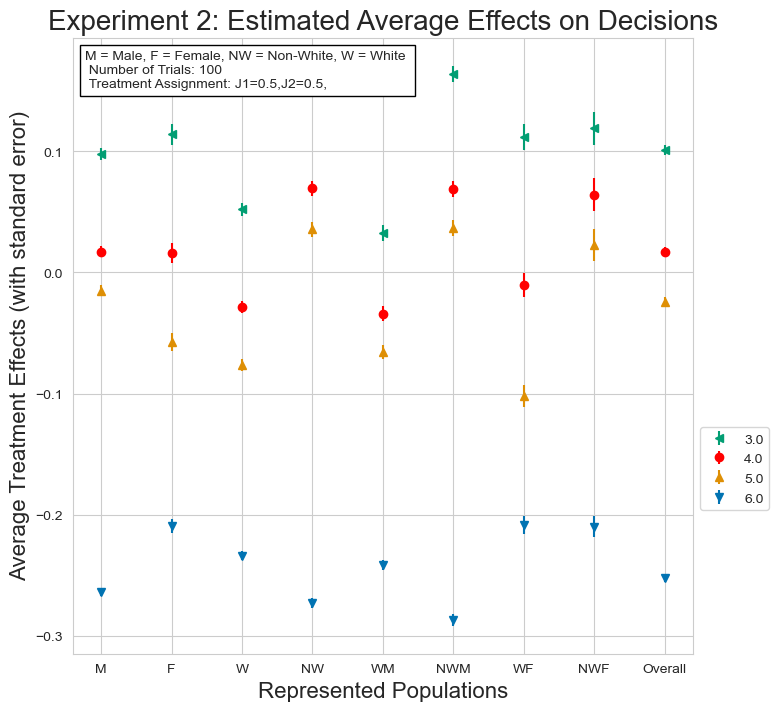}}\\
\subfloat[Judge Assignment: $J_1=0.33$, $J_2=0.33$, $J_3=0.33$]{\label{fig:c}\includegraphics[width=0.45\textwidth]{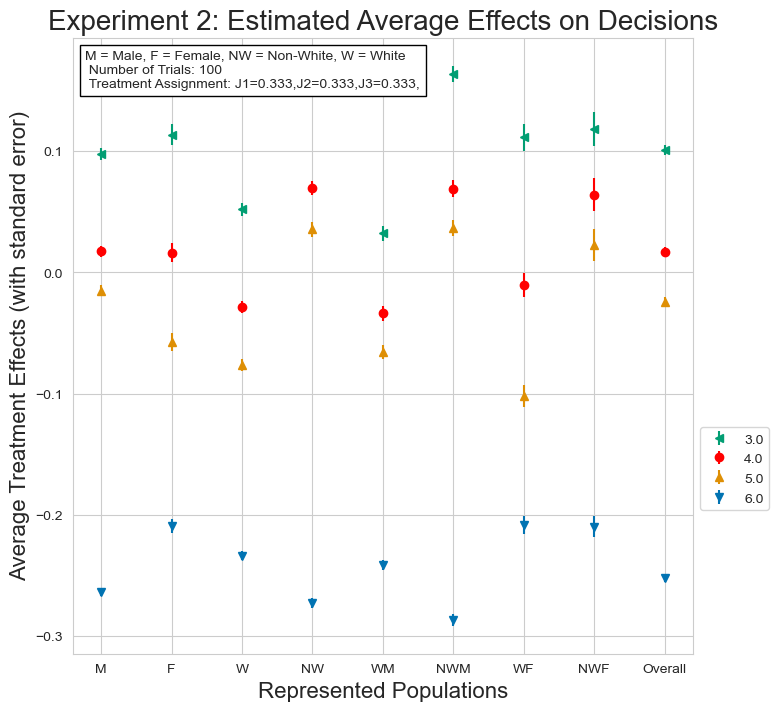}}\qquad%
\subfloat[Judge Assignment: $J_1=0.99$, $J_2=0.01$]{\label{fig:d}\includegraphics[width=0.45\textwidth]{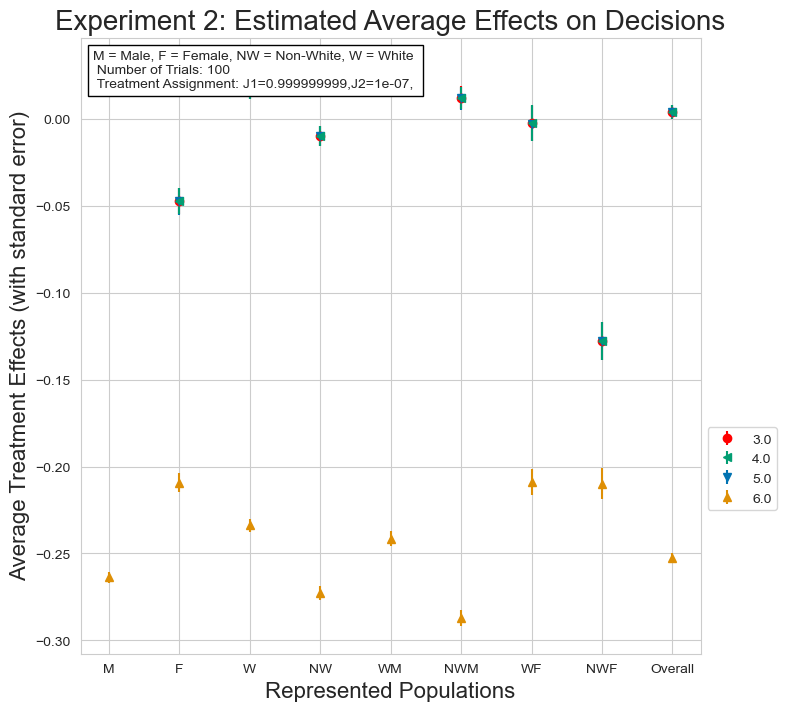}}%
\caption{\textbf{Experiment 2; Average Treatment Effect Changes on Decisions due to modifications in decision-making threshold.} We can set decision thresholds at a level of 3,4,5 or 6 in order to modify the positive predictive rate of the model; the higher the threshold required to make a positive prediction, the lower the model's positive predictive rate, and the higher the probability that positive predictions ($\hat{Y} = 1$) stays below the capacity threshold -- Non-Linear judge model,  baseline judge bias $b_k =0.6$, capacity threshold (of max positive predictive rate) $t=0.2$}
\label{fig:myfig_exp2_1}
\end{figure}

\begin{figure}[htbp]
\centering
\subfloat[Judge Assignment: $J_1=0.6$, $J_2=0.3$, $J_3=0.1$]{\label{fig:a}\includegraphics[width=0.45\linewidth]{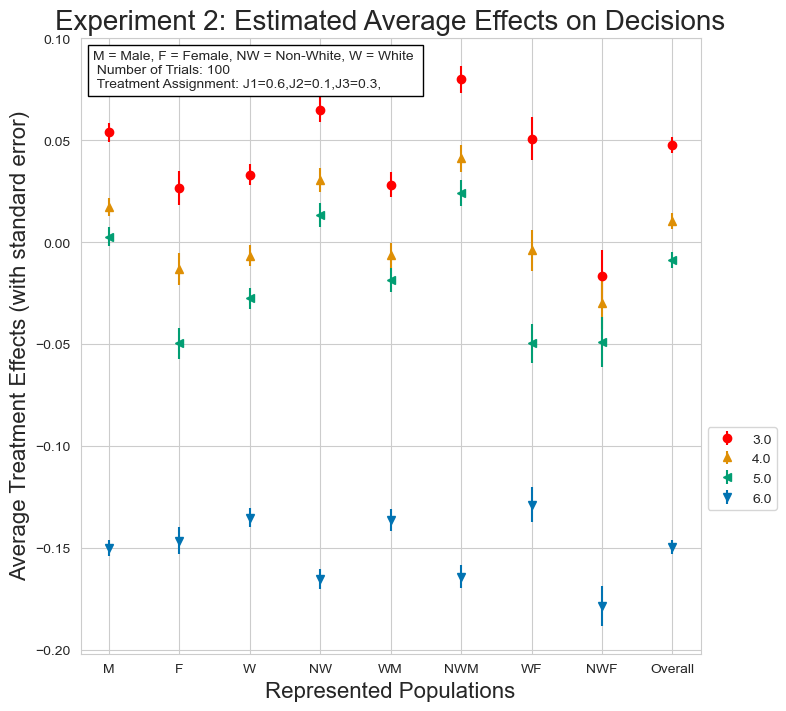}}\qquad
\subfloat[Judge Assignment: $J_1=0.5$, $J_2=0.5$]{\label{fig:b}\includegraphics[width=0.45\linewidth]{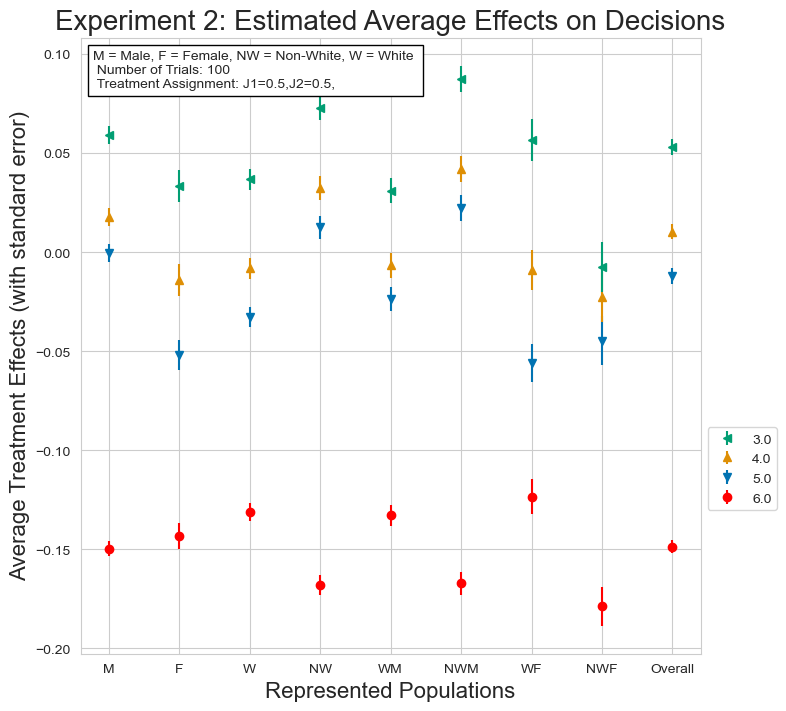}}\\
\subfloat[Judge Assignment: $J_1=0.33$, $J_2=0.33$, $J_3=0.33$]{\label{fig:c}\includegraphics[width=0.45\textwidth]{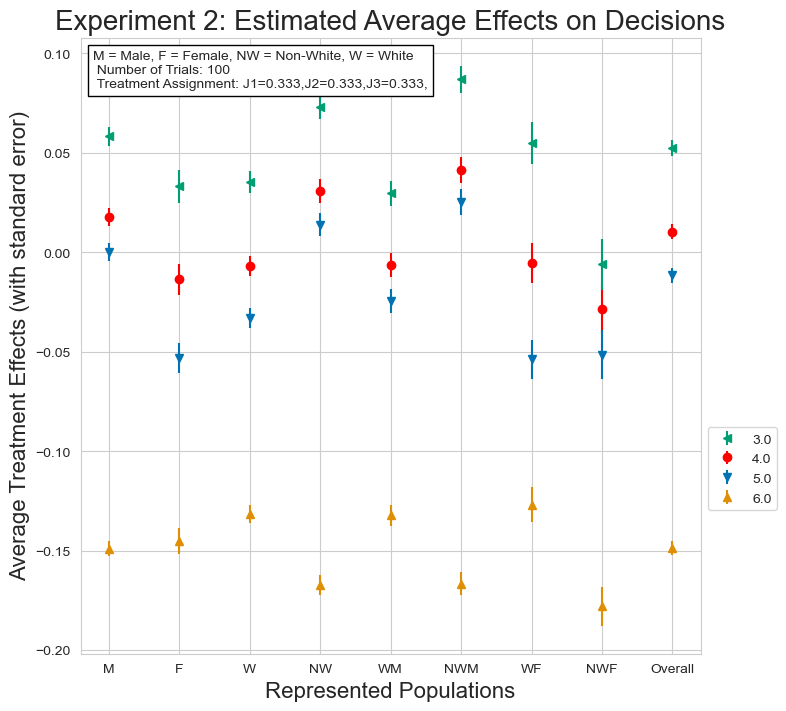}}\qquad%
\subfloat[Judge Assignment: $J_1=0.99$, $J_2=0.01$]{\label{fig:d}\includegraphics[width=0.45\textwidth]{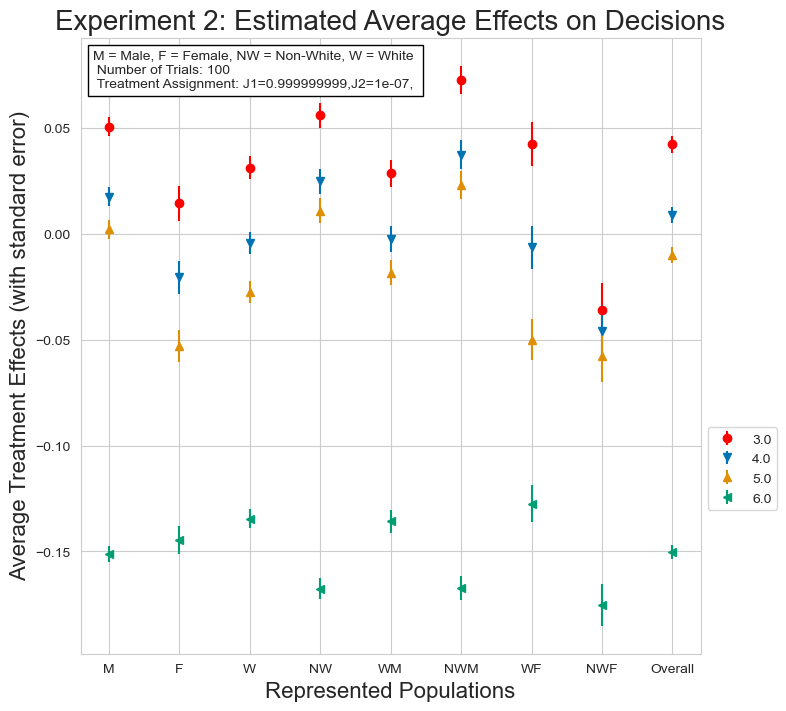}}%
\caption{\textbf{Experiment 2; Average Treatment Effect Changes on Decisions due to modifications in decision-making threshold.} We can set decision thresholds at a level of 3,4,5 or 6 in order to modify the positive predictive rate of the model; the higher the threshold required to make a positive prediction, the lower the model's positive predictive rate, and the higher the probability that positive predictions ($\hat{Y} = 1$) stays below the capacity threshold -- Linear judge model,  baseline judge bias $b_k =0.6$, capacity threshold (of max positive predictive rate) $t=0.2$}
\label{fig:myfig_exp2_2}
\end{figure}

\begin{figure}[htbp]
\centering
\subfloat[Judge Assignment: $J_1=0.6$, $J_2=0.3$, $J_3=0.1$]{\label{fig:a}\includegraphics[width=0.45\linewidth]{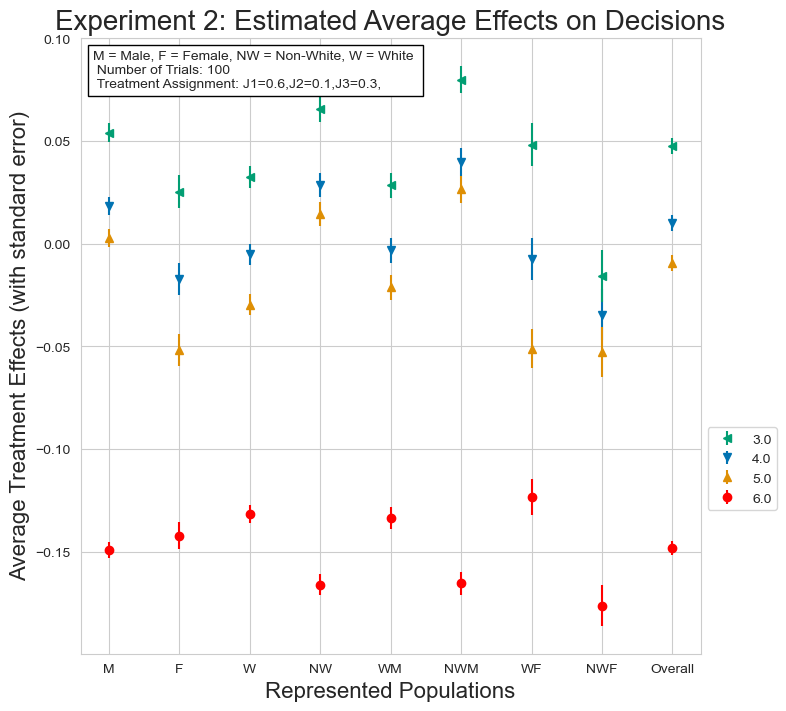}}\qquad
\subfloat[Judge Assignment: $J_1=0.5$, $J_2=0.5$]{\label{fig:b}\includegraphics[width=0.45\linewidth]{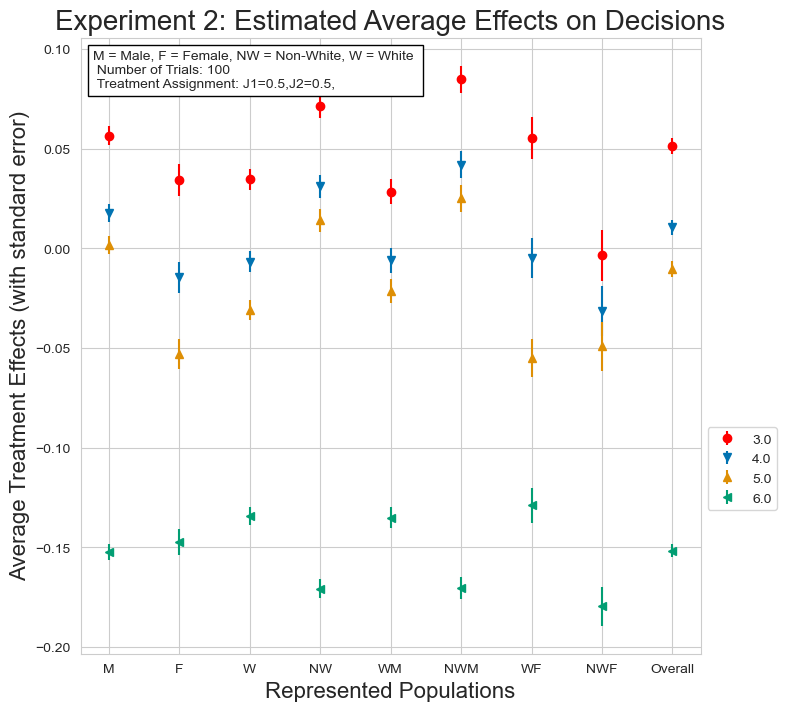}}\\
\subfloat[Judge Assignment: $J_1=0.33$, $J_2=0.33$, $J_3=0.33$]{\label{fig:c}\includegraphics[width=0.45\textwidth]{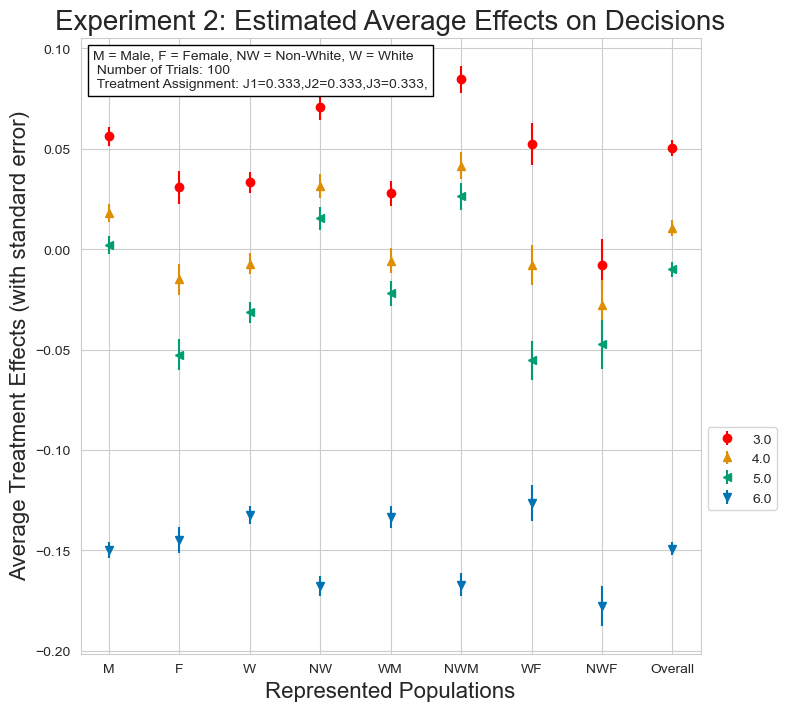}}\qquad%
\subfloat[Judge Assignment: $J_1=0.99$, $J_2=0.01$]{\label{fig:d}\includegraphics[width=0.45\textwidth]{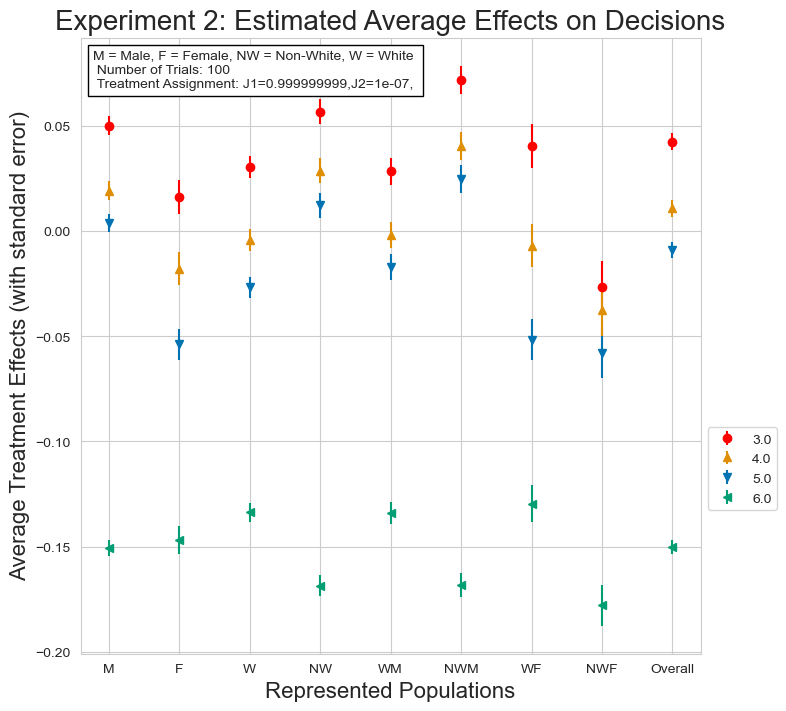}}%
\caption{\textbf{Experiment 2; Average Treatment Effect Changes on Decisions due to modifications in decision-making threshold.} We can set decision thresholds at a level of 3,4,5 or 6 in order to modify the positive predictive rate of the model; the higher the threshold required to make a positive prediction, the lower the model's positive predictive rate, and the higher the probability that positive predictions ($\hat{Y} = 1$) stays below the capacity threshold -- Linear judge model,  baseline judge bias $b_k =0.6$, capacity threshold (of max positive predictive rate) $t=0.15$}
\label{fig:myfig_exp2_3}
\end{figure}

\begin{figure}[htbp]
\centering
\subfloat[Judge Assignment: $J_1=0.6$, $J_2=0.3$, $J_3=0.1$]{\label{fig:a}\includegraphics[width=0.45\linewidth]{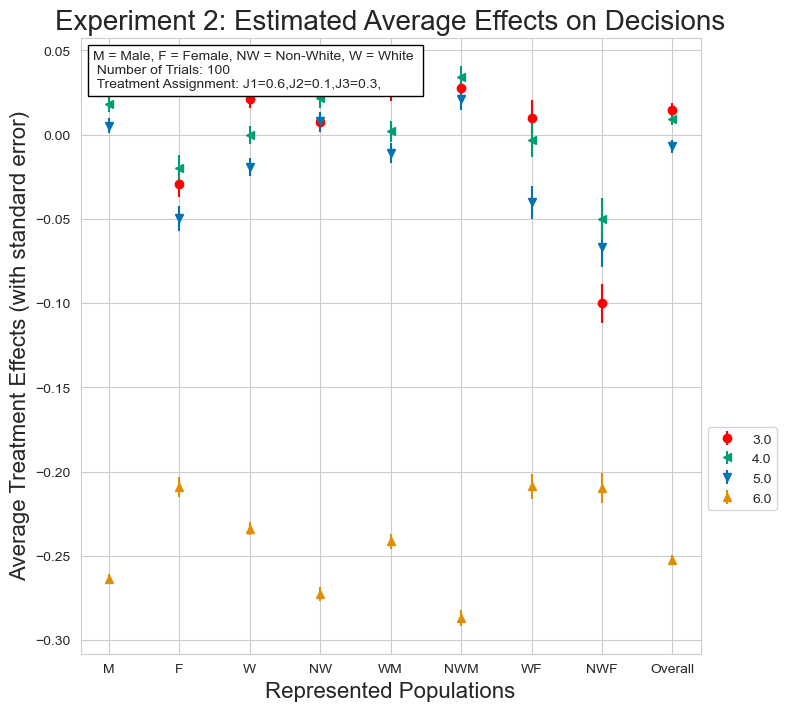}}\qquad
\subfloat[Judge Assignment: $J_1=0.5$, $J_2=0.5$]{\label{fig:b}\includegraphics[width=0.45\linewidth]{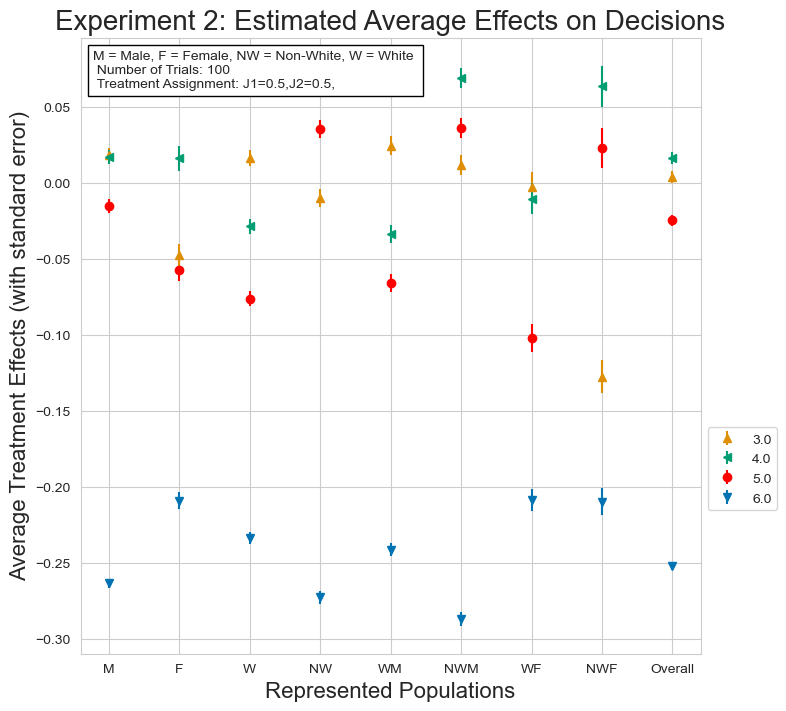}}\\
\subfloat[Judge Assignment: $J_1=0.33$, $J_2=0.33$, $J_3=0.33$]{\label{fig:c}\includegraphics[width=0.45\textwidth]{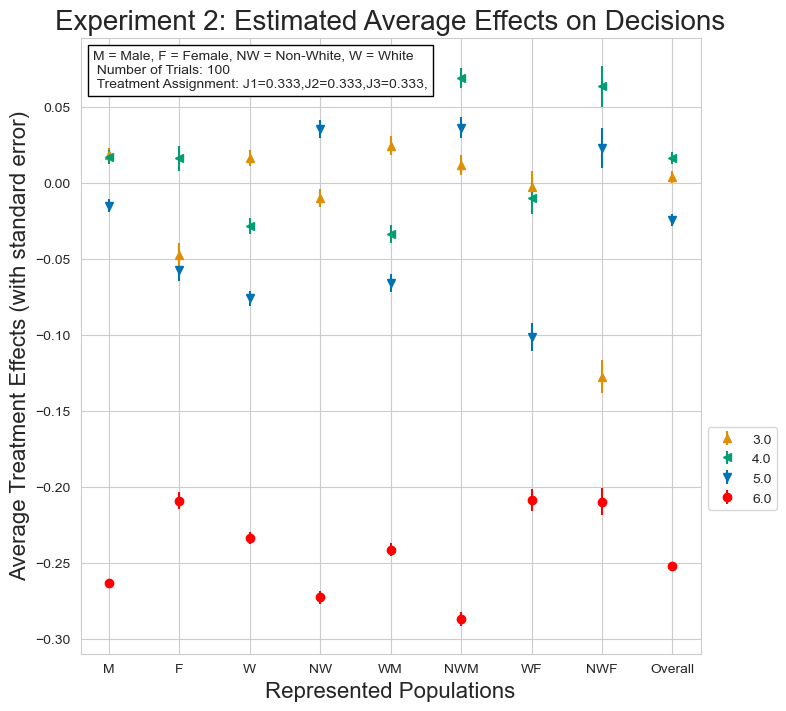}}\qquad%
\subfloat[Judge Assignment: $J_1=0.99$, $J_2=0.01$]{\label{fig:d}\includegraphics[width=0.45\textwidth]{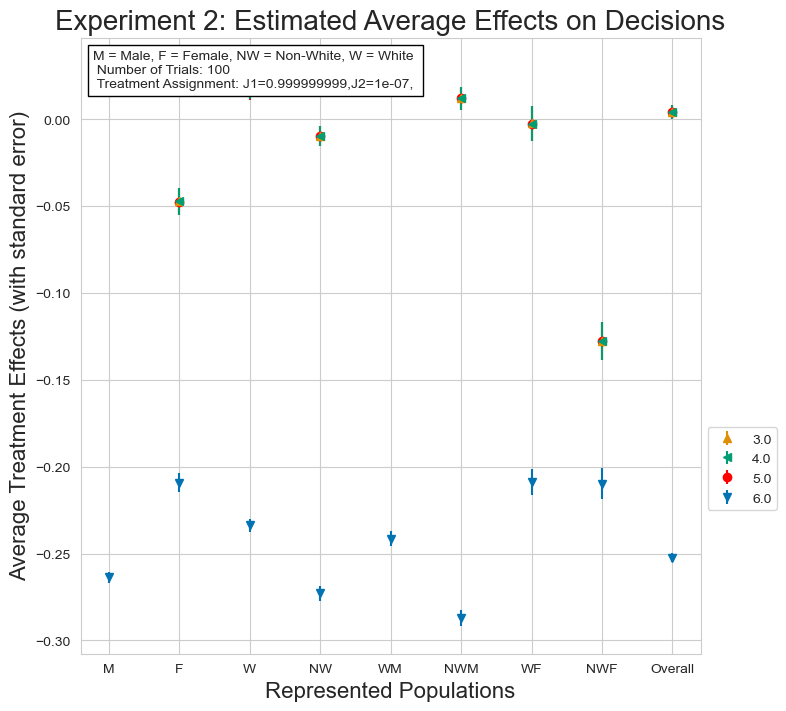}}%
\caption{\textbf{Experiment 2; Average Treatment Effect Changes on Decisions due to modifications in decision-making threshold.} We can set decision thresholds at a level of 3,4,5 or 6 in order to modify the positive predictive rate of the model; the higher the threshold required to make a positive prediction, the lower the model's positive predictive rate, and the higher the probability that positive predictions ($\hat{Y} = 1$) stays below the capacity threshold -- Non-Linear judge model,  baseline judge bias $b_k =0.6$, capacity threshold (of max positive predictive rate) $t=0.15$}
\label{fig:myfig_exp2_4}
\end{figure}

\begin{figure}[htbp]
\centering
\subfloat[Judge Assignment: $J_1=0.6$, $J_2=0.3$, $J_3=0.1$]{\label{fig:a}\includegraphics[width=0.45\linewidth]{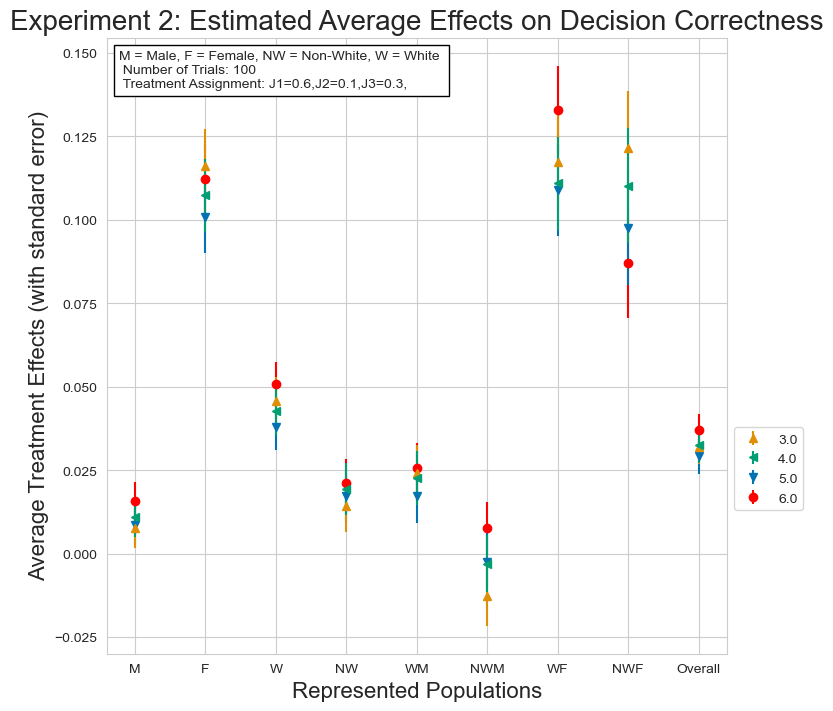}}\qquad
\subfloat[Judge Assignment: $J_1=0.5$, $J_2=0.5$]{\label{fig:b}\includegraphics[width=0.45\linewidth]{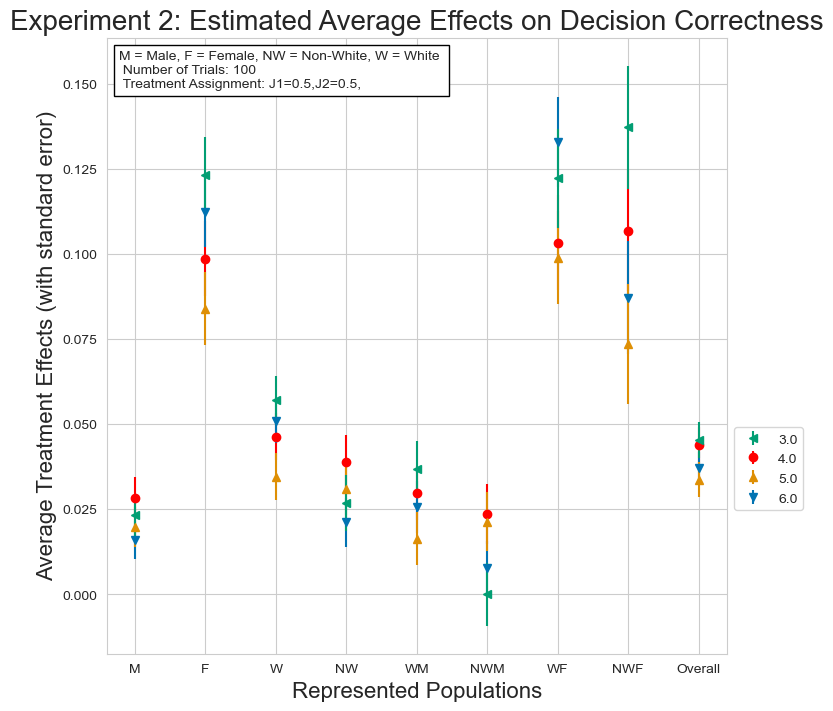}}\\
\subfloat[Judge Assignment: $J_1=0.33$, $J_2=0.33$, $J_3=0.33$]{\label{fig:c}\includegraphics[width=0.45\textwidth]{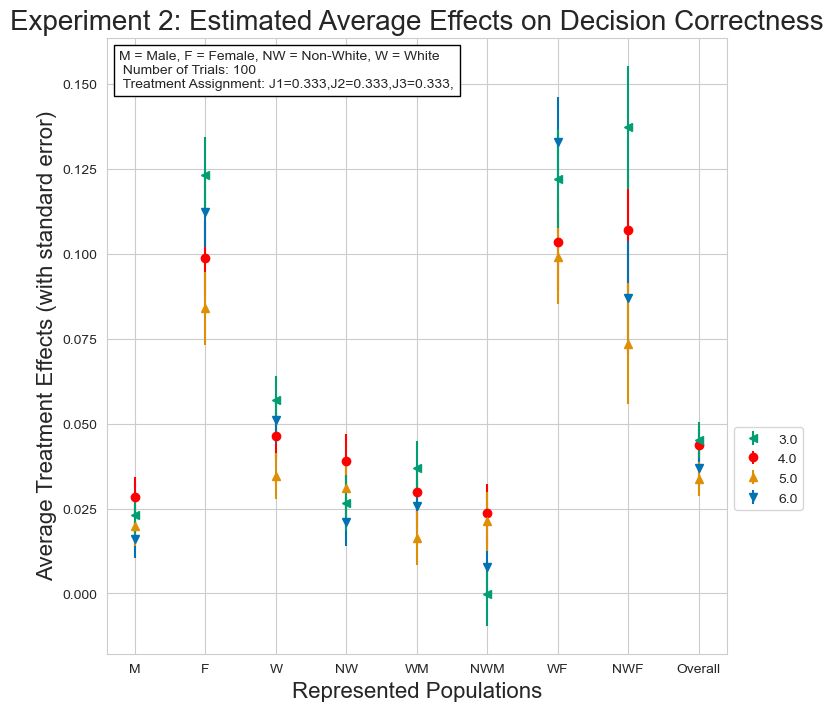}}\qquad%
\subfloat[Judge Assignment: $J_1=0.99$, $J_2=0.01$]{\label{fig:d}\includegraphics[width=0.45\textwidth]{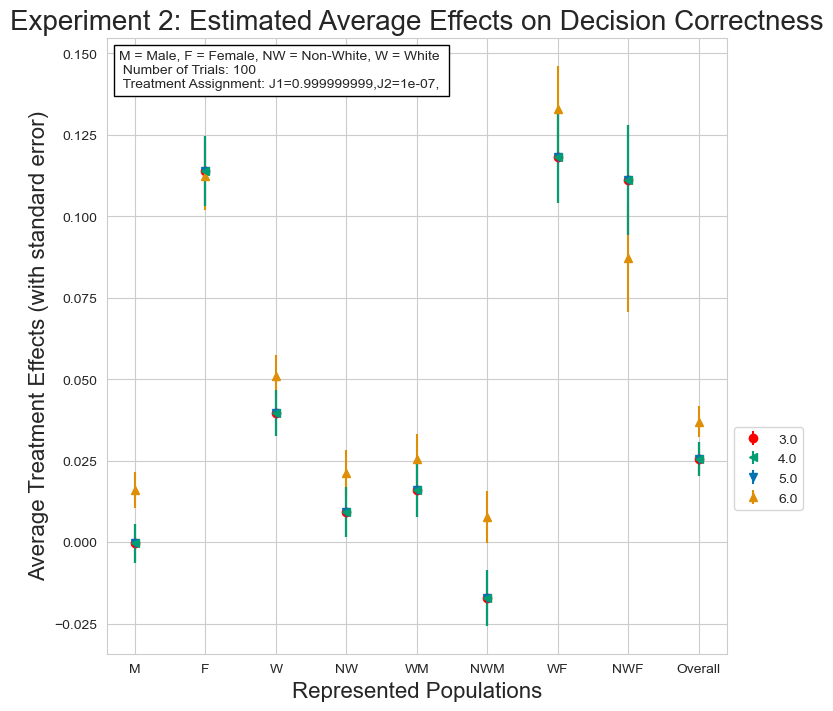}}%
\caption{\textbf{Experiment 2; Average Treatment Effect Changes on Decision Correctness due to modifications in decision-making threshold.} We can set decision thresholds at a level of 3,4,5 or 6 in order to modify the positive predictive rate of the model; the higher the threshold required to make a positive prediction, the lower the model's positive predictive rate, and the higher the probability that positive predictions ($\hat{Y} = 1$) stays below the capacity threshold -- Non-Linear judge model,  baseline judge bias $b_k =0.6$, capacity threshold (of max positive predictive rate) $t=0.2$}
\label{fig:myfig_exp2_5}
\end{figure}

\begin{figure}[htbp]
\centering
\subfloat[Judge Assignment: $J_1=0.6$, $J_2=0.3$, $J_3=0.1$]{\label{fig:a}\includegraphics[width=0.45\linewidth]{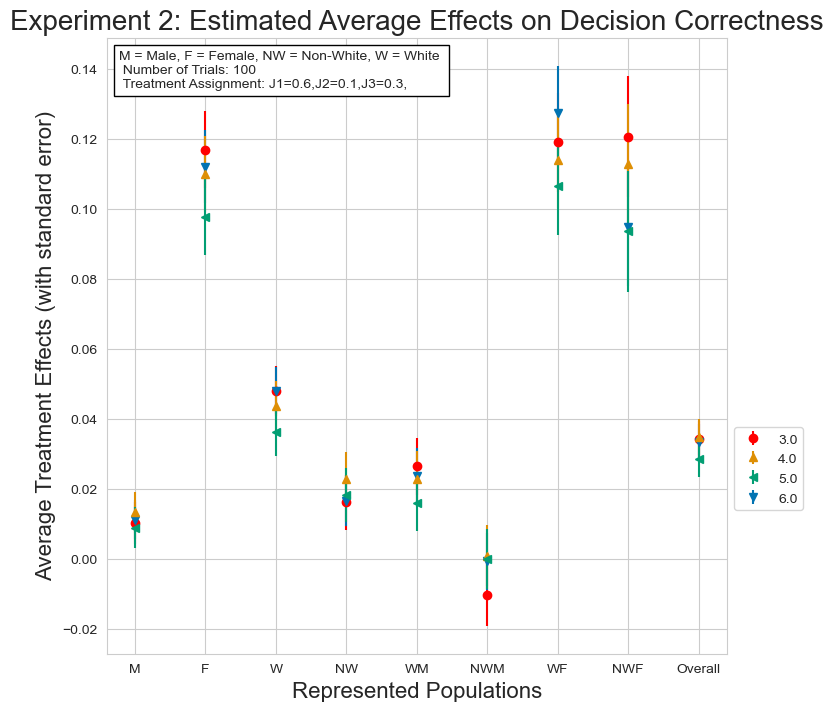}}\qquad
\subfloat[Judge Assignment: $J_1=0.5$, $J_2=0.5$]{\label{fig:b}\includegraphics[width=0.45\linewidth]{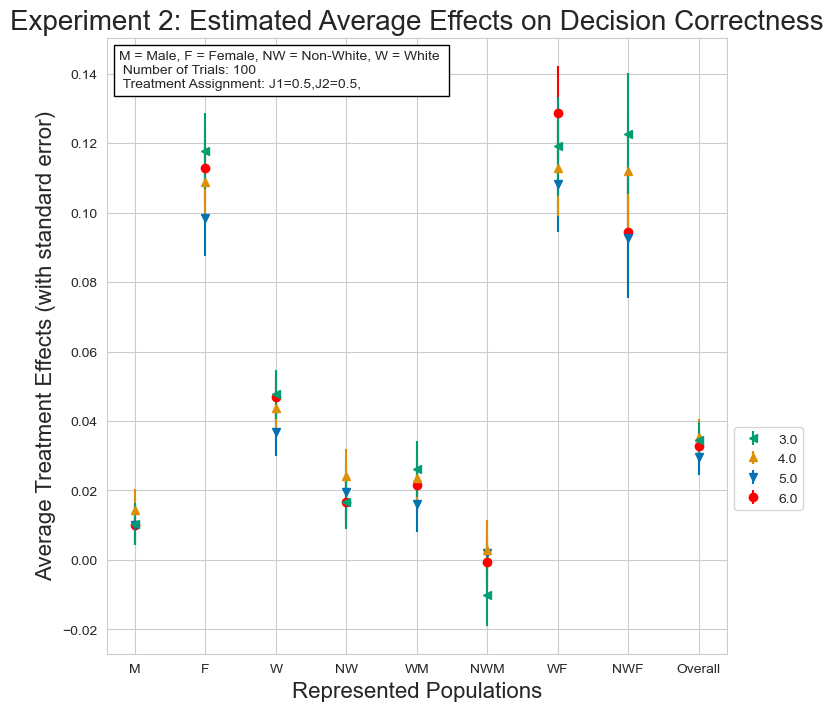}}\\
\subfloat[Judge Assignment: $J_1=0.33$, $J_2=0.33$, $J_3=0.33$]{\label{fig:c}\includegraphics[width=0.45\textwidth]{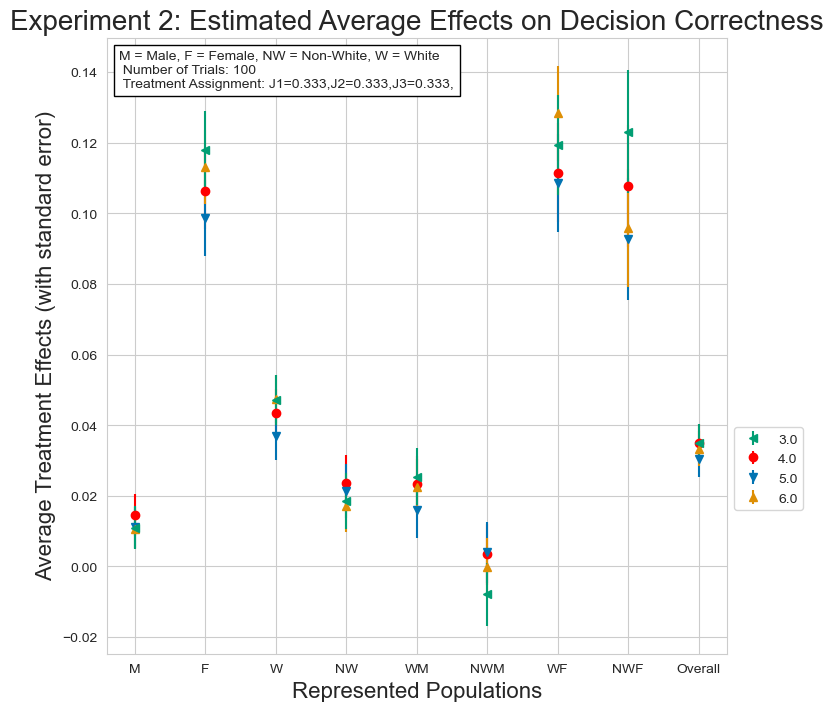}}\qquad%
\subfloat[Judge Assignment: $J_1=0.99$, $J_2=0.01$]{\label{fig:d}\includegraphics[width=0.45\textwidth]{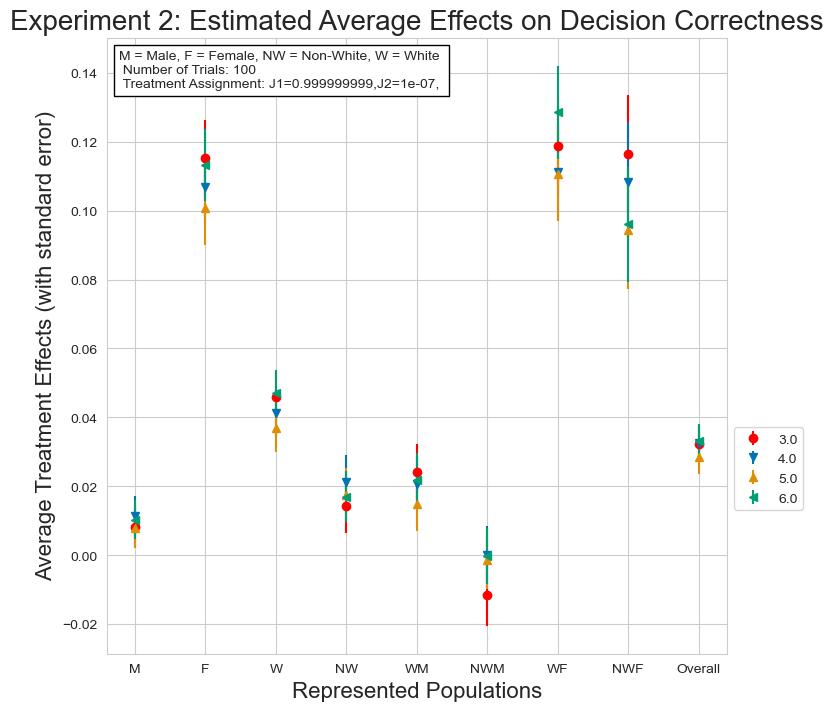}}%
\caption{\textbf{Experiment 2; Average Treatment Effect Changes on Decision Correctness due to modifications in decision-making threshold.} We can set decision thresholds at a level of 3, 4, 5 or 6 in order to modify the positive predictive rate of the model; the higher the threshold required to make a positive prediction, the lower the model's positive predictive rate, and the higher the probability that positive predictions ($\hat{Y} = 1$) stays below the capacity threshold -- Linear judge model,  baseline judge bias $b_k =0.6$, capacity threshold (of max positive predictive rate) $t=0.2$}
\label{fig:myfig_exp2_6}
\end{figure}

\begin{figure}[htbp]
\centering
\subfloat[Judge Assignment: $J_1=0.6$, $J_2=0.3$, $J_3=0.1$]{\label{fig:a}\includegraphics[width=0.45\linewidth]{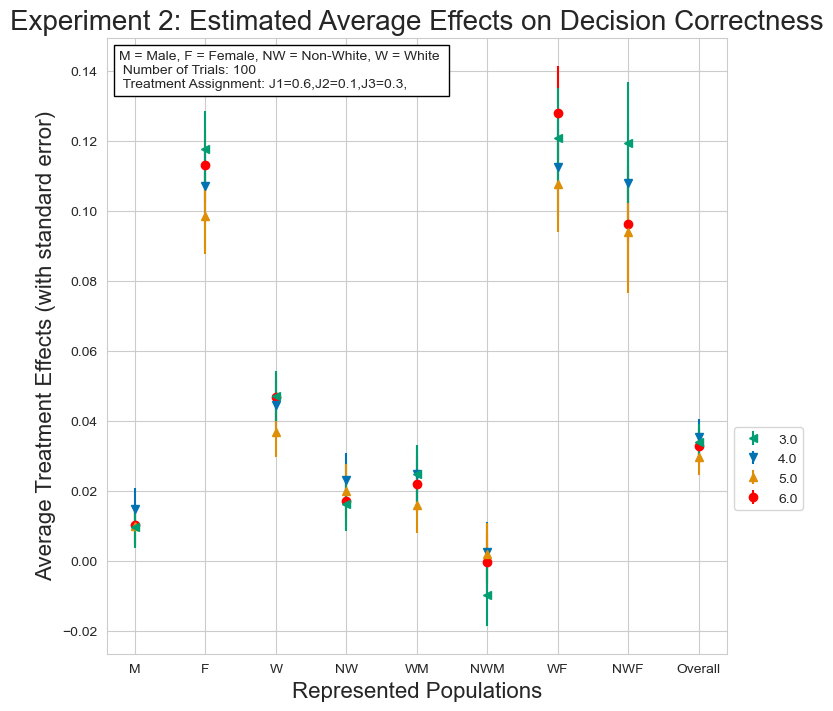}}\qquad
\subfloat[Judge Assignment: $J_1=0.5$, $J_2=0.5$]{\label{fig:b}\includegraphics[width=0.45\linewidth]{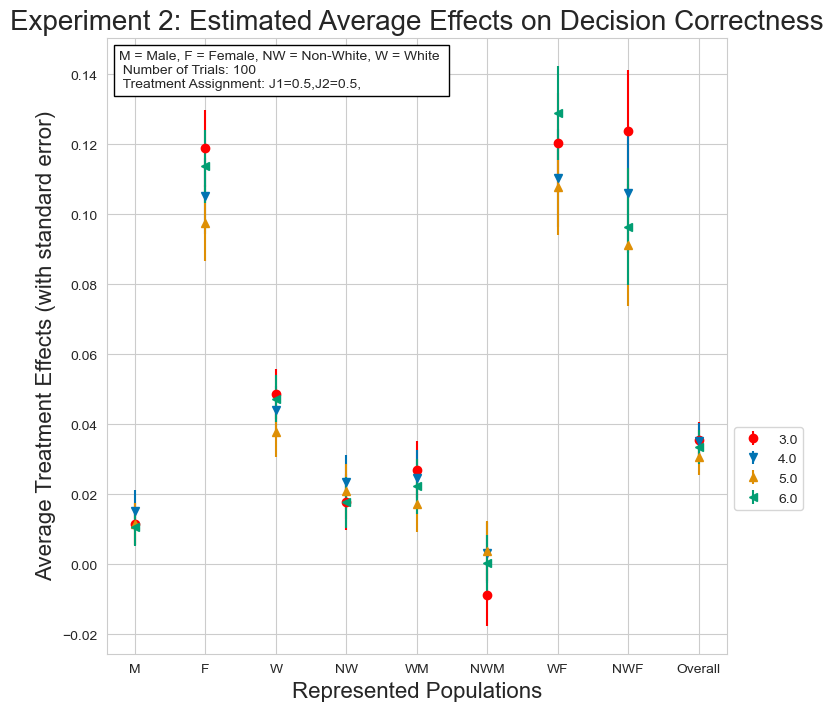}}\\
\subfloat[Judge Assignment: $J_1=0.33$, $J_2=0.33$, $J_3=0.33$]{\label{fig:c}\includegraphics[width=0.45\textwidth]{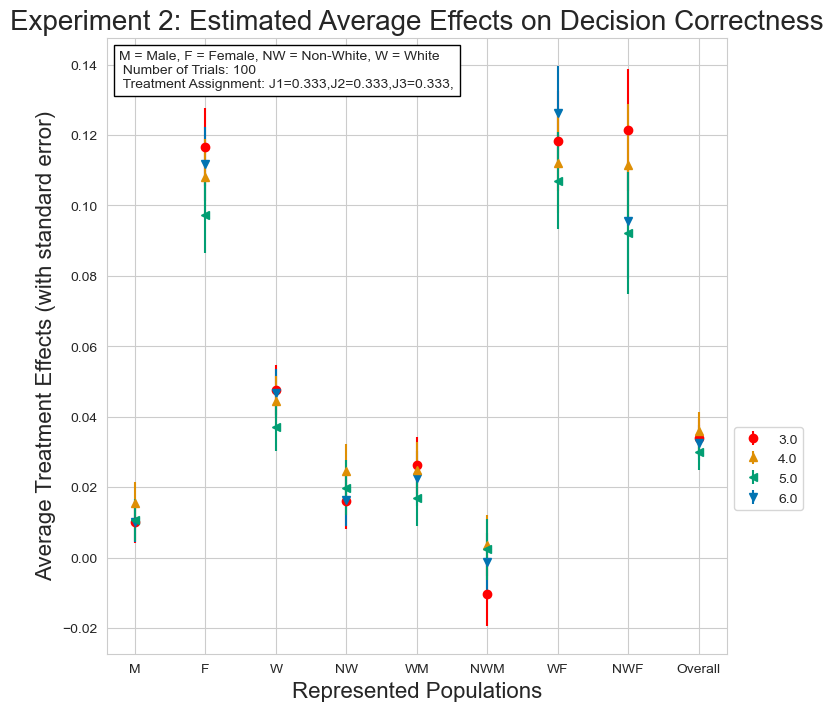}}\qquad%
\subfloat[Judge Assignment: $J_1=0.99$, $J_2=0.01$]{\label{fig:d}\includegraphics[width=0.45\textwidth]{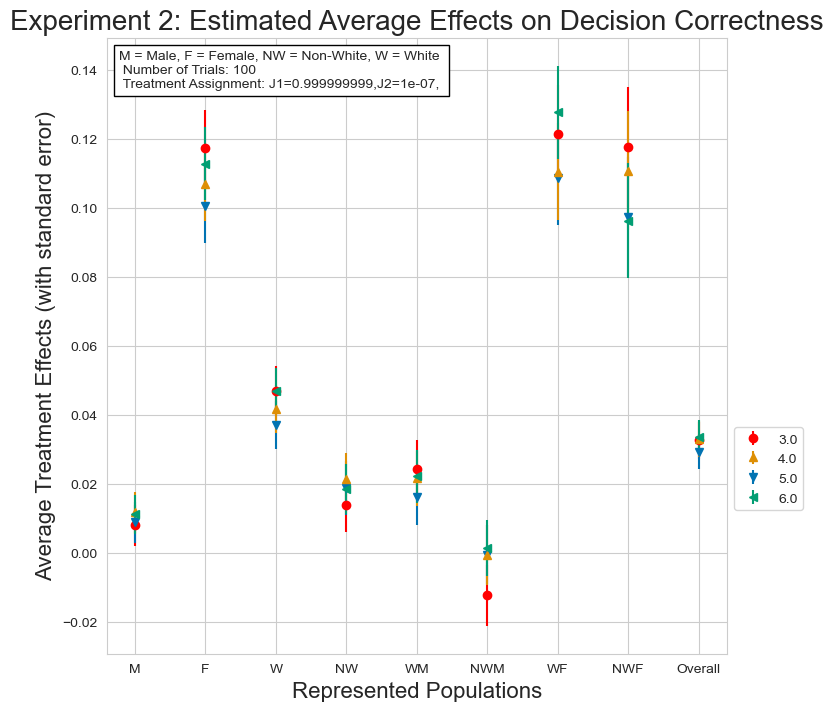}}%
\caption{\textbf{Experiment 2; Average Treatment Effect Changes on Decision Correctness due to modifications in decision-making threshold.} We can set decision thresholds at a level of 3,4,5 or 6 in order to modify the positive predictive rate of the model; the higher the threshold required to make a positive prediction, the lower the model's positive predictive rate, and the higher the probability that positive predictions ($\hat{Y} = 1$) stays below the capacity threshold -- Linear judge model, baseline judge bias $b_k =0.6$, capacity threshold (of max positive predictive rate)  $t=0.15$}
\label{fig:myfig_exp2_7}
\end{figure}

\begin{figure}[htbp]
\centering
\subfloat[Judge Assignment: $J_1=0.6$, $J_2=0.3$, $J_3=0.1$]{\label{fig:a}\includegraphics[width=0.45\linewidth]{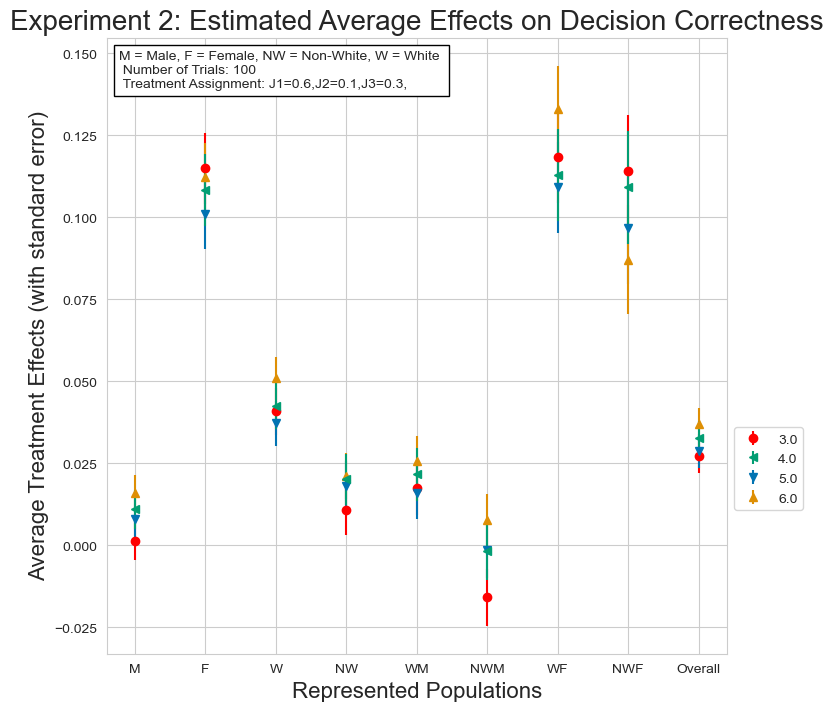}}\qquad
\subfloat[Judge Assignment: $J_1=0.5$, $J_2=0.5$]{\label{fig:b}\includegraphics[width=0.45\linewidth]{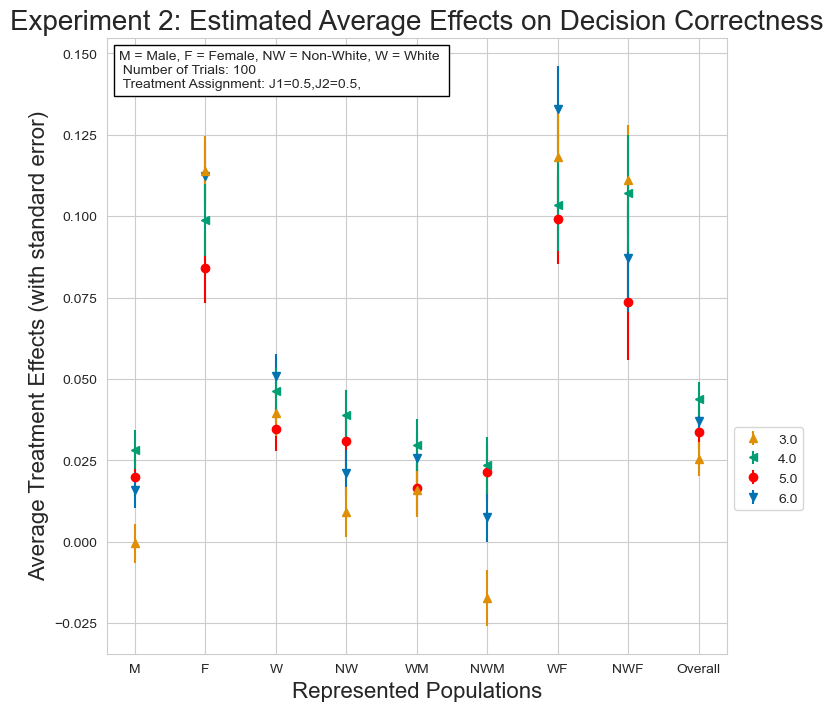}}\\
\subfloat[Judge Assignment: $J_1=0.33$, $J_2=0.33$, $J_3=0.33$]{\label{fig:c}\includegraphics[width=0.45\textwidth]{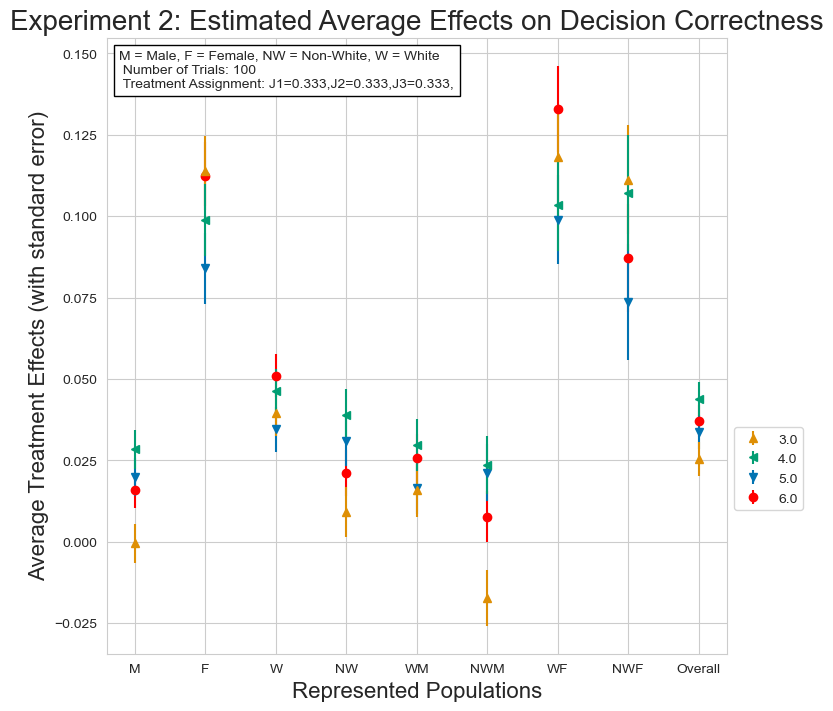}}\qquad%
\subfloat[Judge Assignment: $J_1=0.99$, $J_2=0.01$]{\label{fig:d}\includegraphics[width=0.45\textwidth]{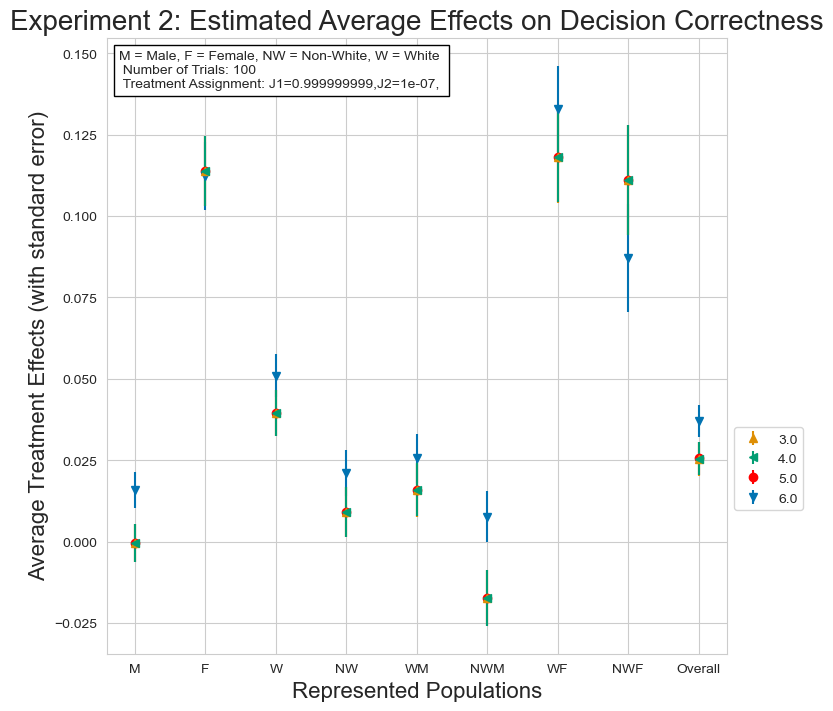}}%
\caption{\textbf{Experiment 2; Average Treatment Effect Changes on Decision Correctness due to modifications in decision-making threshold.} We can set decision thresholds at a level of 3,4,5 or 6 in order to modify the positive predictive rate of the model; the higher the threshold required to make a positive prediction, the lower the model's positive predictive rate, and the higher the probability that positive predictions ($\hat{Y} = 1$) stays below the capacity threshold -- Non-Linear.  baseline judge bias $b_k =0.6$, capacity threshold (of max positive predictive rate) $t=0.15$}
\label{fig:myfig_exp2_8}
\end{figure}

\subsection{Experiment 3: Low Trust Effect}

In the low trust setting, we investigate the role of changing the accuracy of the algorithmic recommendation on judge decision-making under our proposed model. Our model dictates that with less accurate algorithmic recommendations, judge responsiveness decreases. We model this by using our model to derive hypothetical decisions made under simulated alternative prediction models. We simulate each alternative predictor with different algorithmic recommendation accuracy “boosts”, where an accuracy boost of 0 results in the default algorithmic recommendation accuracy of the PSA in the real world experimental setting. This default accuracy is calculated by compared recommended release decisions $D_i = 0$, with the observed outcomes. If the observed outcome is to recidivate and/or commit new crime (i.e. $Y_i = 1$), then we consider the prediction (ie. recommended decision) to be incorrect  with respect to observed outcomes. For the default PSA, we calculate an accuracy of $\approx 54\%$. 

Each accuracy boost occurs in increments of 0.1. If a boost is at 0.1, then 10\% of the algorithmic recommendations that were incorrect in the baseline setting are flipped to the correct decision with respect to the observed outcome. This simulates a setting where an alternative, higher accuracy model can be considered (eg.  with an accuracy boost of 0.1, predictive accuracy of these hypothetical alternative model becomes $\approx 67\%$).

Under various judge model settings, we can see that low accuracy algorithmic recommendations lead to an underestimation of the measured treatment effect on decisions and that this impact is also observed for decision correctness. We find, as with the other cases, that this disparity in measurement is most egregious for the non-white female (NWF) subgroup. 

Main results for this experiment can be found in in Figures ~\ref{fig:myfig_exp3_1}, ~\ref{fig:myfig_exp3_2}, ~\ref{fig:myfig_exp3_3}, ~\ref{fig:myfig_exp3_4}. 



\begin{figure}[htbp]
\centering
\subfloat[Judge Assignment: $J_1=0.6$, $J_2=0.3$, $J_3=0.1$]{\label{fig:a}\includegraphics[width=0.45\linewidth]{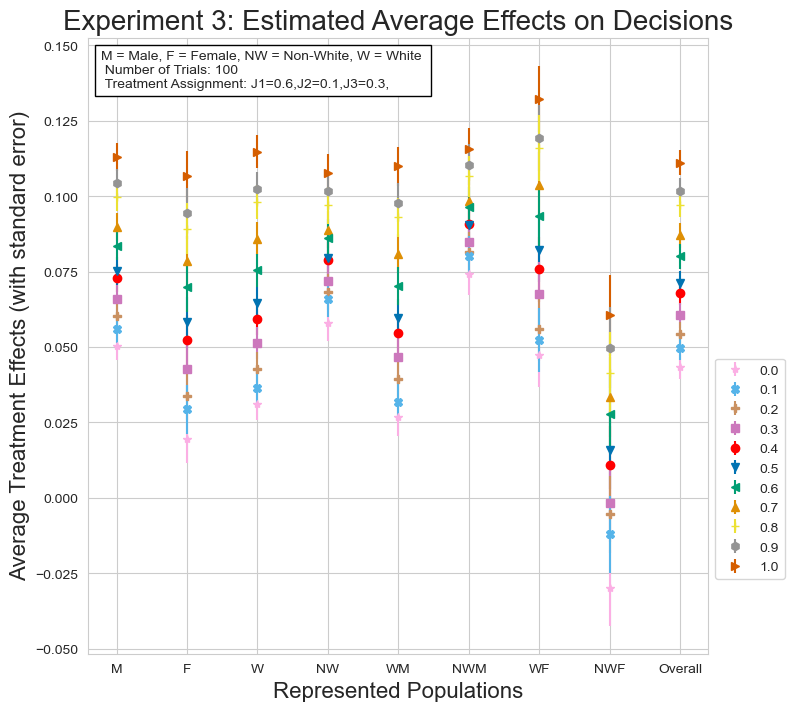}}\qquad
\subfloat[Judge Assignment: $J_1=0.5$, $J_2=0.5$]{\label{fig:b}\includegraphics[width=0.45\linewidth]{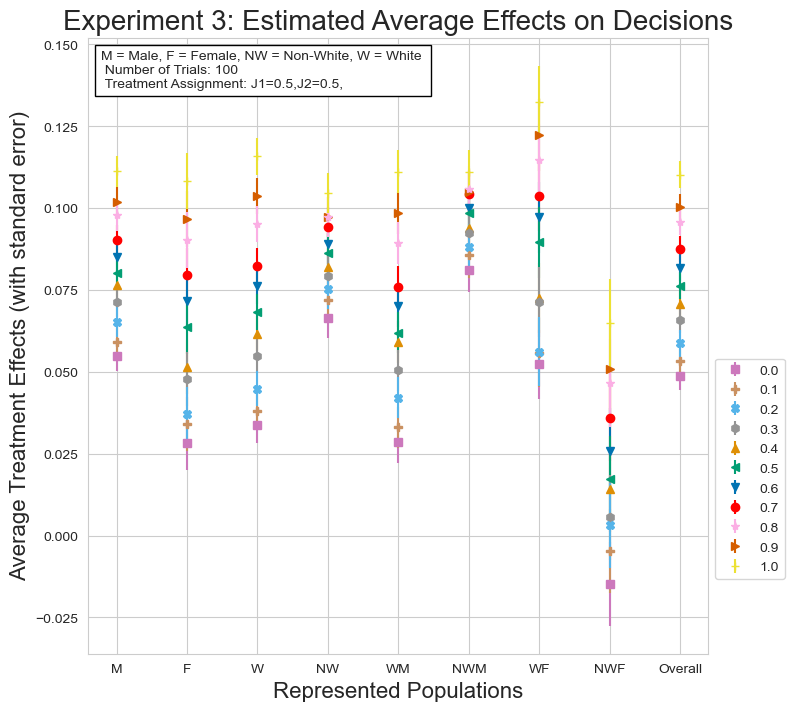}}\\
\subfloat[Judge Assignment: $J_1=0.33$, $J_2=0.33$, $J_3=0.33$]{\label{fig:c}\includegraphics[width=0.45\textwidth]{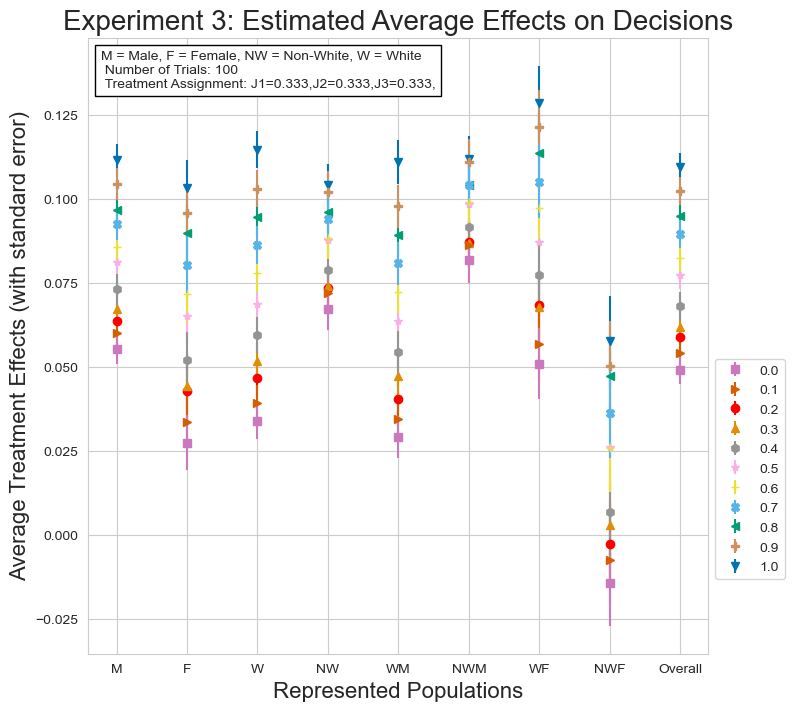}}\qquad%
\subfloat[Judge Assignment: $J_1=0.99$, $J_2=0.01$]{\label{fig:d}\includegraphics[width=0.45\textwidth]{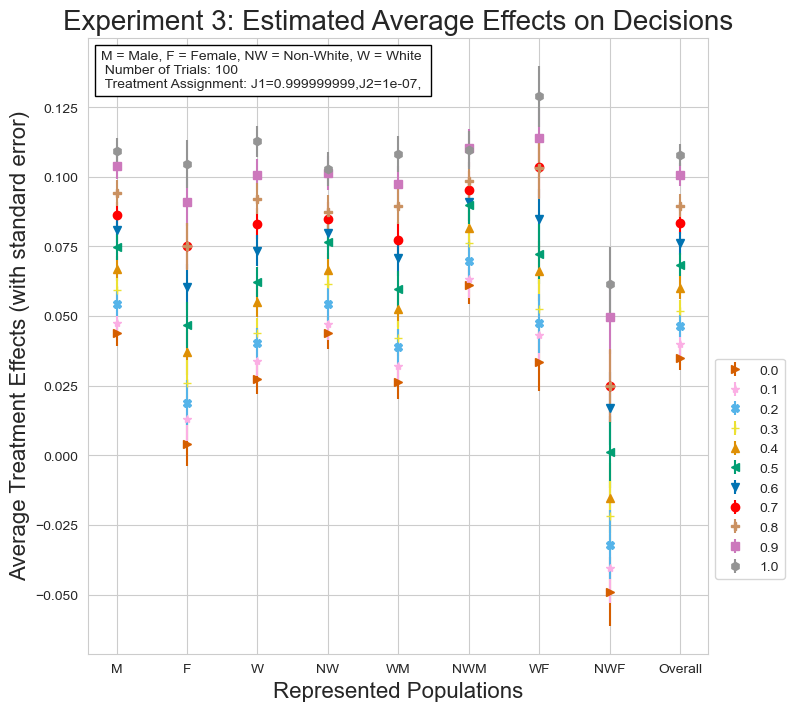}}%
\caption{\textbf{Experiment 3; Average Treatment Effect Changes on Decisions due to modifications in model accuracy.} Each increment of 0.1 in the legend represents an accuracy boost over the baseline observed accuracy of the studied PSA. For example, to simulate an accuracy boost of 0.1, we take 10\% of incorrect cases from the default PSA  (relative to observed outcomes) and we artificially flip the value of the model to match the observed outcome in order to artificially boost the accuracy of a simulated alternative prediction model. An accuracy boost of 1.0 is a perfect model that 100\% matches observed outcomes; an accuracy boost of 0 is the accuracy of the default PSA model -- Linear judge model, baseline judge bias $b_k =0.6$, error rate threshold $t=0.15$}
\label{fig:myfig_exp3_1}
\end{figure}

\begin{figure}[htbp]
\centering
\subfloat[Judge Assignment: $J_1=0.6$, $J_2=0.3$, $J_3=0.1$]{\label{fig:a}\includegraphics[width=0.45\linewidth]{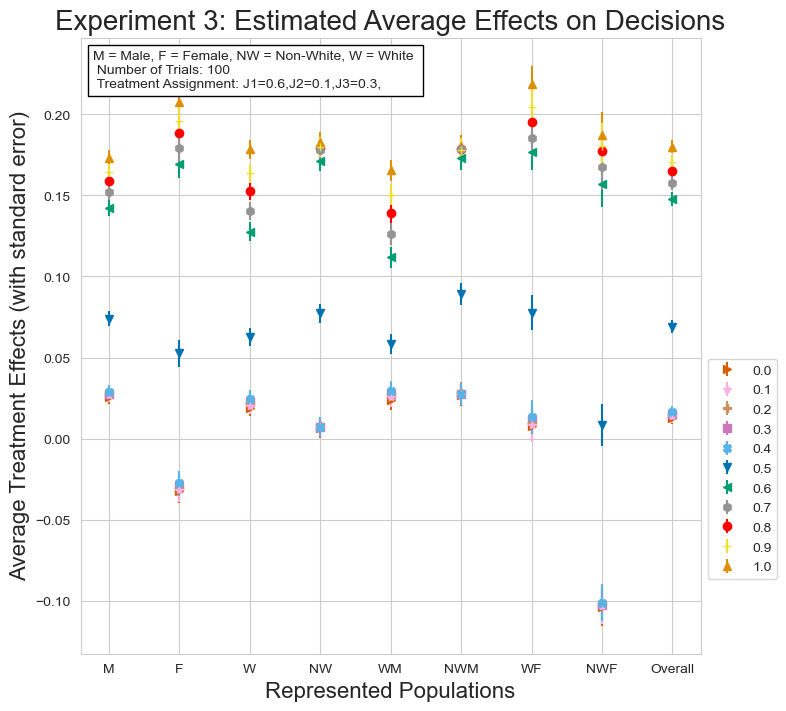}}\qquad
\subfloat[Judge Assignment: $J_1=0.5$, $J_2=0.5$]{\label{fig:b}\includegraphics[width=0.45\linewidth]{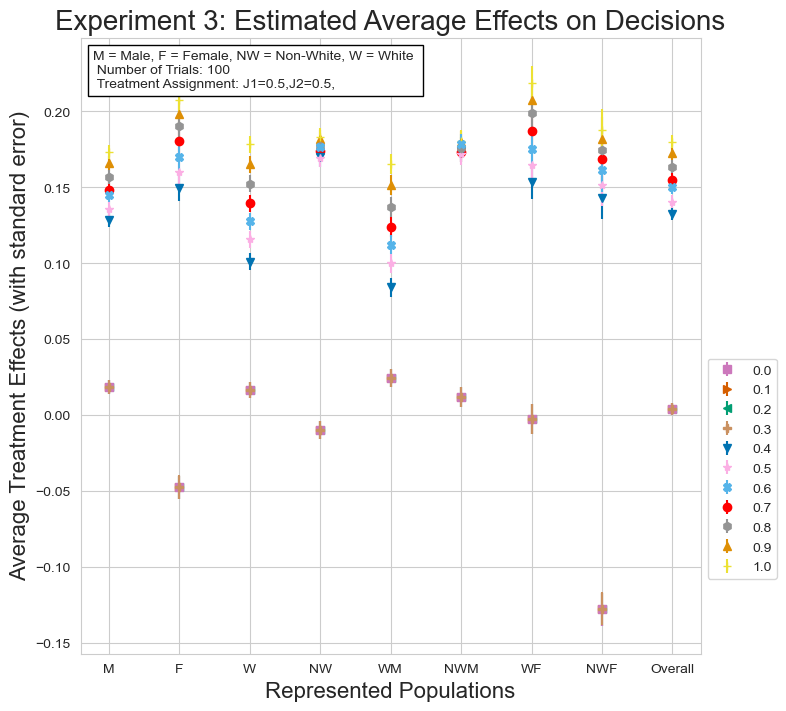}}\\
\subfloat[Judge Assignment: $J_1=0.33$, $J_2=0.33$, $J_3=0.33$]{\label{fig:c}\includegraphics[width=0.45\textwidth]{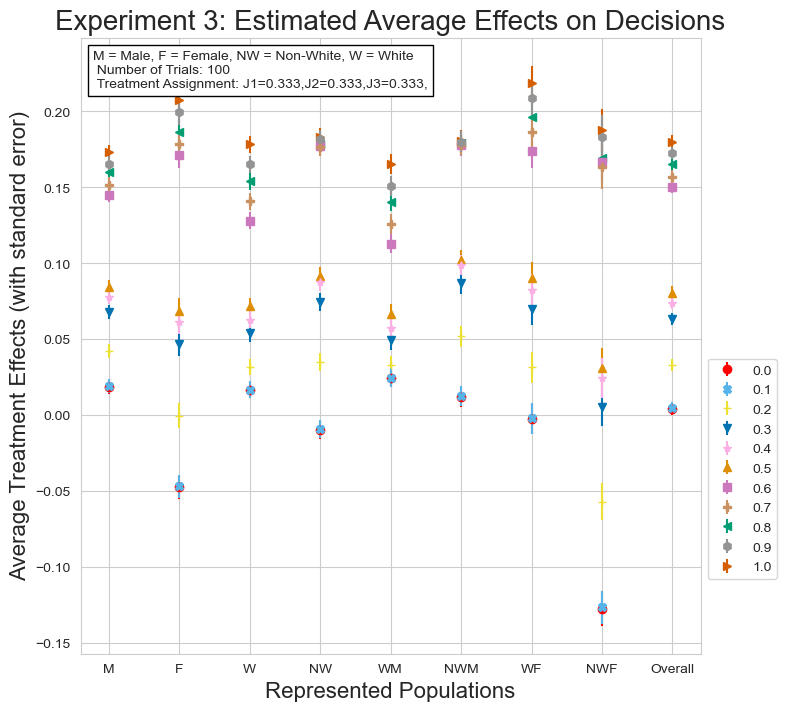}}\qquad%
\subfloat[Judge Assignment: $J_1=0.99$, $J_2=0.01$]{\label{fig:d}\includegraphics[width=0.45\textwidth]{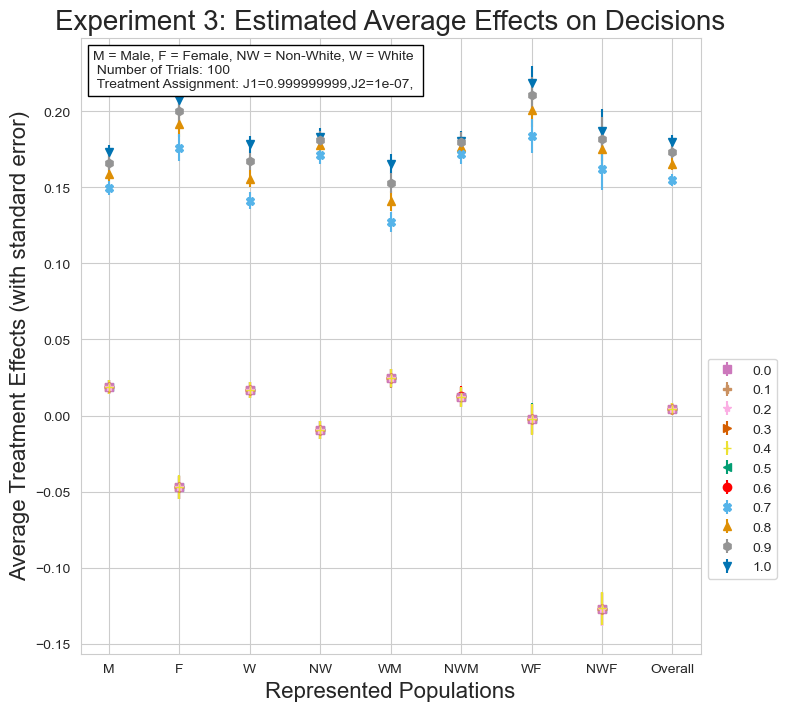}}%
\caption{\textbf{Experiment 3; Average Treatment Effect Changes on Decisions due to modifications in model accuracy.} Each increment of 0.1 in the legend represents an accuracy boost over the baseline observed accuracy of the studied PSA. For example, to simulate an accuracy boost of 0.1, we take 10\% of incorrect cases from the default PSA  (relative to observed outcomes) and we artificially flip the value of the model to match the observed outcome in order to artificially boost the accuracy of a simulated alternative prediction model. An accuracy boost of 1.0 is a perfect model that 100\% matches observed outcomes; an accuracy boost of 0 is the accuracy of the default PSA model -- Non-Linear judge model,  baseline judge bias $b_k =0.6$, error rate threshold $t=0.15$}
\label{fig:myfig_exp3_2}
\end{figure}

\begin{figure}[htbp]
\centering
\subfloat[Judge Assignment: $J_1=0.6$, $J_2=0.3$, $J_3=0.1$]{\label{fig:a}\includegraphics[width=0.45\linewidth]{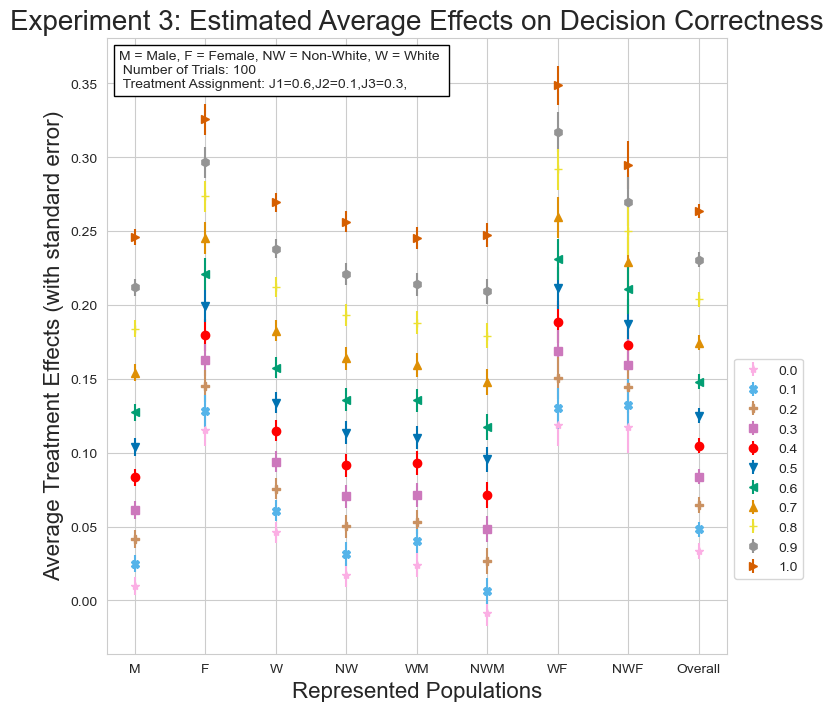}}\qquad
\subfloat[Judge Assignment: $J_1=0.5$, $J_2=0.5$]{\label{fig:b}\includegraphics[width=0.45\linewidth]{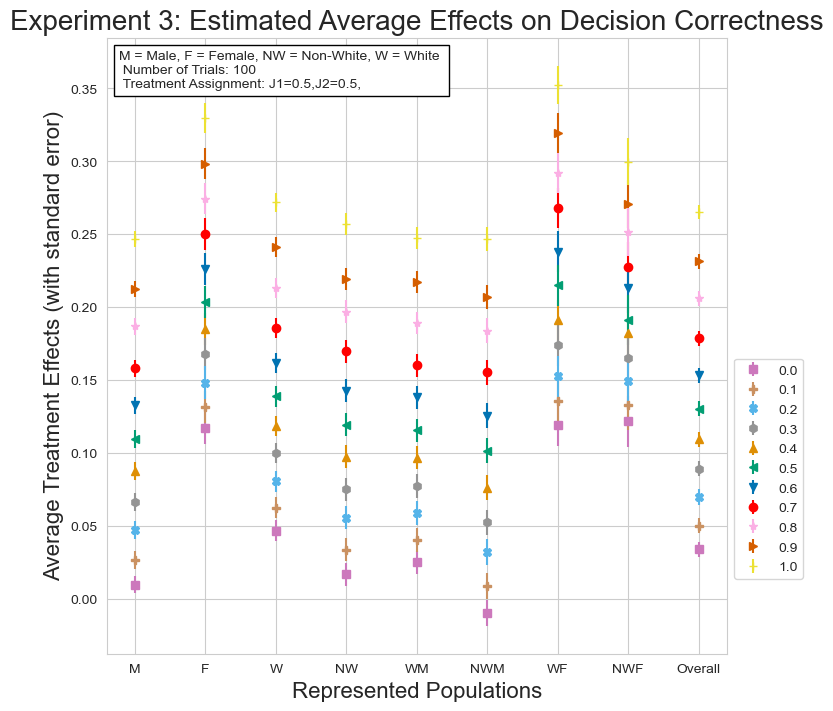}}\\
\subfloat[Judge Assignment: $J_1=0.33$, $J_2=0.33$, $J_3=0.33$]{\label{fig:c}\includegraphics[width=0.45\textwidth]{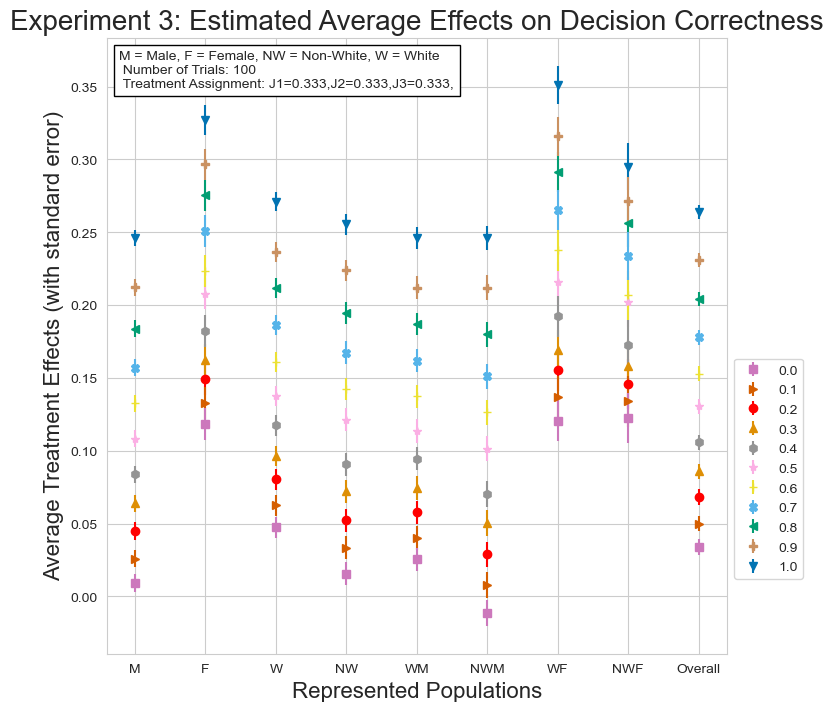}}\qquad%
\subfloat[Judge Assignment: $J_1=0.99$, $J_2=0.01$]{\label{fig:d}\includegraphics[width=0.45\textwidth]{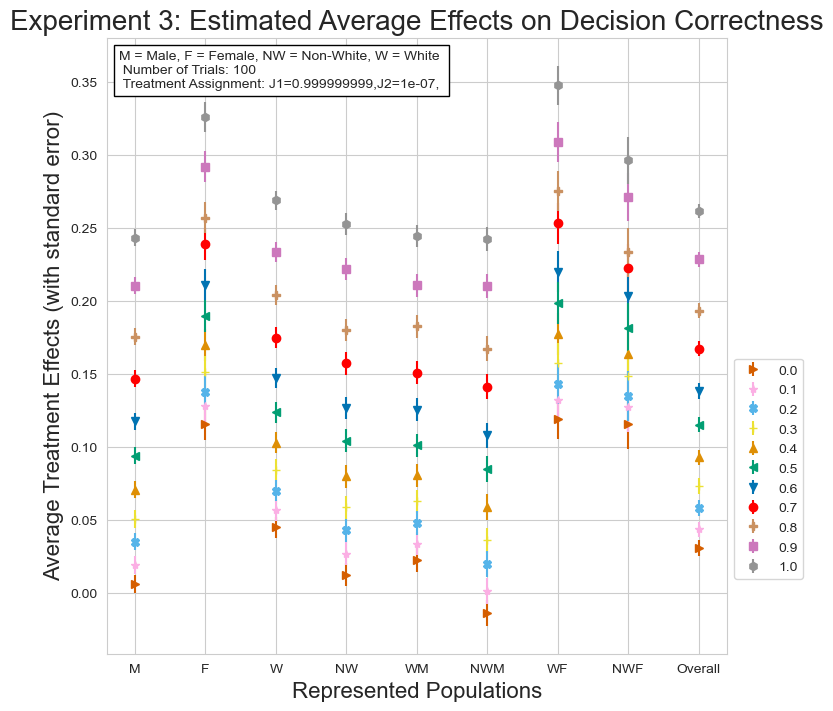}}%
\caption{\textbf{Experiment 3; Average Treatment Effect Changes on Decision Correctness due to modifications in model accuracy.} Each increment of 0.1 in the legend represents an accuracy boost over the baseline observed accuracy of the studied PSA. For example, to simulate an accuracy boost of 0.1, we take 10\% of incorrect cases from the default PSA  (relative to observed outcomes) and we artificially flip the value of the model to match the observed outcome in order to artificially boost the accuracy of a simulated alternative prediction model. An accuracy boost of 1.0 is a perfect model that 100\% matches observed outcomes; an accuracy boost of 0 is the accuracy of the default PSA model -- Linear judge model. $b_k =0.6$, error rate threshold $t=0.15$}
\label{fig:myfig_exp3_3}
\end{figure}

\begin{figure}[htbp]
\centering
\subfloat[Judge Assignment: $J_1=0.6$, $J_2=0.3$, $J_3=0.1$]{\label{fig:a}\includegraphics[width=0.45\linewidth]{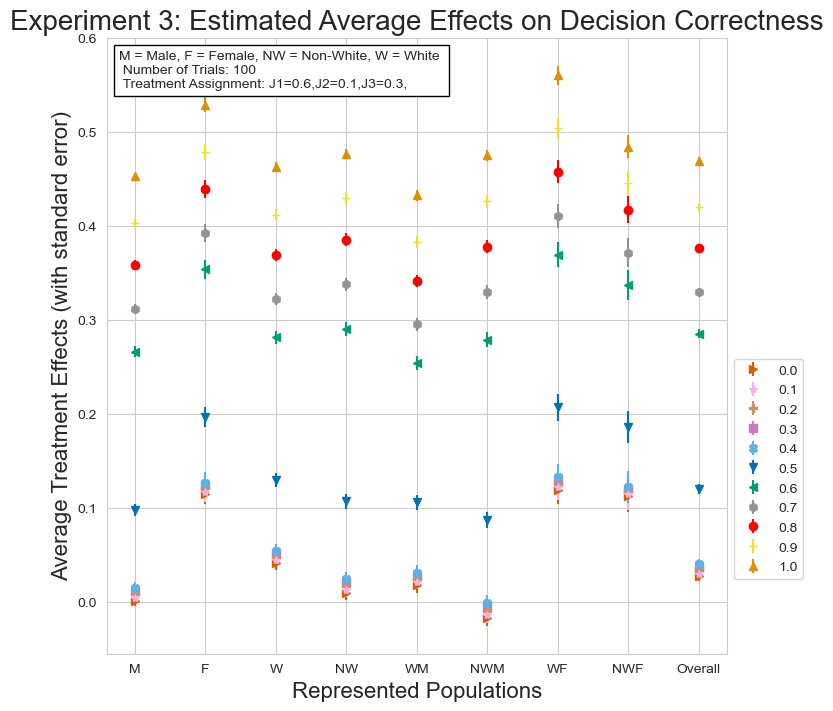}}\qquad
\subfloat[Judge Assignment: $J_1=0.5$, $J_2=0.5$]{\label{fig:b}\includegraphics[width=0.45\linewidth]{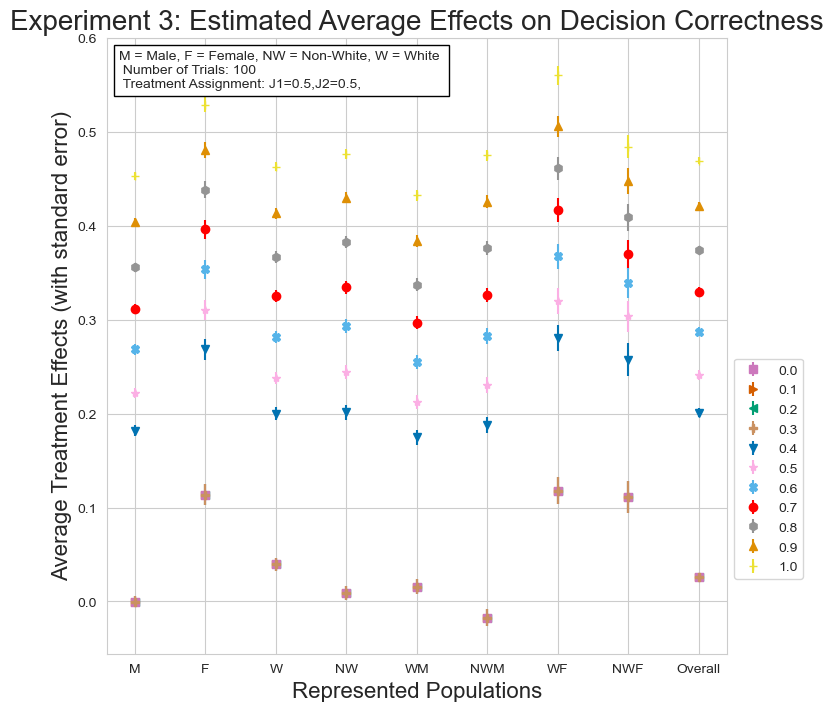}}\\
\subfloat[Judge Assignment: $J_1=0.33$, $J_2=0.33$, $J_3=0.33$]{\label{fig:c}\includegraphics[width=0.45\textwidth]{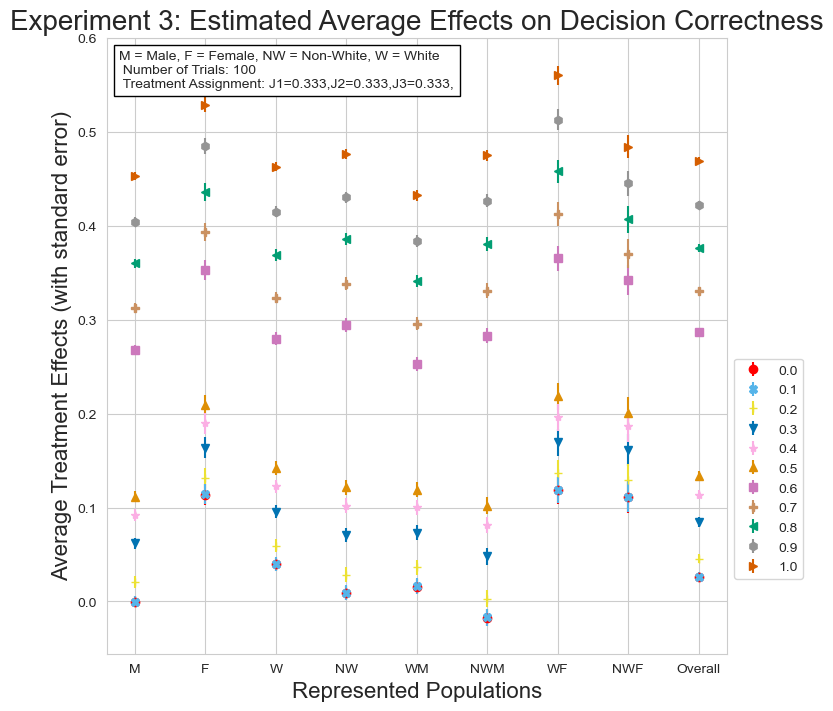}}\qquad%
\subfloat[Judge Assignment: $J_1=0.99$, $J_2=0.01$]{\label{fig:d}\includegraphics[width=0.45\textwidth]{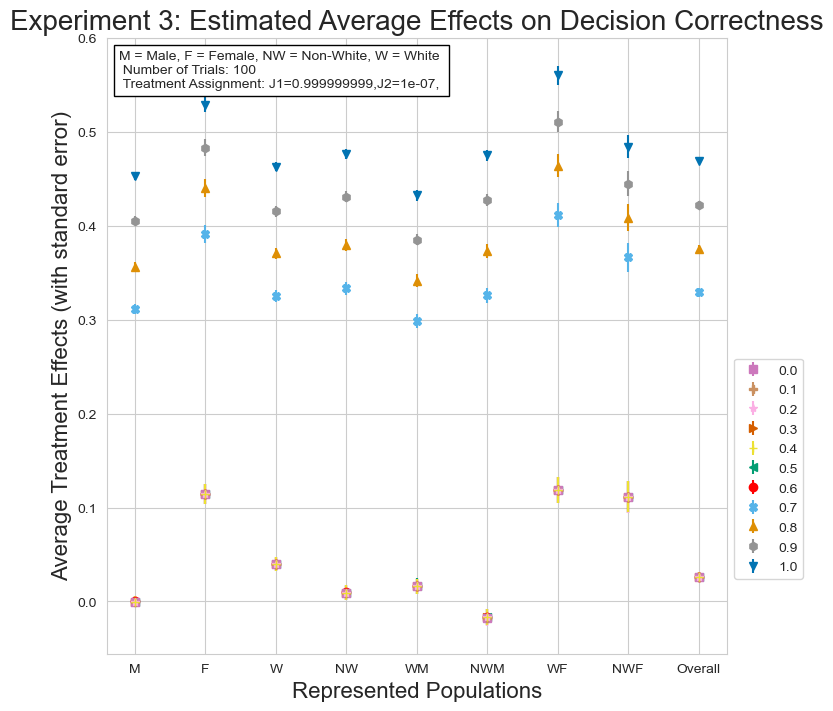}}%
\caption{\textbf{Experiment 3, Average Treatment Effect Changes on Decision Correctness due to modifications in model accuracy.} Each increment of 0.1 in the legend represents an accuracy boost over the baseline observed accuracy of the studied PSA. For example, to simulate an accuracy boost of 0.1, we take 10\% of incorrect cases from the default PSA  (relative to observed outcomes) and we artificially flip the value of the model to match the observed outcome in order to artificially boost the accuracy of a simulated alternative prediction model. An accuracy boost of 1.0 is a perfect model that 100\% matches observed outcomes; an accuracy boost of 0 is the accuracy of the default PSA model -- Non-Linear judge model, baseline judge bias $b_k =0.6$, error rate threshold $t=0.15$}
\label{fig:myfig_exp3_4}
\end{figure}

\subsection{Additional Observations}
In this section, we provide a brief overview of some additional observations that informed the design and implementation of our study. 

\paragraph{No Empirical Modification of Outcomes}
In an ideal world, we could extend this work to observe how the changes in decisions could ultimately impact outcomes. However, because we could not modify outcomes through an active experiment, there was no empirical modification of outcomes in our semi-synthetic empirical explorations. 

\paragraph{Limited Impact on Decision Correctness}
Not only do we calculate the average treatment effects on decisions, but we also calculate the average treatment effects on ~\emph{decision correctness} -- which we define as instances where the decision for release $D_i = 0$ does not result in an outcome of recidivism, or new crime $Y_i = 1$. Note that we cannot observe the correctness for scenarios for which the original decision was to detain $D_i = 1$ , though the default outcome for such decisions is set to  $Y_i = 0$ in the dataset. 

Empirically, our simulations demonstrate that decision correctness very rarely gets impacted by the judge responsiveness factors and related experimental design choices, with changes to the measured treatment effect typically about 10x smaller than that observed for the effects on the actual decisions (see Figures ~\ref{fig:myfig_ex1_3}, ~\ref{fig:myfig_ex1_4}). This indicates potential further work into what kinds of experiment design choices might impact decision correctness more directly -- notably the largest observed ATE changes on decision correctness are related to Experiment 3 (see Figures ~\ref{fig:myfig_exp3_3}, ~\ref{fig:myfig_exp3_4}), which impacts prediction correctness in addition to it's secondary impact via judge compliance.

\paragraph{Statistical Significance Tests}
In the codebase, we provide the code to calculate the statistical significance of observed differences between estimated treatment effect measurements under the various responsiveness factors. This can be used to further validate the empirical observation that in many cases, experiment design choices do have a notable and statistically significant impact on treatment effect measurement. Further details on this analysis can be found in \href{https://anonymous.4open.science/r/exp-design-human-decisions-A3CD/README.md}{our anonymous repository} of supplementary material.

\paragraph{Day-to-Day Decision-making Patterns}

 The original study ~\citep{imai2020experimental} included 1891  cases presented to one judge over 274 days.  In an attempt to better understand the dynamics of decision-making in the original study, we did an initial empirical exploration into compliance patterns given treatment exposure on a given day. What we find is that on most days, the judge has an exposure of $\approx 50\%$ (specifically, a mean treatment exposure of 49.5\% across all days). Overall, the mean agreement rate of the judge with the recommended model is fairly high (71.2\%, over all days) and weakly positively correlated with the treatment exposure rate ($\rho = 0.1095$). However, it seems that the average daily agreement rate of untreated cases (65.9\%) is notably lower that the average daily agreement rate of treated cases (70.1\%), and the latter is also more correlated with the treatment exposure rate ($\rho = -0.1768$ for untreated cases, and $\rho = 0.2562$ for treated cases, a correlation that also holds for next day  with $\rho =0.2414$). Although far from conclusive evidence, this could be interpreted to suggest that there is some  positive correlation between the treatment exposure rate and judge compliance on treated cases, as we hypothesize in our work. 

An overview of such finding can be found in Figure ~\ref{fig:myfig_add}.

\begin{figure}[htbp]
\centering
\subfloat[Rate of treated cases over total assigned casses for a given day.]{\label{fig:a}\includegraphics[width=0.45\linewidth]{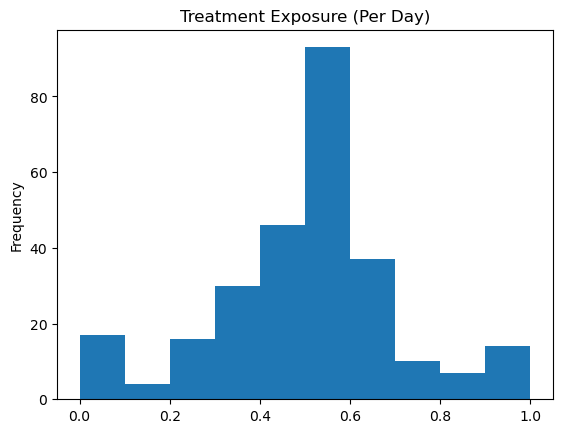}}\qquad
\subfloat[Overall rate of agreement between judge decision-making and model recommendation per day.]{\label{fig:b}\includegraphics[width=0.45\linewidth]{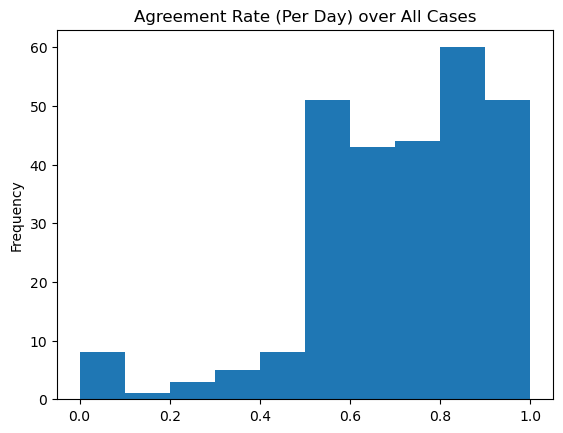}}\\
\subfloat[Rate of agreement between judge decision-making and model recommendation per day, for treated cases.]{\label{fig:c}\includegraphics[width=0.45\textwidth]{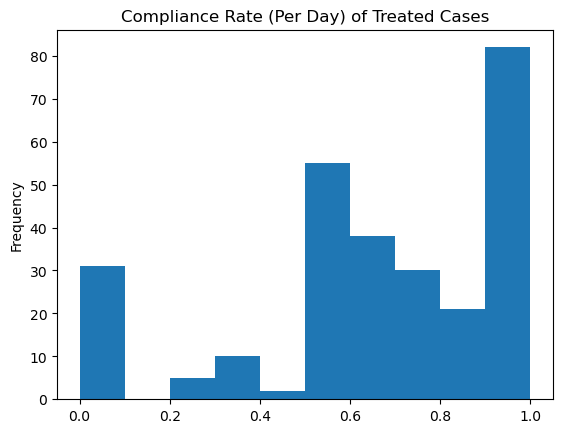}}\qquad%
\subfloat[Rate of agreement between judge decision-making and model recommendation per day, for untreated cases.]{\label{fig:d}\includegraphics[width=0.45\textwidth]{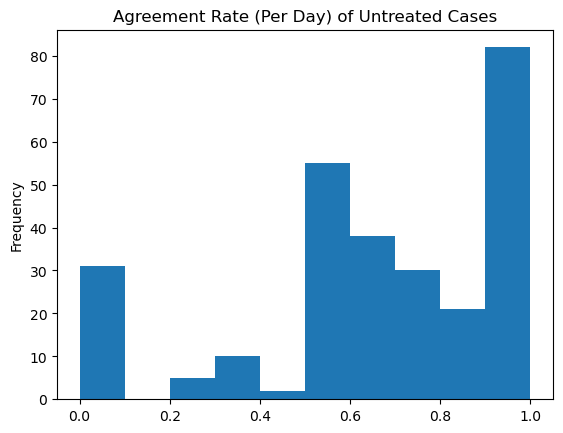}}%
\caption{Additional observations on judge behavior over individual days.}
\label{fig:myfig_add}
\end{figure}

\section{Further Related work}\label{app:further_rel}

\paragraph{Human Interaction Lab Studies}
In contrast to RCTs in the field, there has been work on controlled~\emph{lab} experiments in which experimenters assess the effect of providing algorithmic risk scores in simulated environments to instructed test subjects \citep{green2019disparate}. The test subjects are non-experts and are often role-playing real-world decision-makers on a contrived or real-world inspired task~\citep{lai2023towards}. In the current work, we focus on the setting of a field experiment as it more directly addresses the impact of ADS in deployment.

\paragraph{Spillover effects}
Recently, \citet{RCT_Service_Interventions} studied experiment design where the treatment is delivered by a human service provider who has limited resources. They show that treatment effect sizes are mediated by such capacity constraints, which are \emph{induced} by particular choices in the design of the experiment, such as the number of service providers recruited and the treatment assignment. Our work also views the human decision maker as a mediator of the treatment effect of using algorithmic decision aids in a decision process; however, we are motivated by modelling the human decision maker's cognitive biases induced by the experiment design, rather than their capacity constraints alone. Decision making is also a distinct setting from service provision.

\subsection{Academic Evidence for responsiveness factors}
\label{app:further_rel-evidence}
It has long been understood in social psychology and behavioral economics that the ``framing’’ or context in which nudges or advice is given, meaningfully impacts how much this additional information is heeded during decision-making (see \citep{1kahneman2011thinking, 4klein2017sources, 8tversky1989rational}). This discussion has been further elaborated on in computer science contexts, through the formalization of broader ``algorithm-in-the-loop’’ considerations ~\citep{10green2019principles, 9guo2024decision}, and notions of ``automation bias’’ (ie. the degree of reliance of a human decision-maker on an algorithmic recommendation) ~\citep{14goddard2012automation, 15albright2019if}. 

Specific responsiveness factors are formally supported with the following long-standing evidence: 

\paragraph{Treatment exposure effect}
Humans tend to have a consistency bias ~\citep{1kahneman2011thinking, 2tversky1981framing, 3beresford2008understanding}, especially in high pressure and rapid decision-making contexts ~\citep{4klein2017sources}. This leads them to tend to rely on consistently available information, rather than adjusting their decision-making protocol for each case. This leads to different degrees of responsiveness once the decision-making is more often exposed to a particular stimuli or nudge, versus less responsiveness when they are overall less exposed ~\citep{2tversky1981framing, 4klein2017sources, 8tversky1989rational}. 

\paragraph{Capacity constraint}
Once operating at a particular capacity, especially in high stakes settings, decision-makers become less responsive to input information and nudges, as they are no longer in a position to act on the algorithmic recommendation past a certain threshold of capacity to respond. This is a finding that has been observed pragmatically in many settings by decision theory scholars ~\citep{4klein2017sources}, and economists ~\citep{5boutilier2024randomized}. The evidence is somewhat mixed however on the broader implications for generally resource constrained settings, where it has been observed by computer scientists that decision-makers may opt in those cases to over-rely on algorithmic recommendations, rather than invest in more complex, independent decision-making ~\citep{13buccinca2021trust, 14goddard2012automation}. This is still an area of ongoing discussion and study. 

\paragraph{Low trust model} 
Scholars in human-computer interface (HCI) literature have for some time flagged the relationship between model accuracy and human trust in AI recommendations in decision-making. Specifically, ~\cite{6yin2019understanding} highlights how “people's trust in a model is affected by both its stated accuracy and its observed accuracy, and that the effect of stated accuracy can change depending on the observed accuracy”. The authors conclude that human responsiveness to algorithmic recommendations is lowered by a perception of lower accuracy, a finding replicated empirically in much of the follow-up work in the field ~\citep{7lai2023towards, 12zhang2020effect}.

\subsection{Experimental Design attempts across domains}
\label{app:further_rel-domains}
Field experiments for automated decision systems (ADS) are still a nascent effort -- and not many evaluations are currently completely at a high quality. In fact, we focus so much on Imai et al., 2020 (See Appendix E on the Experiment Case Study) specifically because, as the authors acknowledge, this is one of the first thorough experimental evaluations available of an algorithmic intervention, which is why we rely on this study’s data for our empirical investigation. 

That being said, we can definitely provide more examples of this phenomenon across domains of application, in order to demonstrate how widespread some of the practices we discuss have become beyond the criminal justice space. Here, we provide some examples across various domains. These are experiments completed by a range of involved stakeholders from institutional users (eg. school, hospital, government users), policy program evaluation professionals, and academics across a range of disciplines, including computer science, economics and quantitative social science. 

Healthcare: 

There have been several registered randomized control trials for algorithmic deployment in the medical setting, mainly executed by academics ~\citep{B4plana2022randomized} but also hospitals often publish such studies as part of published pilots for internal use ~\citep{B1strickland2019ibm} and investigation ~\citep{B8borrella2019predict}. Notably, the methodology is very rarely tailored to the context of an algorithmic intervention, which is still an evolving practice ~\citep{B5altman2000we, B6wu2021medical}. For example,  a survey of 41 RCTs identified that the vast majority of the completed assessments (93\%) do not analyze model performance errors, and focus completely on clinical or measured policy outcomes ~\citep{B4plana2022randomized}.  The few studies that do factor this context into their design are often unable to resource the required control to execute an experiment (rather than relying on observational data ~\citep{B2wong2021external, B3yala2022multi}).  In the US specifically, the Food and Drug Administration (FDA) has struggled to properly adapt its audit and RCT regulatory requirements to the novel features of ADS technology, which it perceives as rapidly changing and emergent ~\citep{B7ross2022epic}. 

Education: 

One other domain in which there has been much interest in experimental evaluation for algorithmic interventions is education, where automated decision-making is increasingly popular ~\citep{B11perdomo2023difficult}[11]. Some popular experimental evaluations have been done in this context ~\citep{B9dawson2017prediction, B10liu2022lost}[9, 10], but similar to healthcare, methods are still under-developed and unspecialized. More specifically, randomization tends to happen at the student or class level, even when interventions are oriented around teacher or counselor decision-making ~\citep{B11perdomo2023difficult}[11], and, in many cases, interventions tend to be small or not be as effective in impacting downstream monitored outcomes ~\citep{B12kleinberg2018human} [12].

Criminal Justice:

Given the widespread use of risk assessments in the criminal justice space in the US (already impacting over 60 million defendants a year), this context has been the focus of much of the current work on intervention-based evaluations of automated decision systems ~\citep{B12kleinberg2018human, B13albright2019if, B14grgic2019human, B15engel2021machine, B16green2019principles, B17imai2023experimental, B18redcross2019evaluation, B19charette2018michigan}. Many of the existing lab studies ~\citep{B15engel2021machine, B16green2019principles, B14grgic2019human} and observational studies done so far on the topic  are from this context, but, as mentioned before, most notably, one of very few ~\emph{experimental} evaluations in this context was completed in the criminal justice context ~\citep{imai2020experimental}.

Other: 

In addition to the context mentioned, we note some intervention-based evaluation techniques in use in social services ~\citep{B22wilder2021clinical}, hiring ~\citep{B21hoffman2018discretion} and more. 


\end{document}